\documentclass[conference,compsoc, 10pt]{IEEEtran}

\usepackage{booktabs}   
\usepackage{subcaption} 
\usepackage{color}
\usepackage{amsmath}
\usepackage{amssymb}
\usepackage{graphicx}
\usepackage{amsthm}
\usepackage{stmaryrd}
\usepackage[all]{xy}
\usepackage{multirow}
\usepackage[linesnumbered,ruled,vlined]{algorithm2e}
\SetKwInOut{Input}{Input}
\usepackage[utf8]{inputenc}
\usepackage{tikz}
\usepackage{url}
\usepackage{enumitem}

\allowdisplaybreaks[4]

\newtheorem{theorem}{Theorem}[section]

\newtheorem{lemma}{Lemma}[section]
\newtheorem{remark}{Remark}[section]
\newtheorem{definition}{Definition}[section]
\newtheorem{proposition}{Proposition}[section]
\newtheorem{example}{Example}[section]
\newtheorem{problem}{Problem}[section]

\newtheorem{fact}{Fact}
\newtheorem{claim}{Claim}

\newcommand {\qbar} {{\overline{q}}}
\newcommand {\qI} {{q:=|0\rangle}}
\newcommand {\qU} {{\overline{q}:=U[\overline{q}]}}
\newcommand {\cD } {{\mathcal{D}}}

\newcommand {\cP } {{\mathcal{P}}}
\newcommand {\cH } {{\mathcal{H}}}
\newcommand {\cM } {{\mathcal{M}}}
\newcommand {\cV } {{\mathcal{V}}}
\newcommand {\cU } {{\mathcal{U}}}
\newcommand {\cX } {{\mathcal{X}}}
\newcommand {\cE } {{\mathcal{E}}}
\newcommand {\cF } {{\mathcal{F}}}
\newcommand {\cI } {{\mathcal{I}}}

\newcommand {\bi } {{\boldsymbol{i}}}
\newcommand {\id } {{I}}
\newcommand {\bD} {{\mathbf{D}}}

\newcommand{\Mexist}{{\downarrow}}

\newcommand {\Ldoms } {{\mathsf{dom}}}
\newcommand {\Ldom }[1] {{\mathsf{dom}\!\left(#1\right)}}

\newcommand {\Lfree }[1] {{\mathsf{free}{\left(#1\right)}}}

\newcommand {\dom }[1] {{\mathsf{dom}\!\left(#1\right)}}

\newcommand {\free }[1] {{\mathsf{free}\left(#1\right)}}
\newcommand {\ptype} {{\mathrm{type}}}
\newcommand {\res} {{\mathrm{Res}}}
\newcommand {\LTypeEs } {{\mathsf{dom}}}
\newcommand {\LTypeE }[1] {{\mathsf{dom}\!\left(#1\right)}}
\newcommand {\LTypeFs } {{\mathsf{free}}}
\newcommand {\LTypeF }[1] {{\mathsf{free}{\left(#1\right)}}}
\newcommand {\types } {{\mathsf{dom}}}
\newcommand {\type }[1] {{\mathsf{dom}\!\left(#1\right)}}

\newcommand {\typef }[1] {{\mathsf{free}\left(#1\right)}}
\newcommand {\emptyprog} {{\mathbf{E}}}

\newcommand {\rt }[2] {{\left.{#1}\right|_{#2}}}
\newcommand {\tr } {{\mathrm{tr}}}
\newcommand {\vars } {\mathbf{V}}
\newcommand {\V }[1] {{\mathsf{free}{\left(#1\right)}}}
\newcommand {\Vs } {{\mathsf{free}}}
\newcommand {\var } {\mathsf{var}}

\newcommand {\unia } {{\mathbf{U}}}

\newcommand {\prog } {{\mathbf{C}}}

\newcommand {\sepimp} {\mathrel{-\mkern-6mu*}}

\newcommand {\sem}[1] {\llbracket#1\rrbracket}

\newcommand {\AP} {{\mathcal{AP}}}
\newcommand {\spa } {{\rm span}}
\newcommand {\supp } {{\mathrm{supp}}}

\newcommand {\ol}[1] {{\overline{#1}}}

\newcommand{\sd}{\diamond}%
\newcommand {\sdimp} {\mathrel{-\mkern-2.5mu\diamond}}

\newcommand {\qmimp} {\mathrel{-\mkern-6mu?}}
\def\>{\ensuremath{\rangle}}
\def\<{\ensuremath{\langle}}
\newcommand {\swap} {\mathrm{SWAP}}
\newcommand {\perm} {\mathbf{Perm}}
\newcommand{\qvee}{\sqcup}
\newcommand{\qwedge}{\sqcap}
\newcommand{\qbigvee}{\bigsqcup}

\newcommand{\proj}{{\mathrm{proj}}}

\DeclareRobustCommand{\gimp}{%
	\mathbin{\ooalign{$\rightarrow$\cr\hss\raisebox{1ex}{\scalebox{.5}{G}}\hss}}}

\DeclareRobustCommand{\gimplr}{%
	\mathbin{\ooalign{$\leftrightarrow$\cr\hss\raisebox{1ex}{\scalebox{.5}{G}}\hss}}}

\ifCLASSOPTIONcompsoc
\usepackage[nocompress]{cite}
\else
\usepackage{cite}
\fi

\ifCLASSINFOpdf

\else

\fi
\makeatletter
\let\NAT@parse\undefined
\makeatother
\usepackage[hidelinks]{hyperref} 

\begin{document}

	\title{A Quantum Interpretation of Bunched Logic for Quantum Separation Logic}
	
	\author{%
		\IEEEauthorblockN{Li Zhou}
		\IEEEauthorblockA{Max Planck Institute for Security and Privacy}
		\and
		\IEEEauthorblockN{Gilles Barthe}
		\IEEEauthorblockA{\qquad \quad Max Planck Institute for Security and Privacy;\qquad\quad\\
			IMDEA Software Institute}
		\and
		\IEEEauthorblockN{\quad\ \  Justin Hsu }
		\IEEEauthorblockA{\qquad\qquad\quad\ \  University of Wisconsin--Madison \qquad\qquad}
		\and
		\IEEEauthorblockN{\ \, Mingsheng Ying}
		\IEEEauthorblockA{\ \, University of Technology Sydney;\\
			  \qquad\ \, Institute of Software, Chinese Academy of Sciences;\qquad \\
			\ \, Tsinghua University}
		\and
		\IEEEauthorblockN{\ Nengkun Yu}
		\IEEEauthorblockA{\qquad\qquad\quad\ \ \ University of Technology Sydney\qquad\qquad\quad\ \  }
	}

	\maketitle	
	
	\thispagestyle{plain}
	\pagestyle{plain}
	
	\begin{abstract}
    We propose a model of the substructural logic of Bunched Implications (BI) that is suitable for reasoning about quantum states. In our model, the separating conjunction of BI describes separable quantum states. We develop a program logic where pre- and post-conditions are BI formulas describing quantum states---the program logic can be seen as a counterpart of separation logic for imperative quantum programs. We exercise the logic for proving the security of quantum one-time pad and secret sharing, and we show how the program logic can be used to discover a flaw in Google Cirq’s tutorial on the Variational Quantum Algorithm (VQA).
	\end{abstract}

	\section{Introduction}
  The logic of Bunched Implications (BI) of O'Hearn and Pym~\cite{OP99,Pym02,POY04}
        is a substructural logic that features resource-aware
        connectives. One such connective is $*$, known as separating
        conjunction: informally, an assertion $\phi * \psi$ holds with
        respect to a resource $R$ if the resource $R$ can be split
        into resources $R'$ and $R''$ such that $\phi$ holds with
        respect to $R'$ and $\psi$ holds with respect to $R''$. This
        interpretation is particularly well suited for reasoning about
        programs in settings where computations can have interfering
        effects. In particular, BI has found success as an assertion
        language for Separation Logic~\cite{Rey02,ORY01,IO01}, a
        program logic for reasoning about programs with mutable
        state, and Concurrent Separation Logic~\cite{OHe07,Bro07}, a program
        logic for reasoning about shared-memory concurrent processes.

        Recent works seek to extend the separation logic framework beyond
        memory-manipulating programs by considering other notions of resources
        and other models of computation. Broadly speaking, separation logics are
        a good fit whenever programs manipulate resources in a \emph{local}
        fashion: that is, there is a natural notion of two resources being
        \emph{separate}, and a program can operate on the first resource without
        affecting the second. This idea underlies recent separation logics for
        probabilistic programs, where separation is probabilistic
        independence~\cite{BHL19}.

        Quantum computation is another domain where the ideas of separation
        logic seem relevant. Recent work~\cite{CoeckeFS16,Doc19} suggests that
        reasoning about resources (in particular, entanglement -- a resource
        unique in the quantum world) can bring similar benefits to quantum
        computing and communications. Motivated by this broad perspective, we
        propose a quantum model of BI and develop a novel separation logic for
        quantum programs. Our development is guided by concrete examples of
        quantum algorithms and security protocols.


	
	
	{\vskip 3pt}
	
	\textbf{Motivating Local Reasoning for Quantum Programs}:
        Quantum Machine Learning \cite{BWP17, BVM20} and VQAs
(Variational Quantum Algorithms) \cite{PMS14,MRB16} are new classes of quantum algorithms that emerged in recent years as a leading application of quantum computing. These algorithms solve problems by training parameterized quantum circuits. The trained circuits are usually very large
        in terms of both their size and
        the required quantum resources, i.e., the number of involved quantum bits   (qubits). 
        This makes them particularly challenging to verify
        with existing techniques such as quantum Hoare logic~\cite{Ying11,ZYY19} and verification based on operational semantics~\cite{HRS20}, since the dimension of  the matrices used to represent assertions   increases
        exponentially w.r.t.\, the number of qubits.
        Fortunately, these algorithms can benefit from local reasoning,  since each of their operations is performed locally on a  small number of qubits. Consider for instance the quantum circuit shown in
        Figure \ref{fig VQA2}, that implements a VQA circuit with
        $2 \times 2$ grid qubits. Instead of reasoning about the
circuit as a whole, we would like to reason about sub-circuits
        ${\rm ProcC}(1), {\rm ProcC}(2), {\rm Proc(R)}(1), {\rm
          ProcR}(2)$ separately, and then combine the results 
        to establish the correctness of the whole program. This is
        precisely the kind of reasoning enabled by Quantum Separation Logic (QSL for short).

	
	
	\begin{figure}
	\includegraphics[width=\linewidth]{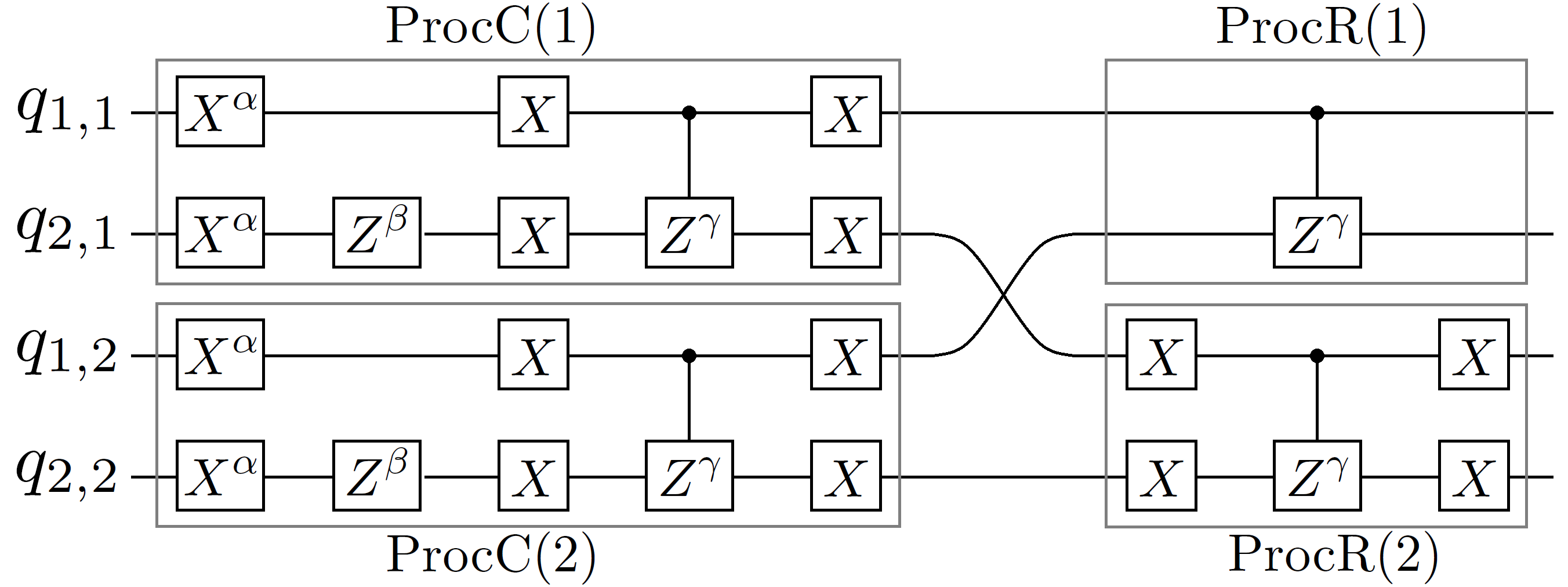}
		\caption{VQA(2) with parameters taken in Sec. \ref{sec: VQA in QSL}.}
		\label{fig VQA2}
	\end{figure}
	
       {\vskip 3pt}
	
	\textbf{Technical Challenges and Contributions}:
  QSL will be developed by first developing a model of BI, where formulas
  describe quantum states and then building a separation logic using these
  assertions as pre- and post-conditions, introducing proof rules to
  reason about quantum programs.
	
	\textbf{\emph{BI and Its Quantum Interpretation}}: To characterize the properties of quantum systems, we first identify a quantum interpretation of BI appropriate for our target applications. 
  We choose to interpret our separating conjunction $\ast$ as separability
  of quantum states. Roughly speaking, $\phi_1\ast\phi_2$ holds in a quantum state $\rho$ if $\rho$ can be factored into two quantum states $\rho_1$ and $\rho_2$ over disjoint registers satisfying $\phi_1$ and $\phi_2$, respectively. 

  \textbf{\emph{Proof System for Program Logic}}: next, we define a
  program logic for a quantum \textbf{while}-language~\cite{Ying11}
  (for simplicity, we do not consider classical variables). Our
  language follows the ``classical control and quantum data'' 
  paradigm. We develop a set of proof rules that are effective for
  verifying quantum programs over a large set of qubits. Our proof
  system has several novel ingredients:
	\begin{enumerate}
		
        \item \textbf{Modification} on BI formulas.  The basic rule
          for assignments in classical program logics is defined using
          the syntactic notion of substitution. Due to the non-cloning
          law of quantum information, the role of assignments has to
          be played by initialization $q:=|0\>$ and unitary
          transformations $\qbar := U[\qbar]$, and inference rules for
          these operations involve a quantum operation (e.g.,
          \cite{Ying11}, \cite{ZYY19}). Unfortunately, the rules for
          initialization and unitary transformations are not simple
          adaptations of the rule for assignment, because a quantum
          generalization of substitution is not straightforward. For
          atomic predicates, substitutions are not always defined. For
          composite formulas, a straightforward definition of
          substitution is too weak for applications.
      %
      %
      We overcome this hurdle
      by introducing a \emph{modification} operation for atomic formulas (see
      Definition \ref{def sub atomic prop}), which is essentially a quantum version of substitution. Extending this operation to composite formulas requires some care (see Definition \ref{def sub 2-BI form}).    
		
    \item 
      \textbf{Frame rule}: The frame rule is one of the most characteristic
      structural rules in separation logic. QSL also enjoys a frame rule
      \textsc{Frame} that is similar in spirit to frame rules from standard
      separation logics, but our new interpretation of separating conjunction
      means that the meaning of the rule is different. Furthermore, the frame
      rule can be generalized slightly: even if the standard side condition for
      frame rules does not hold, the frame rule still applies if the
      post-condition is a \emph{supported} assertion---a concept first proposed
      by Reynolds~\cite{Rey08} in the context of standard separation logic. This
      extra bit of freedom seems to be particular to the quantum setting, and we
      crucially use this feature when using the frame rule to establish
      uniformity.
      %
      %
      The soundness proof of our quantum frame rule requires a nontrivial
      calculation based on \emph{purification}, a fundamental technique used in
      quantum information for transforming mixed states to pure states by
      introducing reference systems~\cite{NC00}.
      %
      %
		
	\item 
		Reasoning about \textbf{entangled predicates}:\footnote{Here, entangled predicates refer to the projections that cannot be factored as a product of projections of its local constituents.} The structural rules \textsc{Frame} and \textsc{Const} enable us to lift local reasoning to global correctness of quantum algorithms \textit{only} when no entanglement occur in the pre- and post-conditions.
      %
      %
      However, entangled predicates play an essential role in revealing the non-local (global) properties of a composite quantum system; for instance, some entangled predicates are used when reasoning about the (in)correctness of VQA (see Sec. \ref{sec: VQA in QSL}).
		With the help of auxiliary variables, we set up a new rule \textsc{UnCR} which enables us to prove the correctness of large quantum algorithms with respect to entangled pre- and post-conditions.
		Intuitively, when the program (as the principal system) combined with auxiliary variables (as ancillary systems), modification can be used to create (mathematically rather than physically) entanglement and rule \textsc{UnCR} is used to preserve correctness under the modification on the auxiliary variables in the pre- and post-conditions (but not in the program). The key idea behind was first proposed in \cite{YZL18} for reasoning about parallel quantum programs;  \textsc{UnCR} is its generalization tailored for our purpose.		
	\end{enumerate}
	
	\textbf{Applications}: To demonstrate the breadth of the application range of our logic QSL, we present several case studies from two very different areas:
	\begin{itemize}
		\item Our first example given in Section \ref{sec-vqa-exam} is formal verification of Variational Quantum Algorithm (VQA) \cite{MRB16,PMS14} for finding the ground state of a quantum system, which has potential applications in quantum chemistry for designing new materials and drugs.
		A typical VQA can be split into different subprograms that are suited to locally reasoning. Then the frame rules together with \textsc{UnCR} are used to derive global correctness with entangled pre- and post-conditions. \textit{In particular, an analysis based on QSL reveals that the VQA presented in the tutorial of Google's Cirq \cite{Cirq} is incorrect}.  
		\item In Section \ref{sec-security-exam}, we use QSL to verify the security
      of quantum one-time pad (QOTP) \cite{BR03,MTW00} and quantum secret
      sharing (QSS) \cite{CGL99,HBB99}. Unlike previous work, the QSL
      verification of QOTP and QSS is \emph{scalable}: increasing the number of
      registers that algorithms employ does not complicate the verification. In
      particular, rule \textsc{Frame} with the supported assertion (SP) enables
      us to avoid the very complicated mathematical calculations used in earlier
      verifications of QOTP~\cite{BHY19}.
	\end{itemize}

	\section{Preliminaries}
	\label{sec-def}
	For the convenience of the reader, we briefly review basic notions of quantum information and programming as well as the logic of bunched implication.  
	
	\subsection{Basics of Quantum Information}
	\label{sec basic Quantum}
	
  The \emph{state space} of a quantum system is a Hilbert space $\cH$, which is
  essentially a vector space in the finite-dimensional case. A \emph{pure state}
  of the system is a unit column vector $|\psi\>\in\cH$. For example, the state
  space of a quantum bit (aka qubit) is a two-dimensional Hilbert space with
  basis states $|0\> = \left[\begin{array}{c} 1 \\ 0\end{array} \right]$ and
  $|1\> = \left[\begin{array}{c} 0 \\ 1 \end{array}\right]$, and any pure state
  of a qubit can be described in the form $\alpha|0\>+\beta|1\> =
  \left[\begin{array}{c} \alpha \\ \beta \end{array}\right]$ satisfying
  normalization condition $|\alpha|^2+|\beta|^2 = 1$. When the state is not
  completely known but could be in one of some pure states $|\psi_i\>$ with
  respective probabilities $p_i$, we call $\{(p_i,|\psi_i\>)\}$ an
  \emph{ensemble} of pure states or a \emph{mixed state}, and the system is
  fully described by the \emph{density operator} $\rho =
  \sum_ip_i|\psi_i\>\<\psi_i|$. For example, the completely mixed state of a
  qubit can be seen as ensemble $\{(0.5,|0\>), (0.5,|1\>)\}$ (i.e. the state is
  either $|0\>$ or $|1\>$ with the same probability 0.5) or density matrix
  $\frac{1}{2}(|0\>\<0|+|1\>\<1|) = \left[\begin{array}{cc} 0.5 & 0 \\ 0 & 0.5
  \end{array}\right].$
	
  The evolution of a quantum system is modelled by a \emph{unitary operator}
  $U$; i.e.  a complex matrix with $UU^\dag=U^\dag U$ being the identity
  operator, where $\dag$ is conjugate transpose. In quantum computing, operators
  are often called \emph{quantum gates}. For example, the Hadamard gate
      $H=\frac{1}{\sqrt{2}}\left[\begin{array}{cc} 1 & 1 \\ 1 & -1
      \end{array}\right]$ maps $|0\>, |1\>$ to their superpositions
      $|\pm\>=\frac{1}{\sqrt{2}}(|0\>\pm|1\>)$.
	
  Unlike a classical system which can be observed directly without changing its
  state, we need to perform a quantum measurement to extract information from a
  quantum state which inevitably leads to state collapse. Formally, a
  \emph{projective quantum measurement} consists of a set of \emph{projections},
  i.e., self-adjoint and idempotent linear operators,\footnote{That is,
  $P:\cH\rightarrow\cH$ is a projection over $\cH$ iff $P = P^\dag = P^2$.}
  $M_0,M_1,\dots, M_n$. When such a measurement is applied to a quantum state
  $\rho$, we obtain one of the classical outcome $i\in\{0,1,\dots,n\}$ with
  probability $p_i = \mathrm{tr}(M_i\rho)$, and the post-measurement state of
  the system is then  $\frac{M_i\rho M_i}{p_i}$.
	
	We use variables $p,q,r, ...$ to denote quantum systems. Operations in quantum
  computing are often performed on a composite system consisting of multiple
  qubits. To indicate which system a state describes or an operation acts on, we
  use subscripts; for example, $\cH_p$ is the state space of system $p$,
  $|0\>_{p}$ is the pure state $|0\>$ of the system $p$ and $|1\>_{q}\<1|$ is
  the density matrix of the system $q$. The composite system is described by the
  tensor product of its subsystems; for example, a composite system $pq$ has the
  state space $\cH_p\otimes\cH_q$, and $|0\>_p\otimes|1\>_q$ (or, $|0\>_p|1\>_q$
  for short) is a pure state in which  subsystem $p$ is in state $|0\>$ and
  subsystem $q$ is in state $|1\>$. Due to the superposition principle, there
  exist states like $$|\Phi\>_{pq} =
  \frac{1}{\sqrt{2}}(|0\>_p|0\>_q+|1\>_p|1\>_q)$$ that cannot be written in the
  simple tensor form $|\phi\>_p|\psi\>_q$, which are called \emph{entangled
  states}. These states play a crucial role in applications of quantum
  computation and quantum communication. 
	
  The state of a composite system fully determines the state of each subsystem.
  Formally, given composite system $pq$ in state $\rho$, subsystem $q$ is then
  in state $\tr_p(\rho)$, where the partial trace $\tr_p(\cdot)$ over $p$ is a
  mapping from operators on $\cH_p\otimes\cH_q$ to operators on $\cH_q$ defined
  by: $$\tr_p(|\phi_p\>_p\<\psi_p|\otimes|\phi_q\>_q\<\psi_q|) =
  \<\psi_p|\phi_p\>\cdots|\phi_q\>_q\<\psi_q|$$  for all
  $|\phi_p\>,|\psi_p\>\in\cH_p$ and $|\phi_q\>,|\psi_q\>\in\cH_q$ together with
  linearity. The state $\tr_q(\rho)$ of subsystem $q$ can be defined
  symmetrically.  We often use the notations $\rt{\rho}{p} \triangleq
  \tr_p(\rho)$ and $\rt{\rho}{q} \triangleq \tr_q(\rho)$ in order to explicitly
  indicate that $\rt{\rho}{p}$ and  $\rt{\rho}{q}$ are states of $p,q$,
  respectively.
	
	\vspace{0.2cm}
	\noindent\textbf{Summary of Notations.}
	Let $\vars$ be  the set of all quantum variables. A \emph{quantum register} is a list of \emph{distinct} variables $\qbar = q_1,\dots,q_n$. Each quantum variable $q$ has a type $\cH_q$, which is the state Hilbert space of quantum system denoted by $q$.
	For a set of quantum variables $S = \{q_1,\dots,q_n\}\subseteq \vars$ (or a quantum register $\qbar = q_1,\dots,q_n$), we fix following notations:
	\begin{itemize}
		\item $\cH_S = \bigotimes_{i=1}^n\cH_{q_i}$: the Hilbert space of $S$.
		\item $\dim(S)$: the dimension of $\cH_S$.
    \item $\cD(S)$: the set of all (mixed) quantum states (i.e. density matrices)
      of $S$. In particular, for any $\rho\in\cD(S)$, its \emph{domain} is
      defined as $\dom{\rho}\triangleq S$; we write $\cD \triangleq
      \bigcup_{S\subseteq \vars}\cD(S)$ for the set of all states.
    \item $\cP(S)$: the set of projections on $\cH_S$. In particular, for any
      $P\in\cP(S)$, its domain is defined as $\free{P} \triangleq S$.  Since
      there is a one-to-one correspondence between projections and closed
      subspaces, we sometimes called closed subspaces of $\cH_S$ projections.
      We write $\cP\triangleq \bigcup_{S\subseteq \vars}\cP(S)$ for the set of
      all projections.
		\item $
		\rt{\rho}{S} \triangleq 
		\tr_{\dom{\rho}\backslash S}(\rho)
		$: the \emph{restriction} of state $\rho$ on $S$,  defined as a reduced density operator over $S\cap\dom{\rho}$. 
	\end{itemize}
	\subsection{Quantum Programs: Syntax and Semantics}
	
	For simplicity of presentation, we consider a purely quantum
        extension of \textbf{while}-language, namely the quantum
        ${\bf while}$-language~\cite{Ying11}---that is, we do not
        allow classical variables.


	\begin{definition}[Syntax \cite{Ying11}]\label{q-syntax}
		The quantum \textbf{while}-programs are defined by
		the grammar:
		\begin{align*} \prog::=\ &\mathbf{skip}\  |\ \prog_1;\prog_2\ |\ q:=|0\rangle\ |\ \qU \\
		&|\ \mathbf{if}\ \left(\square m\cdot M[\qbar] =m\rightarrow \prog_m\right)\ \mathbf{fi}\\ 
		 &|\ \mathbf{while}\ M[\qbar]=1\ \mathbf{do}\ \prog\ \mathbf{od}
		\end{align*}
	\end{definition}
	
	The program constructs defined above are explained as follows. First, $q:=|0\rangle$
  initializes the quantum variable $q$ in a basis state $|0\rangle$, and $\qU$
  applies a unitary transformation $U$ to a sequence $\qbar$ of quantum
  variables. The case statement $\mathbf{if}\cdots\mathbf{fi}$ performs the
  projective measurement $M = \{M_m\}$ on $\qbar$, and then chooses a subprogram
  $\prog_m$ to execute according to measurement outcome $m$. In the loop $\mathbf{while}\cdots\mathbf{od}$, the projective measurement $M = \{M_0, M_1\}$ in the guard has only two possible outcomes $0,1$: if the outcome is $0$ the loop terminates, and if the outcome is $1$ it executes the loop body $\prog$ and enters the loop again. For simplicity of presentation, we will use the following abbreviation:
	${\bf for}\ i = 1,\cdots,N\ {\bf do}\ \prog_i\ {\bf od}\triangleq \prog_1;\cdots;\prog_N$.
	
	For each program $\prog$, we write $\var(\prog)$ for the set of all quantum variables in $\prog$. If $\vars\supseteq\var(\prog)$ is a set of quantum variables, and $\rho\in\mathcal{D}(\vars)$, then
	$\langle \prog,\rho\rangle$ is called a configuration (of domain $\vars$). 
	
	\begin{definition}[Operational Semantics \cite{Ying11}]\label{def-op-sem} The  operational semantics of quantum programs is defined as a transition relation $\rightarrow$ by the following transition rules: 
		\begin{equation*}\begin{split}&({\rm Sk})\ \ \langle\mathbf{skip},\rho\rangle\rightarrow\langle \emptyprog,\rho\rangle\qquad ({\rm In})\ \ \ \langle
		q:=|0\rangle,\rho\rangle\rightarrow\langle \emptyprog,\rho^{q}_0\rangle\\
		&({\rm UT})\ \langle\qU,\rho\rangle\rightarrow\langle
		\emptyprog,U\rho U^{\dag}\rangle\\ 
		&({\rm SC})\ \frac{\langle \prog_1,\rho\rangle\rightarrow\langle
			\prog_1^{\prime},\rho^{\prime}\rangle} {\langle
			\prog_1;\prog_2,\rho\rangle\rightarrow\langle
			\prog_1^{\prime};\prog_2,\rho^\prime\rangle}\\
		&({\rm IF})\ \langle\mathbf{if}\ (\square m\cdot
		M[\qbar]=m\rightarrow \prog_m)\ \mathbf{fi},\rho\rangle\rightarrow\langle
		\prog_m,M_m\rho M_m^{\dag}\rangle\\
		&({\rm L}0)\ \langle\mathbf{while}\
		M[\qbar]=1\ \mathbf{do}\
		\prog\ \mathbf{od},\rho\rangle\rightarrow\langle \emptyprog, M_0\rho M_0^{\dag}\rangle\\
		&({\rm L}1)\ \langle\mathbf{while}\
		M[\qbar]=1\ \mathbf{do}\ \prog\ \mathbf{od},\rho\rangle \\
		&\qquad\qquad\quad\ \ \rightarrow
		\langle \prog;\mathbf{while}\ M[\qbar]=1\ \mathbf{do}\ \prog\ \mathbf{od}, M_1\rho
		M_1^{\dag}\rangle\end{split}\end{equation*}
		$\emptyprog$ is the empty program. In (In), $\rho^{q}_0=\sum_n|0\rangle_q\langle n|\rho|n\rangle_q\langle
		0|$.
		In (SC), we use the convention $\emptyprog;\prog_2=\prog_2.$
		In (IF), $m$ ranges over every possible outcome of measurement $M=\{M_m\}.$
	\end{definition}
	
	Transitions in rules (IF), (L0) and (L1) are essentially probabilistic; but we adopt a convention from \cite{Sel04b} to present them as  non-probabilistic transitions. For example, for each $m$, the transition in (IF) happens with probability $p_m=\tr(M_m^\dag M_m\rho)$ and the  program state $\rho$ is changed to $\rho_m=M_m\rho M_m^\dag /p_m$. We can combine probability $p_m$ and density operator $\rho_m$ into a partial density operator $M_m\rho M_m^\dag=p_m\rho_m$. This convention significantly simplifies the presentation.
	
	\begin{definition}[Denotational Semantics \cite{Ying11}]
		\label{def den sem} 
		Let $\vars$ be a set	of variables. Then for any quantum program $\prog$ with $\var(\prog)\subseteq\vars$, its semantic function of domain $\vars$ is the mapping 
		$\sem{\prog}_{\vars}:\mathcal{D}(\vars)\rightarrow \mathcal{D}(\vars)$ defined by $\sem{\prog}_{\vars}(\rho)=\sum\{\!|\rho^\prime: \langle \prog,\rho\rangle\rightarrow^\ast\langle \emptyprog,\rho^\prime\rangle|\!\}$ for every $\rho\in\mathcal{D}(\vars)$, where $\rightarrow^\ast$ is the reflexive and transitive closure of $\rightarrow$, and $\left\{\!|\cdot|\!\right\}$ denotes a multi-set.
	\end{definition}
	
	Note that auxiliary variables in $\vars\setminus\var(\prog)$ are allowed in the above definition of semantic function $\sem{\prog}_{\vars}$. The following proposition shows that the denotational semantics of a program $\prog$ is independent of these auxiliary variables.   
  \begin{proposition}[Proposition 3.3.5 in \cite{Ying16}]
		\label{prop sem qo} For any program $\prog$ and any set $\vars\supseteq\var(\prog)$ of variables, the semantic function of domain $\vars$ is a cylindric extension of the  semantic function of domain $\var(\prog)$:
		$\sem{\prog}_{\vars} = \sem{\prog}_{\var(\prog)}\otimes\cI_{\vars\backslash\var(\prog)},$
		where $\cI_{\vars\backslash\var(\prog)}$ is the identity quantum operation in $\cH_{\vars\backslash\var(\prog)}$.
	\end{proposition}
	
	
	
	\subsection{Brief review of BI-Logic}
	\label{sec brief review of BI}
  Next, we briefly review the logic of Bunched Implications (BI)
  \cite{OP99,Pym02}. BI is a sub-structural logic with the following syntax:
	$$
	\phi,\psi ::= p\in\AP\ |\ \top\ |\ \bot\ |\ \phi\wedge\psi\ |\ \phi\vee\psi\ |\ \phi\rightarrow\psi\ |\ \phi\ast\psi
	|\ \phi\sepimp\psi
	$$
  where $p$ ranges over a set $\AP$ of atomic propositions. Besides standard
  propositional logic, BI contains a substructural fragment -- the separating
  conjunction $\ast$ and separating implication $\sepimp$ (``magic wand''). A
  distinction between $\ast$ and $\wedge$ is that $\ast$ is not idempotent,
  i.e., $P\ast P \neq P$. For example, in the standard heap model of separation
  logic, the separating conjunction $P\ast Q$ is true of a heap if it can be
  split into two heaplets, one of which makes $P$ true and the other of which
  makes $Q$ true. The implication $\sepimp$ is adjoint to $*$. For example, $P
  \sepimp Q$ holds in some heap if adding a separate heap satisfying $P$ leads
  to a combined heap satisfying $Q$.
	
	The most general semantics of BI is given in terms of a kind of Kripke structures, called BI frames. Standard BI frame is based on a pre-ordered commutative monoid:
	\begin{definition}[BI frame \cite{OP99}]
		\label{def BI frame}
		A BI frame is a tuple $\cX = (X,\circ,\preceq,e)$, where $X$ is a set equipped with a preorder $\preceq$, and $\circ: X\times X\rightarrow X$ is a partial binary operation with an unit element $e$ and satisfying the following conditions:
		\begin{enumerate}
			\item (Unit Existence) for all $x$, $x =  x\circ e = e\circ x$;
			\item (Commutativity) $x\circ y = y\circ x$;
			\item (Associativity) $x\circ (y\circ z) = (x\circ y)\circ z$;
			\item (Compatible with $\preceq$) 
			$x\preceq x^\prime$ and $y\preceq y^\prime$ and both $x\circ x^\prime$ and $y\circ y^\prime$ are defined, then $x\circ x^\prime\preceq y\circ y^\prime$.
		\end{enumerate}
    Above, equalities state that either both sides are defined and equal, or both
    sides are undefined.
	\end{definition}
	
	Intuitively, if we choose the collections of resources as possible worlds, then $\circ$ can be interpreted as a commutative combination of resources. The identity $e$ is an empty resource or lack of resource, and combine any resource $x$ and empty resource $e$ yields $x$ itself. Based on the combination, a preorder is defined: if $x$ is a combination of resources $y$ and $z$, it should be ``larger'' than $y$ since it contains $y$. 
	
	
  The semantics of formulas depends on the semantics of atomic propositions.  A
  \emph{valuation} is a mapping $\cV: \AP \rightarrow \wp(X)$, and it is
  \emph{monotonic} if $x\in\cV(p)$ and $y \succeq x$ implies $y \in \cV(p)$. A
  BI frame $\cX$ together with a monotonic valuation $\cV$ gives a BI model
  $\cM$.
	
	\begin{definition}[Satisfaction in BI models \cite{OP99}]
		\label{def satisfaction BI}
    Given a BI formula $\phi$ and a BI model $\cM = (X,\circ,\preceq,e,\cV)$.
    For each $x\in X$, the relation $x\models \phi$ is defined by induction on
    $\phi$:
		\begin{align*}
		&  x\models_\cM p \ \text{\rm iff}\ x\in\cV(p)\\
		&  x\models_\cM \top :\  \text{always} \qquad
		x\models_\cM \bot :\  \text{never} \\
		&  x\models_\cM \phi_1\wedge\phi_2 \ \ \;\!  \text{\rm iff}\ \ x\models_\cM\phi_1\ \text{and}\ x\models_\cM\phi_2 \\
		&  x\models_\cM \phi_1\vee\phi_2   \ \ \;\!  \text{\rm iff}\ \ x\models_\cM\phi_1\ \text{or}\ x\models_\cM\phi_2\\
		&  x\models_\cM \phi_1\rightarrow\phi_2  \  \text{\rm iff}\ \ \forall x^\prime\succeq x,\ x^\prime\models_\cM\phi_1\ \text{implies}\ x^\prime\models_\cM\phi_2\\
		&  x\models_\cM \phi_1\ast\phi_2  \ \ \ \;\!\! \text{\rm iff}\ \ \exists y,z\ \text{s.t.}\ y\circ z\ \text{is defined and}\ x \succeq y\circ z,\\ 
		&\qquad\qquad\qquad\qquad\ \  y\models_\cM \phi_1\ \text{and}\ z\models_\cM\phi_2 \\
		&  x\models_\cM \phi_1\sepimp\phi_2  \ \;\! \text{\rm iff}\ \ \forall y\ \text{s.t.}\ x\circ y\ \text{is defined}, \\ 
		&\qquad\qquad\qquad\qquad\ \   y\models_\cM \phi_1\ \text{implies}\ x\circ y\models_\cM \phi_2. 
		\end{align*}
	\end{definition}
	
	
  Following \cite{Pym02} (see also \cite{Doc19}), a sound and complete
  Hilbert-style proof system of BI is presented in Supplementary Material \ref{sec app sub Hilbert-style rules for BI}.

	\section{Quantum Interpretation of BI Logic}\label{sec-q-interpret}
	
  Now, we are ready to present our quantum model of BI, using the
  resource semantics of BI. After defining the model, we introduce some atomic
  propositions. To lay the groundwork for the separation logic, we explore a
  technical property called \emph{restriction}---which will be important for the
  frame rule---and we define a \emph{modification} operation, an analog of
  substitution that we will use for reasoning about initialization and unitary
  transformations.
	
	\subsection{BI Frame of Quantum States}
	
  The basic idea of our model is to consider quantum states over specific
  registers as resources. Then, the separating conjunction is introduced to
  model independent combinations of spatially separate quantum resources
  (quantum states over disjoint registers). Formally, we define:
	
	\begin{definition}
		\label{def tensor coupling}
		The partial binary functions $\circ: \cD\times\cD\rightarrow \cD$ on quantum
    states is defined by:
    \[
      \rho_1\circ\rho_2 \triangleq
      \begin{cases}
        \rho_1 \otimes \rho_2 &: \text{if } \dom{\rho_1}\cap \dom{\rho_2}=\emptyset \\
        \text{undefined} &: \text{otherwise.}
      \end{cases}
    \]
	\end{definition}	
	
  Essentially, $\circ$ takes the tensor product of two quantum states with
  disjoint domains. Note that in our setting, the tensor product $\otimes$ is
  commutative since every quantum state $\rho\in\cD$ is tagged with its domain.
  For example, $|1\>_p\<1|\otimes|0\>_q\<0| = |0\>_q\<0|\otimes|1\>_p\<1|$
  denote the same state in $pq$. For the partial order over quantum states, we
  take the following:	

	\begin{definition}
		\label{def partial order quantum state}
		Let $\preceq$ be the partial order over $\cD$: 
		$
		\rho\preceq\rho^\prime \text{\ iff\ } 
		\dom{\rho}\subseteq\dom{\rho^\prime} \text{\ and\ } \rho = \rt{\rho^\prime}{\dom{\rho}}.
		$
	\end{definition}	
  Intuitively, $\rho \preceq \rho^\prime$ means that $\rho$ describes a
  subsystem of $\rho^\prime$; more precisely, if we discard the subsystem
  $\dom{\rho^\prime}\backslash\dom{\rho}$ of $\rho^\prime$, then the remaining
  subsystem is in state  $\rho$. Combining all of the ingredients defined, we
  have:
	\begin{proposition}\label{prop quantum BI frame} $(\cD,\circ, \preceq,1)$ forms a BI frame, where scalar number $1$ is understood as the state over the empty register. \end{proposition}	

	
	\subsection{Atomic Propositions about Quantum States}
	\label{sec free choice of AP}
	
  To complete our description of the quantum BI logic, we introduce three atomic
  propositions and interpret them in quantum states. In general, we have a great
  deal of freedom in selecting these atomic propositions; the only requirement
  is that their interpretation must be monotone with respect to the pre-order
  $\preceq$. Our atomic propositions are fairly general, but motivated by
  applications of our separation logic.
	
	{\vskip 3pt}
	
	\noindent\textbf{Propositions denoting free variables.} We first introduce a set of atomic propositions $\bD[S]$ for each variable set $S\in\vars$ with domain defined by $\free{\bD[S]}\triangleq S$, and interpret it as the state with domain at least $S$:
	\begin{equation}\label{atom-0}
	\sem{\bD[S]} \triangleq \left\{\rho\in\cD: S\subseteq\dom{\rho}\right\}.\end{equation}
	
	{\vskip 3pt}
	
	\noindent\textbf{Propositions for qualitative analysis.} For qualitative analysis of quantum programs, we often use projection operators as atomic propositions.   
	For a projection $P\in\cP$ as an atomic proposition, its semantics $\sem{P}$ is defined as the following set of quantum states:
	\begin{equation}\label{atom-1}\begin{split}
	\sem{P} &\triangleq \left\{\rho\in\cD: \free{P}\subseteq\dom{\rho}\text{\ \&\ }\supp\big(\rt{\rho}{\free{P}}\big)\subseteq P\right\}. 
	\end{split}
	\end{equation}
	where the \emph{support} of a state $\rho\in\cD$ is the (topological) closure of the subspace spanned by its eigenvectors with nonzero eigenvalues, or equivalently, $\supp(\rho) = \{|\phi\>\in\cH_{\dom{\rho}}:\ \<\phi|\rho|\phi\> = 0\}^\bot.$\footnote{${}^\bot$ stands for ortho-complement.}
	Let us carefully explain the definition of $\sem{P}$. In the case that $\rho$ has the same domain of $P$, it is natural to define $\rho\in\sem{P}$ if its support $\supp(\rho)$ lies in $P$, or equivalently, $\rho$ is invariant under projection operator $P$. 
	In the case where $\dom{\rho}$ and $\free{P}$ are not the same, in order to make $\sem{P}$ upward-closed (i.e., monotonic): $\rho\in\sem{P}$ and $\rho\preceq\rho^\prime$ imply  $\rho^\prime\in\sem{P}$, it is appropriate to require that $\rho\in\sem{P}$ iff (i) $\dom{\rho}\supseteq\free{P}$; and (ii) the restricted state of $\rho$ on $\free{P}$ is in $\sem{P}$.

	{\vskip 3pt}
	
	\noindent\textbf{Atomic propositions expressing uniformity in quantum security.} As is well-known, probabilistic uniformity is a basic property in verification of security protocols. To describe uniformity in quantum protocols, we introduce an atomic proposition $\unia[S]$ for each $S\subseteq\vars$ denoting finite-dimensional quantum systems. Its domain is $\free{\unia[S]}\triangleq S$. 
	The semantics of $\unia[S]$ is defined as the following set of quantum states:
	\begin{equation}\label{atom-3}
	\sem{\unia[S]} \triangleq \left\{\rho\in\cD: S\subseteq\dom{\rho}\text{\ \&\ }\rt{\rho}{S} = \frac{\id_S}{\dim(S)}\right\},\end{equation}
  where $I_S$ is the identity density on the quantum system over registers $S$.
  The intuition behind defining equation (\ref{atom-3}) is quite simple: for a
  state $\rho$ in $\sem{\unia[S]}$ such that $S\subseteq\dom{\rho}$, its
  restriction on $S$ should be the completely mixed state,
  $\frac{\id_S}{\dim(S)}$, which means ``uniformly distributed" over all
  orthonormal bases of the system denoted by $S$. 
	
	{\vskip 3pt}
	
	\noindent\textbf{Axiom schema for atomic formulas.} With the interpretation of atomic propositions, we have: 
	\begin{proposition}
		\label{prop axiom projection}
    \strut{}
		\begin{enumerate}
			\item For all $S\subseteq \vars$ and identity operator $I_S$ over $\cH_S$, we have:  $\models \bD[S]\leftrightarrow \id_{S}.$ 
			\item For all $P,Q\in\cP$ with disjoint domains, we have:
        $\models  P\wedge Q \leftrightarrow (P\otimes Q) ;$  	
			\item If $S_1\subseteq S_2$, then $\models\unia[S_2]\rightarrow\unia[S_1]$.
			\item If $S_1, S_2$ are disjoint, then:
        $$ \models(\unia[S_1]\ast\unia[S_2])\leftrightarrow\unia[S_1\cup S_2] .$$
		\end{enumerate}
    Note that $\otimes$ is \emph{not} a connective in BI: instead, it stands
    for the mathematical tensor product. Thus, $P\otimes Q$ is a
    projection and can be considered as atomic formula.
	\end{proposition}
	
	\subsection{Restriction Property}
	\label{sec res}
	After choosing (the interpretation of) atomic propositions in the quantum
  frame $(\cD,\circ,\preceq,1)$, the semantics of all BI formulas can be defined
  using Definitions \ref{def satisfaction BI}. As is well-known, the frame rule 
  plays an essential role in separation logic, and in turn it heavily relies on
  the restriction property that satisfaction only depends on the free variables
  appearing in a BI formula $\phi$. The restriction property was also identified and
  generalized in prior work on probabilistic separation logic \cite{BHL19}.
	However, the restriction property: $$\rho\models\phi\Rightarrow
  \rt{\rho}{\free{\phi}}\models\phi$$ where $\free{\phi}$ stands for the free
  variables occurring in $\phi$, does not hold for our quantum setting, even for
  the ordinary implication $\phi=\phi_1\rightarrow\phi_2$ (see Definition
  \ref{def satisfaction BI} for its semantics). Essentially, the validity of the
  restriction property in the probabilistic setting can be attributed to a
  fundamental fact in probability theory---the existence of
  extensions.\footnote{For two joint-distributions $\mu_{AB}$ and $\mu_{BC}$
  over sets $A,B$ and $B,C$ respectively, if they are consistent on $B$ (with
the same marginal on $B$) then there exists joint-distribution $\mu_{ABC}$ over
$A,B,C$ which takes $\mu_{AB}$ and $\mu_{BC}$ as marginals.}
  Unfortunately, this does not always hold for quantum systems. Indeed, it is violated by the well-known phenomenon of  
  ``Monogamy'' --  one of the most fundamental properties of
  entanglement.\footnote{If two qubits $A$ and $B$ is maximally correlated, then they
  cannot be correlated at all with a third qubit $C$; more precisely, if $A$ and $B$ are in a maximally entangled state, then $A$ and $C$ cannot be in any entangled state.}
	
  Since we wish to have a frame rule in QSL, we need to recover the restriction property to a certain extent.
  While not all formulas satisfy this property, we can identify a subset of them  that
  do satisfy it.
	\begin{definition}
		\label{def Res}
		The formulas generated by following grammar are denoted by $\res$.
		$$
		\phi,\psi ::= p\in\AP\ |\ \top\ |\ \bot\ |\ \phi\wedge\psi\ |\ \phi\vee\psi\ |\ \phi\ast\psi
		$$
	\end{definition} 
	\begin{proposition}
		\label{pro res set}
		Any formula $\phi\in\res$ is restrictive; that is, for any $\rho\models\phi$, $\rt{\rho}{\free{\phi}}\models\phi$.
	\end{proposition}
	
	The above simple treatment of restriction property is sufficient for the purpose of this paper. A more intrinsic way for recovering this property in the quantum setting will be discussed in Section~\ref{sec:discussion}.

	\subsection{Quantum Modification of BI Formulas}
	\label{sec modification BI formulas}
  In classical program logic, substitution is used in the inference rule about
  assignment statements. In the quantum setting, due to no-cloning of quantum
  data, the role of assignment is played by two basic constructs: unitary
  transformation and initialization. We conclude this section by defining a
  technique of modifying BI formulas, which we will need reasoning about these
  operations.
	
	\begin{definition}[Modification of atomic propositions]
		\label{def sub atomic prop}
    Let $\prog$ be a unitary transformation $\qU$ or an initialisation $\qI$.
    For any $p\in\AP$, we write $p[\prog]$ for the $\prog$-modification of $p$.
    For the three classes of  atomic propositions defined in Sec. \ref{sec free choice of
    AP},  $p[\prog]$ is defined as follows:
	\begin{enumerate}
			\item For an atomic proposition $\bD[S]$ defined in Eq. (\ref{atom-0}), $\bD[S][\prog] \triangleq \bD[S]$;
			\item 
			For an atomic  proposition  $P\in\cP$ as a   projection defined in Eq. (\ref{atom-1}), 
		\begin{align*}
			&P[\qU]\triangleq\left\{
			\begin{array}{ll}
			P_{U[\qbar]} & \text{if}\ \qbar\subseteq\free{P}; \\
			P &\qbar\cap\free{P}=\emptyset; \\
			{\rm undefined} & {\rm otherwise;}
			\end{array}
			\right. \\
			&P[\qI]\triangleq\left\{
			\begin{array}{ll}
			\bD[q]\wedge\lceil P\rceil_q & \text{if}\ q\in\free{P}; \\
			P &{\rm otherwise;}
			\end{array}
			\right. 
			\end{align*}
			where projections $P_{U[\qbar]}$ and $\lceil P\rceil_q$ are given as follows: $$P_{U[\qbar]} = 
	(U^{\qbar\dag}\otimes\id_{\free{P}\backslash\qbar})P(U^{\qbar}\otimes\id_{\free{P}\backslash\qbar}),$$ and
	$\lceil P\rceil_q ={\qbigvee}\{{\rm ~closed~subspaces}\ T:
     |0\>_q\<0|\otimes T\subseteq P\} \in \cP(\free{P}\backslash q)$. Here,  $\qvee$ is the disjunction of projections  in quantum logic, that is, for projections $P,Q$ with the same domain, $P\qvee Q=\overline{\operatorname{\spa}(P\cup Q)}$ with \textquotedblleft$\overline{\, \cdot\, }$\textquotedblright\ standing for (topological)  closure.
			\item For any atomic proposition $\unia[S]\in\cU$ for uniformity defined in Eq.(\ref{atom-3}), 
			\begin{enumerate}
				\item If $\qbar\subseteq S$ or $\qbar\cap S=\emptyset$, then $\unia[S][\qU] \triangleq\unia[S]$; 
				
				otherwise, $\unia[S][\qU]$ is undefined;
				\item If $q\notin S$, then $\unia[S][\qI] \triangleq\unia[S]$; 
				
				otherwise, $\unia[S][\qI]$ is undefined.
			\end{enumerate}
		\end{enumerate}
	\end{definition}
	
	The modification of some atomic propositions/BI formulas may not exist;
	we write $\phi[\prog]\Mexist$ whenever $\phi[\prog]$ is defined.	
	The notion of modification can be easily extended to all BI formulae: 
	\begin{definition}[Modification of BI formulas]
		\label{def sub 2-BI form} 
		Let $\prog$ be unitary transformation $\qU$ or initialisation $\qI$. The modification $\phi[\prog]$ of BI formula $\phi$ is defined by induction on the structure of $\phi$:
		\begin{enumerate}
			\item if $\phi\equiv \top$ or $\bot$, then $\phi[\prog] \triangleq \phi$;
			\item if $\phi\equiv p\in\AP$, then $\phi[\prog]$ is defined according to Definition \ref{def sub atomic prop};
			\item if $\phi\equiv \phi_1\ \triangle\ \phi_2$ where $\triangle\in\{ \wedge,\vee\}$ and $\phi_1[\prog]\Mexist$ and $\phi_2[\prog]\Mexist$, then $\phi[\prog]\triangleq\phi_1[\prog] \ \triangle\ \phi_2[\prog]$;
			\item if $\phi\equiv \phi_1\ast\phi_2$, $\phi_i[\prog]\Mexist$ and $\qbar\subseteq\free{\phi_i}$ or $\qbar\cap\free{\phi_i}=\emptyset$ for $i=1,2$, then
			\begin{enumerate}
				\item if $\prog\equiv\qU$, then $\phi[\prog]\triangleq\phi_1[\prog] \ast \phi_2[\prog]$;
				\item if $\prog\equiv\qI$, then
				\begin{itemize}
					\item if $q\notin\free{\phi_1}\cup\free{\phi_2}$, $\phi[\prog]\triangleq\phi_1[\prog]\ast \phi_2[\prog]$;
					\item if only one of $q\in\free{\phi_1}$, $q\in\free{\phi_2}$ is
            satisfied,
            then $\phi[\prog]\triangleq(\phi_1[\prog] \wedge \phi_2[\prog])\wedge(\bD[\free{\phi_1}\backslash
            q]\ast\bD[\free{\phi_2}\backslash q])$;
				\end{itemize}
        The reason for the complexity of this case will be seen in the program
        logic; roughly speaking, initialization on $q$ is special because it can
        introduce independence: it makes $q$ independent from all variables.
			\end{enumerate}
            \item otherwise, $\phi[\prog]$ is undefined.
		\end{enumerate}
	\end{definition}
	
	
	A close connection between the semantics of a BI formula $\phi$ and its modification $\phi[\prog]$ is shown in the   following:  
	
	\begin{proposition}
		\label{pro modification}
		Let $\prog$ be unitary transformation $\qU$ or initialisation $\qI$, and $\phi$ be any BI formula. If its modification  $\phi[\prog]$ is defined, then:
		\begin{enumerate}
			\item $\phi$ and $\phi[\prog]$ have the same domain: $\free{\phi} = \free{\phi[\prog]}$;
			\item for all $\rho\in\cD(\free{\phi}\cup\var(\prog))$, if $\rho\models \phi[\prog]$, then $\sem{\prog}(\rho)\models \phi$.
		\end{enumerate}
	\end{proposition}
	

	\section{Separation Logic for Quantum Programs}
	\label{sec QSL}
	
  Now we are ready to present our separation logic for quantum programs, using
  quantum BI formulas as the assertion language. 	
	
	\subsection{Judgments and Validity}
	
  Let us first define judgments (correctness formulas) in quantum separation
  logic. A judgment is a Hoare triple of the form $\{\phi\}\prog\{\psi\}$ with
  both precondition $\phi$ and postcondition $\psi$ being \emph{restrictive} BI
  formulas (cf. Definition \ref{def Res}).
	
	\begin{definition}[Validity] Let $\vars$ be a set of quantum variables with $\free{\phi},\free{\psi},\var(\prog)\subseteq\vars$. Then a correctness formula $\{\phi\}\prog\{\psi\}$ is true in the sense of partial correctness with respect to $\vars$, written 
		$\vars \models\{\phi\}\prog\{\psi\}$, 
		if we have:
		$$\forall \rho\in\cD(\vars),\quad \rho\models\phi\Rightarrow\sem{\prog}_{\vars}(\rho)\models\psi.$$ Here, satisfaction relation $\rho\models\phi$ and $\sem{\prog}_{\vars}(\rho)\models\psi$ are defined according to the quantum interpretation of BI logic given in Section \ref{sec-q-interpret}.
	\end{definition}
	
	
	The following theorem indicates that satisfaction does not depends on auxiliary variables.
	\begin{theorem}
		\label{thm eq glb var set}
		For any two sets $\vars$ and $\vars^\prime$ containing all free variables of $\phi, \psi$ and $\prog$,
		$$\vars\models\{\phi\}\prog\{\psi\} \text{\ if\ and\ only\ if\ }\vars^\prime\models\{\phi\}\prog\{\psi\}.$$
	\end{theorem}
	
	As a consequence, we can drop $\vars$ from $\vars\models\{\phi\}\prog\{\psi\}$ and simply write $\models\{\phi\}\prog\{\psi\}$. 
	
	In the remainder of this section, we gradually develop the proof system for our quantum separation logic. For better readability, this proof system is organised as several sets of inference rules.

	\subsection{Inference Rules for Program Constructs}
	
	
	\begin{figure}\centering
		\begin{equation*}\begin{split}
		&\textsc{Skip}\ \frac{}{\{\phi\}\mathbf{skip}\{\phi\}}\quad 
		\textsc{Init}\ \frac{\phi[\qI]\Mexist}{\left\{\phi[\qI] \right\}\qI\{\phi\}} \\[0.1cm]
		&\textsc{Unit}\ \frac{\phi[\qU]\Mexist}{
			\{\phi[\qU]\}\qbar:=U\left[\qbar\right]\{\phi\}}\\[0.1cm]
		&\textsc{Seq}\quad
		\frac{\{\phi\}\prog_1\{\psi\}\ \ \ \ \ \ \{\psi\}\prog_2\{\mu\}}{\{\phi\}\prog_1;\prog_2\{\mu\}}\\[0.1cm]
		&\textsc{RIf}\quad\frac{\left\{\phi\ast M_m\right\}\prog_m\{\psi\}\ {\rm for\ all}\ m\quad \psi\in \text{CM}}{\{\phi\ast\bD(\qbar)\}\mathbf{if}\cdots\mathbf{fi}\{\psi\}}\\[0.1cm]
		&\textsc{RLoop}\quad\frac{\{\phi\ast M_1\}\prog\{\phi\ast\bD(\qbar)\}\quad\phi\in\text{CM}}{\{\phi\ast\bD(\qbar)\}\mathbf{while}\{\phi\wedge M_0\}} 
		\end{split}\end{equation*}
		\caption{Inference Rules for Program Constructs. In \textsc{Init} and \textsc{Unit}, $\Mexist$ means the existence of modification. In \textsc{RIf} and \textsc{RLoop}, $\mathbf{if}\cdots\mathbf{fi}$ and $\mathbf{while}$ are abbreviations of $\mathbf{if}\ (\square m\cdot M[\qbar] = m \rightarrow \prog_m )\ \mathbf{fi}$ and $\mathbf{while}\ M[\qbar]=1\ \mathbf{do}\ \prog\ \mathbf{od}$ respectively, and  $M_0,M_1,M_m$ in assertions are regarded as projective predicates acting on $\qbar$. In \textsc{Perm}, $\perm(\qbar\mapsto\qbar^\prime)[\qbar]$ stands for the unitary transformation which permutes the variables from $\qbar$ to $\qbar^\prime$ (see Section \ref{sec basic Quantum} for details).
		}
		\label{fig proof system 1}
	\end{figure}

	The first set of our inference rules are designed for reasoning about basic quantum program constructs and displayed in Fig. \ref{fig proof system 1}. Some of them deserve careful explanations:
	
	\vspace{0.2cm}
	
    \noindent $\ \bullet\, $ \textbf{Rules} \textsc{Init} \textbf{and} \textsc{Unit}: With
      the definition of modification of BI formulas and Proposition \ref{pro
      modification} in mind, the rules \textsc{Init} and \textsc{Unit} are
      similar to the (backwards) inference rule $\{\phi[e/x]\}x:=e\{\phi\}$ for
      assignment in classical program logics.
		
		
		

	\vspace{0.2cm}
	
    \noindent$\ \bullet\, $ \textbf{Rules} \textsc{RIf} \textbf{and} \textsc{RLoop}:
      These two rules use the separating conjunction to perform reasoning about
      different execution paths. 
      Note that condition $\phi\in {\rm CM}$ is imposed in the premises of the rules \textsc{RIf} \textbf{and} \textsc{RLoop}. 
      
The set ${\rm CM}$ of assertions is formally defined as follows:
	\begin{definition}
		\label{def CM}
			A formula $\phi$ is closed under mixtures (CM), written $\phi\in{\rm CM}$, if for any $\rho,\rho^\prime$, whenever  $\dom{\rho}=\dom{\rho^\prime}$, $\rho\models\phi$ and $\rho^\prime\models\phi$, we have:
			$\forall\ \lambda\in[0,1],\ \lambda\rho+(1-\lambda)\rho^\prime\models\phi.$
		\end{definition}
		\begin{example}
			For two projections $P_0 = |0\rangle\langle0|$ and $P_1 = |1\rangle\langle1|$, $P_0\wedge P_1$ is ${\rm CM}$, but $P_0\vee P_1$ is not ${\rm CM}$ (both states $|0\rangle\langle0|$ and $|1\rangle\langle1|$ satisfies $P_0\vee P_1$, but their affine combination $\frac{I}{2} = \frac{|0\rangle\langle0|+|1\rangle\langle1|}{2}$ does not satisfy $P_0$ nor $P_1$ and thus does not satisfy $P_0\vee P_1$).
		\end{example}
		To see why  the  condition $\phi\in {\rm CM}$ necessary, we note that a quantum program can be executed in different paths with non-zero probabilities, and its semantic function maps the input to a weighted summation of the outputs from different execution paths. The condition $\phi\in {\rm CM}$ is introduced so that satisfaction relation is preserved under affine combination. 
		The following proposition identifies a class of formulas closed under mixture. 
		
		\begin{proposition}
			\label{prop CM}
			The formulas generated by following grammar are CM:
			$$
			\phi,\psi ::= p\in\AP\ |\ \top\ |\ \bot\ |\ \phi\wedge\psi\ |\ \unia[S]\ast\phi
			$$
		\end{proposition}
		
		We need to pay special attention on the application of 
		separating conjunctions $\ast$ in \textsc{Rif} and \textsc{RLoop}. Since the
    quantum measurement in the guards of {\bf if}-statements and {\bf while}
    loops may change the quantum state, we hereby consider a special kind of
    inputs that satisfying $\phi \ast \id_\qbar$. Thus the subsystem being
    measured is uncorrelated to the part of the state described by $\phi$, which ensures that the post-measurement state still satisfies $\phi$. In \textsc{RLoop}, although $\phi\ast M_0$ is satisfied for each path, it does not belong to CM in general. Thus, only a weaker postcondition $\phi\wedge M_0\in {\rm CM}$ can be achieved. 

	%
	%

	\subsection{Structural rules}
	
	\begin{figure}\centering
		\begin{equation*}\begin{split}
		&\textsc{Weak}\quad  \frac{\phi\gimp\phi^\prime\quad \{\phi^{\prime}\}\prog\{\psi^{\prime}\}\quad 
			\psi^{\prime}\gimp\psi  }{\{\phi\}\prog\{\psi\}} \\
		&\textsc{Conj}\quad \frac{\{\phi_1\}\prog\{\psi_1\}\quad\{\phi_2\}\prog\{\psi_2\}}{\{\phi_1\wedge \phi_2\}\prog\{\psi_1\wedge \psi_2\}} \\[0.1cm]
		&\textsc{Disj}\quad \frac{\{\phi_1\}\prog\{\psi_1\}\quad\{\phi_2\}\prog\{\psi_2\}}{\{\phi_1\vee\phi_2\}\prog\{\psi_1\vee\psi_2\}} \\
		&\textsc{Const}\quad \frac{\{\phi\}\prog\{\psi\}\quad\free{\mu}\cap\var(\prog)=\emptyset}{\{\phi\wedge\mu\}\prog\{\psi\wedge\mu\}} \\[0.1cm]
		&\textsc{Frame}\quad \frac{\begin{split}\{\phi\}\prog\{\psi\}&\quad\free{\mu}\cap\var(\prog)=\emptyset\\ \free{\psi}\cup\var&(\prog)\subseteq\free{\phi}\text{\ or\ }\psi\in{\rm SP} \end{split}}{\{\phi\ast\mu\}\prog\{\psi\ast\mu\}} 
		\end{split}\end{equation*}
		\caption{Structural Rules. Since $\gimp$ is strictly weaker than $\rightarrow$, \textsc{Weak} is stronger than ordinary weak rule.
		}\label{fig proof system 2}
	\end{figure}
	
	The second set of rules consists of the structural rules, presented in Fig. \ref{fig proof system 2}. 
	The rules \textsc{Conj} and \textsc{Disj} are similar to their counterparts in
  classical program logics. 	To explain the other rules,
let us fist define the global implication:		
		\begin{definition}[Global implication]
			For any BI formulas $\phi,\psi$, the global implication $\phi\gimp\psi$ is defined as the abbreviation of $\bD[\free{\phi}\cup\free{\psi}]\rightarrow(\phi\rightarrow\psi)$.
		\end{definition}
		Trivially, $\gimp$ is strictly weaker than $\rightarrow$. The difference is that, $\phi\gimp\psi$ is already enough to ensure that for any state $\rho$ with $\dom{\rho}\supseteq\free{\phi}\cup\free{\psi}$, $\rho\models\phi$ implies $\rho\models\psi$. For example, we have following proposition:
		\begin{proposition}
		    \label{prop glo imp}
				For all $\phi\in{\rm Res}$ and $S\subseteq \vars$, it holds that $\models \phi\gimplr \phi\wedge \bD[S].$	
		\end{proposition}
    
    Now we are ready to carefully examine the remaining rules in Fig. \ref{fig proof system 2}.  
	
	\vspace{0.2cm}
	
	\noindent$\bullet\ $\textbf{Rules} \textsc{Weak}:
 This rule is also similar to its counterpart in classical program logics, but there is a subtle difference between them.   Since only global states (i.e. the states whose domain
    contains all free variables appearing in the assertions and programs) are 
    considered in  defining the validity of the Hoare triple, we use  $\gimp$ in the premise of the \textsc{Weak} rule for comparing assertions. It is easy to see that the rule is also sound when using $\rightarrow$, but the  \textsc{Weak} rule with $\gimp$ is stronger. 


	
	
	\vspace{0.2cm}
	
	\noindent$\bullet\ $\textbf{Rules} \textsc{Const}: This rule states that if any
      variable appearing in program $\prog$ is not free in $\mu$, then $\mu$ is preserved and thus can be conjoined to the pre- and post-conditions. 
		The principle behind is that $\mu$ is restrictive, i.e., the satisfaction of $\mu$ depends only on the reduced state over subsystem $\free{\mu}$, which trivially remains unchanged after executing $\prog$. 
		An interesting application of this rule is proving  product predicates from local reasoning using Proposition \ref{prop axiom projection}.
	
	\vspace{0.2cm}
	
	\noindent$\bullet\ $\noindent\textbf{Rules} \textsc{Frame}: 	
		The condition $\free{\mu}\cap\var(\prog)=\emptyset$ in the premise ensures that $\mu$ can be conjoined with the pre- and post-conditions. The condition  $\free{\psi}\cup\var(\prog)\subseteq\free{\phi}$ guarantees that, if the input satisfies $\phi\ast\mu$, which asserts that subsystems $\free{\phi}$ and $\free{\mu}$ are uncorrelated, then after executing $\prog$, these two subsystems are still independent since $\var(\prog)\subseteq\free{\phi}$, and furthermore, by the downward closed property of independence, subsystems $\free{\psi}$ and $\free{\mu}$ are uncorrelated as $\free{\psi}\subseteq\free{\phi}$.
		It is particularly interesting to note that the latter condition can be altered by  $\psi\in{\rm SP}$ defined in the following:
		\begin{definition}[Supported Assertion, c.f. \cite{Rey08}]
			\label{def SP}
			A formula $\psi$ is called supported, written $\psi\in{\rm SP}$, if $\sem{\psi}$ is nonempty then it has a least element, or equivalently, there exists a $S\subseteq\vars$ such that 1. at most one $\rho\in\cD(S)$ satisfies $\psi$ and 2. if $\sigma\models\psi$, $\sigma\succeq\rho$.
		\end{definition}
		
		Trivially, any uniformity proposition $\unia[S]$ and any atomic proposition
    defined by a projection of rank 1 are in SP; more examples of SP are given
    in the  Supplementary Material \ref{sec app sub proof prop SP}. The frame rule with SP condition is nontrivial
    and it will be very useful in our later case studies on verification of
    quantum information-theoretic security; indeed, this application uncovered
    the condition $\psi\in{\rm SP}$. 
    Note that under this condition, the frame rule is sound even without any restriction on $\free{\psi}$, $\free{\phi}$ and $\var(\prog)$. This seems counter-intuitive; but in fact, the premise $\{\phi\}\prog\{\psi\}$ is much stronger than it looks at first sight, given that the postcondition $\psi\in{\rm SP}$. If the input satisfies precondition $\phi$, then an execution of $\prog$ is almost equivalent to first erasing any information on subsystem $\free{\psi}$ (of course, it is now uncorrelated with the rest part of the whole system), and then regenerating the singleton that satisfies the postcondition $\psi$.
	
	\subsection{Reasoning about Entangled Predicates}
	\label{sec Local Reasoning Entangle}
	
	Many quantum algorithms are designed following the same pattern: start from a
  large entangled state, and then operate on various subsystems. Inevitably,
  entanglements often appear in the preconditions and/or postconditions of Hoare
  triples appropriate for specifying the correctness of these algorithms.
  But the frame rule itself is not strong enough to verify them. To see this more clearly, let us consider the following simple example: 
  
 \begin{example}
 \label{exam verify entanglement}
 Let $|\Phi^\pm\> = \frac{1}{\sqrt{2}}(|00\>\pm|11\>)$ be two 
 Bell states (entanglement). Define projections  $\Phi^\pm = |\Phi^\pm\>\<\Phi^\pm|$. The program $\prog\equiv {\sqrt{Z}}[q_1]; {\sqrt{Z}}[q_2]$ transforms one Bell states to the other; that is, both $\{\Phi^+\}\prog\{\Phi^-\}$ and $\{\Phi^-\}\prog\{\Phi^+\}$ are true. However, they cannot be proved by using \textsc{Frame} or \textsc{Const} to lift local correctness of ${\sqrt{Z}}[q_1]$ and ${\sqrt{Z}}[q_2]$ to global predicates $\Phi^\pm$, since $\Phi^\pm$ cannot be written in the form of $\Phi^\pm \not\equiv \phi_{q_1}\ast\psi_{q_2}$ or $\Phi^\pm \not\equiv \phi_{q_1}\wedge\psi_{q_2}$. \end{example}
	
	Fortunately, our frame rule  can be combined with a technique for reasoning about entangled predicates proposed in \cite{YZL18} to handle this problem. Originally, this technique was introduced for parallel quantum programs. Here, we need to reformulate it in a way convenient for our purpose. 
	A combination of this technique with the frame rule can significantly broaden the range of applications of our quantum separation logic. 
	To this end, we need to generalise Definition \ref{def sub atomic prop} from
  modification by a unitary transformation, and initialisation to modification
  by a general quantum operation.\footnote{Quantum operation is used to describe
    the evolution of a (open) quantum system and can be characterized by an
    superoperator $\cE$, namely a completely-positive and trace-non-increasing
    linear map from $\cD$ to $\cD$.  For every superoperator $\cE$, there exists
  a set of Kraus operators $\{E_i\}_i$ (linear operators that satisfy completeness condition $\sum_iE_i^\dag E_i = I$) such that $\cE(\rho) = \sum_i E_i\rho
E_i^\dag$ for any input $\rho$.}
	\begin{definition}[$\cE$-Modification]
		\label{def qo modification}
		Let $\cE$ be quantum operation on $\qbar$. The $\cE$-Modification $\phi[\cE[\qbar]]$ acting on register $\qbar$ of a BI formula $\phi$ is defined inductively:
		\begin{enumerate}
			\item (Atomic Proposition) For any $P\in\cP$, we have:\footnotemark
			\begin{enumerate}
			    \item if $\qbar\subseteq\free{P}$, $$P[\cE[\qbar]]\triangleq\left(\big(\cE^\ast_{\qbar}\otimes\cI_{\free{P}\backslash\qbar}\big)(P^\bot)\right)^\bot;$$
			    \item if $\qbar\cap\free{P}=\emptyset$, $P[\cE[\qbar]]\triangleq P$;
			    \item otherwise, $P[\cE[\qbar]]$ is undefined;
			\end{enumerate}
			\footnotetext{Here $\bot$ stands for the ortho-complement, for not only the projections but Hermitian operators, in the sense that $A^\bot = \spa\{|\psi\>\in\cH_{\free{A}}: A|\psi\> = 0\}$. $\cE^\ast$ is dual of	$\cE$; in detail, $\cE^\ast(A) = \sum_iE_i^\dag A E_i$ if $\cE$ has the operator-sum representation $\cE(\rho) = \sum_iE_i \rho E_i^\dag$.}
			\item (Composite) Write $\phi[\cE[\qbar]]\Mexist$ if $\phi[\cE[\qbar]]$ is defined.
			\begin{enumerate}			
				\item if $\phi\equiv \top$ or $\bot$, then $\phi[\cE[\qbar]] \triangleq \phi$;
				\item if $\phi\equiv p\in\AP$, then $\phi[\cE[\qbar]]$ is defined according to Clause {\rm (1)};
				\item if $\phi\equiv \phi_1\ \triangle\ \phi_2$ where $\triangle\in\{ \wedge,\vee\}$ and both $\phi_1[\cE[\qbar]]\Mexist$ and $\phi_2[\cE[\qbar]]\Mexist$, then $\phi[\cE[\qbar]]\triangleq\phi_1[\cE[\qbar]] \ \triangle\ \phi_2[\cE[\qbar]]$
				\item otherwise, $\phi[\cE[\qbar]]$ is undefined.
			\end{enumerate}
		\end{enumerate}
	\end{definition}
	Intuitively, if $\phi[\cE[\qbar]]\Mexist$, then for any state $\rho$, $\cE(\rho)\models\phi$ if and only if $\rho\models \phi[\cE[\qbar]]$.

	\begin{figure}\centering
		\begin{equation*}\begin{split}
		&\textsc{UnCR}\quad \frac{\{\phi\}\prog\{\psi\}\ \ \qbar\cap\var(\prog)=\emptyset\ \ \phi[\cE[\qbar]]\Mexist\ \psi[\cE[\qbar]]\Mexist}{\{\phi[\cE[\qbar]]\}\prog\{\psi[\cE[\qbar]]\}} 		
		\end{split}
		\end{equation*}
		\caption{Proof rule for dealing with entangled predicates. $\Mexist$ means the existence of modification.
		}\label{fig proof system 3}
	\end{figure}
	
	Now we can introduce a new inference rule \textsc{UnCR} (stands for ``uncorrelated'') in Fig. \ref{fig proof system 3}. 
	This rule plays an essential role in the verification of VQA (see Section \ref{sec-vqa-exam}). We divide VQA into several pieces and reason locally, but the global predicate we desired is an entangled predicate that cannot be constructed using \textsc{Frame}. \textsc{UnCR} is the bridge for structural reasoning from local to global predicates. In addition, a formal verification of Example \ref{exam verify entanglement} using \textsc{UnCR} can be found in Supplementary Material \ref{sec app sub verification exam verify entanglement}.


	\subsection{Soundness}
	
	To conclude this section, we show that quantum separation logic QSL consisting
  of all the proof rules listed in Figure \ref{fig proof system 1}--\ref{fig
  proof system 3} are sound. The detailed proof can be found in the Supplementary Material \ref{sec app sub proof thm sound QSL}.
	
	\begin{theorem}[Soundness of QSL]
		\label{thm sound QSL}
    A program $\prog$ is \emph{almost surely terminating} if for all inputs
    $\rho$, $\tr(\sem{\prog}(\rho)) = \tr(\rho)$. If $\prog$ is a most surely
    terminating program, then $\vdash\{\phi\}\prog\{\psi\}\ {\rm implies}\
    \models\{\phi\}\prog\{\psi\}.$
	\end{theorem}	
	
	\section{Local Reasoning: Analysis of Variational Quantum Algorithms}\label{sec-vqa-exam}
	
	From now on we present a couple of examples to demonstrate  applicability of our quantum separation logic. 
	Variational quantum algorithms (VQA) are a class of hybrid quantum/classical algorithms solving a fundamental problem in quantum chemistry -- determine the ground state of a quantum system \cite{PMS14,MRB16}. It has been identified as one of the first practical applications of  near-term Noisy Intermediate Scale Quantum (NISQ) computers \cite{Pre18}, and thus were chosen  as an example in the tutorials of several quantum programming platforms including Google's Cirq~\cite{Cirq}. Surprisingly, using the inference rules presented in the last section, we are able to show that the implementation of VQA in the tutorial of Cirq is actually incorrect;  that is, the approximation of ground energy computed by the quantum circuit given there is sometimes far from the real one.
	
	\subsection{Variational Quantum Algorithm (VQA)}
	A typical VQA uses a hybrid computing system consisting of a QPU (quantum processing unit) and CPU to find a good approximation of the ground energy and ground state of a given Hamiltonian of the form:
	$$H = \sum_{i,\alpha}h_\alpha^i\sigma_\alpha^i + \sum_{i,j,\alpha,\beta}h_{\alpha\beta}^{ij}\sigma_\alpha^i\sigma_\beta^j+\cdots$$
	where $h$'s are real numbers, and superscripts $i,j,\cdots$ identify the subsystem and subscripts $\alpha,\beta,\cdots\in\{x,y,z\}$ indicate the appropriate Pauli operators. The algorithm can be described in four steps:
	\begin{enumerate}
		\item Define a set of ansatz states $|f(\boldsymbol{\theta})\>$, which are characterized by  parameters $\boldsymbol{\theta} = (\theta_1,\theta_2,\cdots,\theta_n)$ and can be efficiently prepared by a quantum circuit $\prog(\boldsymbol{\theta})$. The goal of the algorithm is to find the optimal parameters $\boldsymbol{\theta}_{\min}$ which minimize the energy $\<f(\boldsymbol{\theta})|H|f(\boldsymbol{\theta})\>$. Then $\<f(\boldsymbol{\theta}_{\min})|H|f(\boldsymbol{\theta}_{\min})\>$ and $|f(\boldsymbol{\theta}_{\min})\>$ can be set as an approximation of the ground energy and ground state, respectively.
		
		\item Use the QPU to execute the quantum computation represented as quantum circuit $\prog(\boldsymbol{\theta})$ in order to generate state $|f(\boldsymbol{\theta})\>$ and compute the expectations of $\sigma_\alpha^i,\ \sigma_\alpha^i\sigma_\beta^j\ ,\cdots$ in all the terms of $H$;
		
		\item Use the CPU to sum up the expectations of all the terms of $H$ with the weights $h$'s and thus evaluate $\<f(\boldsymbol{\theta})|H|f(\boldsymbol{\theta})\>$;
		
		\item Feed $\<f(\boldsymbol{\theta})|H|f(\boldsymbol{\theta})\>$ to an classical minimization algorithm. If the optimization is not completed, prepare the parameters $\boldsymbol{\theta}$ for the next round and go to step (2); otherwise, terminate and return $\boldsymbol{\theta}$ as output.
	\end{enumerate}
	
	\subsection{VQA in the Tutorial of Cirq}
	
	The VQA presented in the tutorial of Google's Cirq \footnote{\url{https://quantumai.google/cirq/tutorials/variational_algorithm}} deals with a 2D $+/-$ Ising model of size $N\times N$ with objective Hamiltonian (observable)
	$$H = \sum_{(i,j)}h_{ij}Z_{ij} + \sum_{(i,j;i^\prime,j^\prime)\in S}J_{ij;i^\prime j^\prime}Z_{ij}Z_{i^\prime j^\prime},$$
	where each index pair $(i,j)$ is associated with a vertex in a the $N\times N$ grid, $S$ is the set of all neighboring vertices in the grid, and all $h_{ij}$ and $J_{ij;i^\prime j^\prime}$ are either $+1$ or $-1$.
	The algorithm for preparing the ansatz state with real parameters $(\alpha,\beta,\gamma)$ given in the tutorial of Cirq can be rewritten in the quantum-{\bf while} language with $N\times N$ grid of qubits as follows:
	\begin{align*}
	{\rm VQA}(N)\equiv\ &{\bf for\ } j = 1,\cdots,N {\bf\ do}\ {\rm ProcC}(j)\ {\bf od}; \\
	&{\bf for\ } i = 1,\cdots,N {\bf\ do}\ {\rm ProcR}(i)\ {\bf od}.
	\end{align*}
	Here, subprogram ${\rm ProcC}(j)$ acts on the $j$th column of qubits and ${\rm ProcR}(i)$ acts on the $i$th row of qubits; each of them is a sequential composition of unitary transformations (see the Supplementary Material \ref{sec app sub VQA in the Tutorial of Cirq} for detailed subprograms).
	
	\subsection{Specifying and Proving Incorrectness in Quantum Separation Logic}
	\label{sec: VQA in QSL}
	As pointed out at the beginning of this section, we can use our quantum separation logic to show that algorithm ${\rm VQA}(N)$ is indeed incorrect. Let us first describe its incorrectness in our logical language.  
	Suppose the Hamiltonian $H$ has eigenvalues $E_0,E_1,...$ ranged in increasing order, with corresponding eigenspaces (projections) $Q_0, Q_1...$. 
	If for each $i\leq n$, we can find a  precondition $P_i\in\cP$ such that $\models\{P_i\} {\rm VQA}(N) \{1-\sum_{k=0}^{i}Q_i\}\ (i=0,1,...,n),$ then by showing that $|0\>$ (the initial state of quantum circuit) is close to $P_i$; that is, $\<0|P_i|0\> \ge \delta_i$, we can conclude that the approximate ground energy computed by ${\rm VQA}(N)$ is at least: 
	\begin{equation}\label{VQA-meters}E_0 + \sum_{i=1}^{n}(E_{i+1}-E_{i})\delta_i.\end{equation}
	Therefore, whenever the quantity in (\ref{VQA-meters}) is far away from the real ground energy $E_0$, then ${\rm VQA}(N)$ is incorrect.
	
	
	To illustrate our idea more explicitly, let us consider the simplest case of $2\times 2$ grid ($N = 2$) with parameters:
	$$h = {\small\left[\begin{array}{cc}-1 & -1 \\ 1 & 1\end{array}\right]},\quad Jc = {\small\left[\begin{array}{c}-1 \\ -1\end{array}\right]},\quad Jr = \left[\begin{array}{cc}-1 & 1\end{array}\right]$$
	and $J_{ij;(i+1)j} = Jr_{ij}$ and $J_{ij;i(j+1)} = Jc_{ij}$; see Fig. \ref{fig VQA2} for its circuit model.
	The eigenvalues of the Hamiltonian $H$ in this case are $E_0,\cdots,E_5 =
  -6,-4,-2,0,2,4$ with corresponding eigenspaces $Q_0,Q_1,\cdots,Q_5$,
  respectively.  Using QSL, we are able to prove:
  \[
    \vdash \{P_i\} {\rm VQA}(N) \left\{1-\sum_{k=0}^{i}Q_i\right\}\ (i=0,1)
  \]
  where
  \begin{align*}
  &\<0|P_0|0\> = 1-\frac{1}{16}\sin(\alpha\pi)^4 \ge \frac{15}{16}\\
  &\<0|P_1|0\> = 1-\frac{1}{32}(7+\cos(2\alpha\pi))\sin^2(\alpha\pi)\ge \frac{13}{16},
  \end{align*}
	by first reasoning about each subprogram ${\rm ProcC}(1)$, ${\rm ProcC}(2)$, ${\rm Proc(R)}(1)$, ${\rm ProcR}(2)$ and then using \textsc{Const} and \textsc{UnCR} to lift these local reasoning to global correctness above (details can be found in Supplementary Material \ref{sec app sub Verifying Incorrectness in QSL}).
	Then it follows from (\ref{VQA-meters}) that the approximate ground energy of VQA is at least $-2.5$, which is much higher than the real ground energy $E_0 = -6$.

  Our quantum separation logic can also apply to higher dimensionional versions
  of this program. In general, since the number of qubits in each subprogram of
  VQA is $\frac{1}{N}$ of that of the entire system, there is no extra cost for
  local reasoning no matter how large $N$ is. 
	Besides revealing the incorrectness of ground energy, we can prove that parameters $\beta,\gamma$ are helpless for finding the ground energy 
	in the sense that the expectation of measurement outcome\footnote{The QPU executes ${\rm VQA}(N)$ and then measures each qubit in computational basis and feed the outcome to CPU.} is independent of $\beta,\gamma$.
	
	\section{Scalable Reasoning: Verification of Security}\label{sec-security-exam}
	
	A major distinction between classical and quantum information can be stated as the no-cloning theorem that it is impossible to create an identical copy of an arbitrary unknown quantum state. Exploiting this fundamental property among others, many quantum cryptographic protocols with information-theoretical security have been proposed, including quantum key distribution, quantum one-time pad \cite{MTW00,BR03} and quantum secret sharing \cite{HBB99, CGL99}.
	
	In this section, we show how quantum separation logic developed in this paper can be used to verify the security of quantum one-time pad and quantum secret sharing. In particular, such verification is scalable in the sense that only a constant computational resource is required in the verification as the length of protocols and the involved qubits increase.
	
	Uniformity is essential in proving the information-theoretical security of many quantum cryptographic protocols. For convenience, let us first present a useful rule: 
	\begin{equation}
	\label{fig proof system uniform}
	\textsc{FrameU}\quad \frac{\{\top\}\prog\{\unia[S_1]\}\quad S_2\cap(\var(\prog)\cup S_1)=\emptyset}{\{\unia[S_2]\}\prog\{\unia[S_1,S_2]\}}.
	\end{equation}
	This rule is derived by instantiating $\phi\equiv\top$,  $\psi\equiv\unia[S_1]$ and $\mu\equiv\unia[S_2]$ in the frame rule \textsc{Frame} and using axiom scheme (see Proposition \ref{prop axiom projection} (4)).
	
	\subsection{Security of Quantum One-Time Pad}
	
	Let us first verify the security of quantum one-time pad (QOTP) \cite{MTW00, BR03}, one of the basic quantum encryption schemes in quantum cryptography. Similar to the classical one-time pad, a one-time pre-shared secret key is employed to encrypt and decrypt the quantum data.
	
	\subsubsection{Single-Qubit Case} To warm up, we consider the simplest case for protecting one-qubit data.   
	The QOTP scheme consists of three parts: key generation $\mathbf{KeyGen}$, encryption $\mathbf{Enc}$ and decryption $\mathbf{Dec}$, which can be written as programs:
	{\small
		\begin{align*}
		&\mathbf{KeyGen}[a,b] \equiv  a := |0\>; b := |0\>; \ 
		a := H[a];  b := H[b];  \\
		&\qquad\qquad\qquad\quad\ \, \mathbf{if}\ \cM[a,b]=00\rightarrow \mathbf{skip}\ \square\ \ \! 01\rightarrow \mathbf{skip}\ \\
		&\qquad\qquad\qquad\quad\ \, \qquad\qquad\  \square\ \ \!  10\rightarrow \mathbf{skip}\ \square\ \ \!  11\rightarrow \mathbf{skip} \quad \mathbf{fi}\\ 
		&\mathbf{Enc}[a,b,q] \equiv\  \mathbf{if}\ \cM[a,b]=00\rightarrow \mathbf{skip}\ \ \ \square \  01\rightarrow q = Z[q]\\
		&\qquad\qquad\qquad\ \  \square\ 10\rightarrow q = X[q]\ \square\ 11\rightarrow q = Z[q];q = X[q]\  \mathbf{fi} \\
		&{\rm QOTP}[a,b,q] \equiv  \mathbf{KeyGen}[a,b]; \mathbf{Enc}[a,b,q]
		\end{align*}
	}\\[-0.8cm]
	
	\noindent Here, registers $a$ and $b$ are used as the secret key, and measurement 
    \begin{align*}
        \cM = \{&M_{00} = |00\>_{ab}\<00|, M_{01} = |01\>_{ab}\<01|, \\
	            &M_{10} = |10\>_{ab}\<10|, M_{11} = |11\>_{ab}\<11|\}
    \end{align*}
	is introduced to generate and detect the value of secret key, which returns a two-bit classical outcome with a certain probability. Register $q$ is the input quantum data which we want to protect. $H$ is the Hadamard gate and $X,Z$ are Pauli gates as usual. 
	
	Security of QOTP for the single-qubit case can be specified as the following uniformity: 
	\begin{equation}\vdash\label{q-one-pad-1}\{\top\}{\rm QOTP}[a,b,q]\{\unia[q]\}.\end{equation}
	This fact has been formally verified using quantum Hoare logic with ghost variables \cite{Unr19b} and relational quantum Hoare logic in \cite{BHY19}. 
	
	\subsubsection{Multi-Qubit Case -- Scaling Up} Now we show how can the verification for single-qubit be easily scaled up to the multi-qubit case using the frame rule in our quantum separation logic.  
	The protocol for protecting $n$-qubit data stored in register $\qbar = q_1,\cdots,q_n$ can be written as:
	$${\rm QOTP}(n)\equiv {\bf for}\ i=1,\cdots, n\ {\bf do}\ {\rm QOTP}[a_i,b_i,q_i]\ {\bf od}$$
	where $a_1,b_1,\cdots,a_n,b_n$ are secret key of size $2n$. Its security can be stated as the following uniformity:
	\begin{equation}\vdash\label{q-one-pad-n}\{\top\}{\rm QOTP}(n)\{\unia[q_1,\cdots,q_n]\},\end{equation}
	which shows that, no matter what is the plain text initialised on $\qbar$, after encryption, the cipher text is always uniform and the eavesdropper cannot release any useful information.
	This judgment is proved as follows. First, it follows from (\ref{q-one-pad-1})
  that 
  $$\vdash \{\top\}{\rm QOTP}[a_i,b_i,q_i]\{\unia[q_i]\}\quad (i=1,...,n).$$
  Using \textsc{FrameU} we obtain for all $i=1,...,n$:
	$$\vdash\{\unia[q_1,\cdots,q_{i-1}]\}{\rm QOTP}[a_i,b_i,q_i]\{\unia[q_1,\cdots,q_i]\}$$ Then (\ref{q-one-pad-n}) is derived by repeatedly using rule \textsc{Seq}. 
	
	\subsubsection{Discussion} A comparison between the security verification of QOTP in quantum Hoare logic \cite{Unr19b, BHY19} and in quantum separation logic presented above is interesting. Only the single-qubit case was considered in \cite{Unr19b}. A crucial step in the verification for the multi-qubit case given in \cite{BHY19} is based on a complicated transformation of quantum predicates, which cannot be proved by the logic itself, but is derived from a mathematical result proved by quite  involved calculations in the previous literature \cite{MTW00}. 
	In contrast, the verification in quantum separation logic avoids such complicated  calculations by using the frame rule \textsc{FrameU}.
	
	\subsection{Security of Quantum Secret Sharing}
	
  Now we turn to verify the security of another quantum cryptographic protocol:
  quantum security sharing. Similar to classical secret sharing
  \cite{Bla79,Sha79}, quantum secret sharing addresses the problem of how to
  distribute a secret amongst a group of participants so that the secret can be
  reconstructed by a sufficient number of participants while any individual has
  no information about it \cite{HBB99, CGL99}. For concreteness, let us focus on
  a typical scheme.
	
	\subsubsection{Quantum $(2,3)$ Threshold Scheme} 
  The $(2,3)$ threshold scheme for sharing a single secret \emph{qutrit} $p$ (a
  $3$-dimensional quantum state) takes $p$ as the input and outputs three
  qutrits $p^\prime,q^\prime,r^\prime$ so that each of them has no information
  about the input secret while any two of them can recover the input.  Formally,
  it can be written as the following program: 
	\begin{equation*}
	{\bf Enc}[p,q,r]\equiv q:=|0\>;\ r:=|0\>;\ p,q,r := U_{\rm enc}[p,q,r]
	\end{equation*}
	where unitary transformation $U_{\rm enc}$ maps $|i\>|0\>|0\>$ to $|e_i\>$
	for $i = 0,1,2$, where $|e_i\>$ are three orthonormal states:
	\begin{align*}
	|e_i\> = \frac{1}{\sqrt{3}}\sum_{k=0}^2|k\>|k\oplus_3 i\>|k\oplus_3 2i\>
	\end{align*}
	where $\oplus_3$ stands for the addition modulo 3. For secretly sharing information of multiple qutrits $\ol{p} =
  p_1,\cdots,p_n$, this scheme can simply be  generalised to:
	\begin{equation*}
    {\rm QSS}(n)\equiv {\bf for}\ i=1,\cdots, n\ {\bf do}\ {\bf Enc}[p_i,q_i,r_i]\ {\bf od}.
	\end{equation*}

	
	\subsubsection{Security as Uniformity}
	\label{sec QSS sec uni}
	Quantum secret sharing is designed for against both dishonest agents and eavesdroppers \cite{HBB99,KKI99,CGL99}. Let us first consider the case without any eavesdropper during transmission. In this case, the security of ${\rm QSS}(n)$ can be specified as the following judgment:
	\begin{equation}\label{2-3-security}\vdash\{\top\}{\rm QSS}(n)\{\unia[q_1,\cdots,q_n]\}.\end{equation}
	The above judgment can be easily proved in our quantum separation logic. First, using rules \textsc{Unit}, \textsc{Init} and \textsc{Seq} directly we obtain:
	\begin{equation}\label{eqn QSS single}\vdash\{\top\}{\bf Enc}[p,q,r]\{P_S[p,q,r]\},\end{equation}
	where projection 
	$P_S =  |e_0\>\<e_0|+|e_1\>\<e_1|+|e_2\>\<e_2|.$
	It is easy to check that
	$\models P_S[p,q,r]\rightarrow (\unia[p]\wedge \unia[q]\wedge \unia[r]).$ 
	Based on this we can conclude:  $$\vdash\{\top\}{\bf Enc}[p,q,r]\{\unia[\alpha]\}$$ for $\alpha\in\{p,q,r\}$. 
	This proves the security for the case of a 
	
	\noindent single qutrit. To generalise it to the case of multiple qutrits, we can use \textsc{FrameU} to derive:
	$$\vdash\{\unia[q_1,\cdots,q_{i-1}]\}{\bf Enc}[p_i,q_i,r_i]\{\unia[q_1,\cdots,q_i]\}$$
	from $\vdash\{\top\}{\bf Enc}[p_i,q_i,r_i]\{\unia[q_i]\}.$ Then by setting
  formulas $\phi_i = \unia[q_1,\cdots,q_{i-1}]$ and $\phi_1 = \top$,  we have
  $\vdash\{\phi_i\}{\bf Enc}[p_i,q_i,r_i]\{\phi_{i+1}\}$ for all $1\le i\le n$,
  and (\ref{2-3-security}) is obtained by repeatedly using rule \textsc{Seq}.

	\section{Discussion and Related Work}
	\label{sec:discussion}
	
	\vspace{-0.05cm}
	
	In this section, we briefly discuss an issue about restriction property left open in Subsection \ref{sec res} as well as some previous work on verification of quantum programs.  	
	
	\subsection{Restriction property and BI with domain}
	
	\vspace{-0.05cm}
	
	Our quantum interpretation  of standard BI logic is sufficient for the applications discussed in this paper. 
	However, it has a drawback: 
 the restriction property does not hold for all BI formulas, and thus the assertions in our QSL (Quantum Separation Logic) are confined in a special class of BI formulas (see Def. \ref{def Res}), which do not include implication and separating implication. 
 One possible solution to this issue is to redefine the BI logic so that the  restriction property becomes intrinsic -- similar to the monotonicity. We can introduce a notion of \emph{domain} into BI: the domain $\dom{x}$ of a \textit{state} $x$ is the set of variables specified by the state. Then a  basic idea in classical separation logic \cite{Bro07,IO01,ORY01,OHe07}, called the domain assumption for stack, can be adopted in defining  satisfaction relation: $x\models \phi$ is defined only when $\dom{x}\supseteq\free{\phi}$, where $\free{\phi} $ is the set of free variables in a BI-\textit{formula} $\phi$. The domain assumption  guarantees that  the restriction property is true even when the extension of joint quantum states does not exist (see Sec. \ref{sec res}). In this way, BI is upgraded to BID (BI with domain), and all BID formulas can be safely used as assertions in QSL. Details of this approach can be found in the  Supplementary Material \ref{sec 2BID}.
	
	\subsection{Related work}
	
	\vspace{-0.05cm}
	
	Quantum programming has become an active research field  in recent years after two decades of development \cite{HSM20}. Various analysis, verification, testing and debugging methodologies and techniques for quantum programs have been developed \cite{Ak05, DP06, BS04, BS06, BJ04, CMS06, Kak09, Ran16, FDJ07, YDF10, Ying16}. In particular, several quantum program logics have been established, including quantum Hoare logic~\cite{Ying11, Unr19b, ZYY19} for verifying correctness of one quantum program and relational quantum Hoare logic~\cite{Unr19a, BHY19, LU19} for verifying equivalence of two quantum programs.
The frame rule plays a key role in our QSL. We should mention that a frame rule
was also introduced in relational quantum Hoare logic  \cite{Unr19a, BHY19,
LU19}. But it was defined using the ordinary conjunction $\wedge$ and thus  is
similar to our \textsc{Const}. The frame rule in QSL is given using the
separating conjunction $\ast$. Of course, the intuitions behind them are the
same---an assertion is preserved by a program if it is independent of the program.
	
	The target applications of our SQL is verification of large-scale quantum programs, where the size of the representation   of assertions and the complexity of the involved  calculations can increase exponentially w.r.t the number of qubits. Two different approaches to this issue were proposed in \cite{HRS20} and \cite{Bor20}, essentially based  on the operational semantics. They have achieved obvious success, in particular for those  large-scale quantum programs with a good algebraic structure that can be  inductively defined. 
It seems that sometimes our QSL can be used in combination with  them; for example, some larger VQAs (Variational Quantum Algorithms) can be divided into several blocks,  each of which has a good algebraic structure and thus  can be verified using the tools developed in  \cite{HRS20,Bor20}. Then our QSL can be employed to lift these local reasoning to the 
global correctness of VQAs.

	\section{Conclusion}
	\label{sec conclusion}
	In this paper, we have developed a quantum separation logic QSL that enables local reasoning for scalable verification of quantum programs written in a simple quantum programming language, namely the quantum extension of \textbf{while}-language.	The applicability of QSL has been demonstrated in 
	the formal verification and analysis of several practical quantum algorithms and cryptographic protocols, including a VQA (Variational Quantum Algorithm), quantum one-time pad, and quantum secret sharing.
	
	There are several interesting topics for future research along this line:
	
	(1) We would like to explore more applications of our logic QSL in the verification of those algorithms identified as practical applications of near-term Noisy Intermediate Scale Quantum (NISQ) computers \cite{Pre18}; for example, quantum machine learning from quantum data. We will also try to apply QSL in the security analysis of more quantum cryptographic protocols rather than those considered in this paper, in particular QKD (Quantum Key Distribution).  
	
	(2) Currently, QSL can only be used to quantum \textbf{while}-programs
    without indexed variables, like arrays. However, indexed variables has
    already been frequently used in writing large quantum algorithms. We would
    like to extend our logic for a more sophisticated quantum program language
    with indexing.
	
	(3) Resource theory has been emerging as a subarea of quantum information theory in recent years. 
	Roughly speaking, it aims at understanding  how the resources  with  quantum advantage in computing and communication can be generated and transformed (e.g. only using LOCC (local operations and classical communication))\cite{HO13, Plenio2014, VMG14}.
  As briefly mentioned in the Introduction, some  connections between resource theory and the resource semantics of BI were  already noticed in   \cite{Doc19,CoeckeFS16,fritz_2017}.
	We would like to see how quantum separation logic can be used to reason about these quantum resources.
	

	\bibliographystyle{IEEEtran}
	\bibliography{main}
	
	
	\onecolumn
	
	{\centering\Large \textbf{Supplementary material and deferred proofs}}

	\setcounter{section}{0}
	
	\begin{appendices}
		
		\section{Preliminary}
		\label{sec basic Quantum app}
		
		In the main text we give a brief introduction of quantum information (see Section \ref{sec basic Quantum}). A extended introduction is given here for the convenience of reader.
		
		Quantum Information is built on the linear algebra. We first give the mathematical preliminary needed for understanding quantum information/computation.
		
		\subsection{Mathematical Preliminary}
		
		We write $\mathbb{C}$ for the set of complex numbers. For each
		complex number $\lambda\in \mathbb{C}$, $\lambda^{\ast}$ stands for
		the conjugate of $\lambda$. A (complex) vector space is a nonempty
		set $\mathcal{H}$ together with two operations: vector addition $+:
		\mathcal{H}\times \mathcal{H}\rightarrow \mathcal{H}$ and scalar
		multiplication $\cdot : \mathbb{C}\times \mathcal{H}\rightarrow
		\mathcal{H}$, satisfying the following conditions:
		\begin{enumerate}\item $(\mathcal{H},+)$ is an Abelian group, its zero element $\mathbf{0}$ is called the zero vector; \item
			$1|\varphi\rangle=|\varphi\rangle$; \item $\lambda (\mu
			|\varphi\rangle)=\lambda\mu |\varphi\rangle$; \item $(\lambda
			+\mu)|\varphi\rangle =\lambda |\varphi\rangle +\mu |\varphi\rangle$;
			and \item $\lambda (|\varphi\rangle +|\psi\rangle)=\lambda
			|\varphi\rangle + \lambda |\psi\rangle$ \end{enumerate} for any
		$\lambda,\mu\in\mathbb{C}$ and $|\varphi\rangle,
		|\psi\rangle\in\mathcal{H}$.
		
		An inner product over a vector space $\mathcal{H}$ is a mapping
		$\langle\cdot|\cdot\rangle:\mathcal{H}\times \mathcal{H}\rightarrow
		\mathbb{C}$ satisfying the following properties:
		\begin{enumerate}\item $\langle\varphi|\varphi\rangle\geq 0$ with
			equality if and only if $|\varphi\rangle =0$; \item
			$\langle\varphi|\psi\rangle=\langle\psi|\varphi\rangle^{\ast}$; and
			\item $\langle\varphi|(\lambda_1|\psi_1\rangle+\lambda_2|\psi_2\rangle)=
			\lambda_1\langle\varphi|\psi_1\rangle+\lambda_2\langle\varphi|\psi_2\rangle$\end{enumerate}
		for any $|\varphi\rangle, |\psi\rangle, |\psi_1\rangle,
		|\psi_2\rangle \in \mathcal{H}$ and for any $\lambda_1,\lambda_2\in
		\mathbb{C}$. Sometimes, we also write
		$(|\varphi\rangle,|\psi\rangle)$ for the inner product
		$\langle\varphi|\psi\rangle$ of $|\varphi\rangle$ and
		$|\psi\rangle$.
		
		For any vector $|\psi\rangle$ in $\mathcal{H}$, its length
		$||\psi||$ is defined to be $\sqrt{\langle\psi|\psi\rangle}$. A
		vector $|\psi\rangle$ is called a unit vector if $||\psi||=1$.
		A family $\{|\psi_i\rangle\}_{i\in I}$ of
		unit vectors is called an orthonormal basis of $\mathcal{H}$ if
		\begin{enumerate}\item $|\psi_i\rangle\perp |\psi_j\rangle$ for any $i,j\in I$ with $i\neq j$; and \item
			$|\psi\rangle=\sum_{i\in I}\langle\psi_i|\psi\rangle|\psi_i\rangle$
			for each $|\psi\rangle\in\mathcal{H}.$\end{enumerate} In this case,
		the cardinality of $I$ is called the dimension of $\mathcal{H}$. We use $\dim(\cH)$ to denote the dimension of $\cH.$
		
		A Hilbert space is defined to be a complete inner product space;
		that is, an inner product space in which each Cauchy sequence of
		vectors has a limit. According to a basic postulate of quantum
		mechanics, the state space of an isolated quantum system is
		represented by a Hilbert space, and a pure state of the system is
		described by a unit vector in its state space.
		
		\begin{example}\label{ex-qub}\begin{enumerate}\item The state space of qubits is the
				$2-$dimensional Hilbert space:
				$$\mathcal{H}_2=\{\alpha|0\rangle+\beta|1\rangle:\alpha,\beta\in\mathbb{C}\}.$$
				The inner product in $\mathcal{H}_2$ is defined by $$(\alpha
				|0\rangle+\beta |1\rangle,\alpha^{\prime} |0\rangle+\beta^{\prime}
				|1\rangle)=\alpha^{\ast}\alpha^{\prime}+\beta^{\ast}\beta^{\prime}
				$$ for all
				$\alpha,\alpha^{\prime},\beta,\beta^{\prime}\in\mathbb{C}$. Then
				$\{|0\rangle, |1\rangle\}$ is an orthonormal basis of
				$\mathcal{H}_2$, called the computational basis.
				
				
				\item The space $l_2$ of square summable sequences is $$\mathcal{H}_\infty=\Big\{\sum_{n=-\infty}^{\infty}\alpha_n|n\rangle:\alpha_n\in\mathbb{C}\
				{\rm for\ all}\ n\in\mathbb{Z}\ {\rm and}\
				\sum_{n=-\infty}^{\infty}|\alpha_n|^{2}<\infty\Big\},$$ where
				$\mathbb{Z}$ is the set of integers. The inner product in
				$\mathcal{H}_\infty$ is defined by
				$$\Big(\sum_{n=-\infty}^{\infty}\alpha_n|n\rangle,\sum_{n=-\infty}^{\infty}\alpha^{\prime}_n|n\rangle\Big)=\sum_{n=-\infty}^{\infty}
				\alpha_n^{\ast}\alpha_n^{\prime}$$ for all
				$\alpha_n,\alpha_n^{\prime}\in\mathbb{C}$, $-\infty <n<\infty$. Then
				$\{|n\rangle: n\in\mathbb{Z}\}$ is an orthonormal basis of
				$\mathcal{H}_\infty$, called the computational basis.
		\end{enumerate}\end{example}

		A (linear) operator on a Hilbert space $\mathcal{H}$ is a mapping
		$A:\mathcal{H}\rightarrow\mathcal{H}$ satisfying the following
		conditions:\begin{enumerate}\item
			$A(|\varphi\rangle+|\psi\rangle)=A|\varphi\rangle+A|\psi\rangle$;
			\item $A(\lambda |\psi\rangle)=\lambda A|\psi\rangle$
		\end{enumerate} for all $|\varphi\rangle,|\psi\in\mathcal{H}$ and
		$\lambda\in\mathbb{C}$. If $\{|\psi_i\rangle\}$ is an orthonormal
		basis of $\mathcal{H}$, then an operator $A$ is uniquely determined
		by the images $\{A|\psi_i\rangle\}$ of basis vectors
		$\{|\psi_i\rangle\}$ under $A$. In particular, $A$ can be
		represented by matrix
		$$A=\left(\langle\psi_i|A|\psi_j\rangle\right)_{ij}$$ when
		$\mathcal{H}$ is finite-dimensional. An operator $A$ on $\mathcal{H}$ is said to be bounded if there is a
		constant $c\geq 0$ such that $\|A|\psi\rangle\|\leq c\cdot\|\psi\|$
		for all $|\psi\rangle\in\mathcal{H}$. In the paper, only bounded operators are considered and for simplicity, we omit ``bounded''.
		The identity operator
		on $\mathcal{H}$ is denoted $I_{\mathcal{H}}$, and the zero operator
		on $\mathcal{H}$ that maps every vector in $\mathcal{H}$ to the zero
		vector is denoted $0_\mathcal{H}$.
		
		For any operator $A$ on $\mathcal{H}$, there exists a unique linear
		operator $A^{\dag}$ on $\mathcal{H}$ such that
		$$(|\varphi\rangle,
		A|\psi\rangle)=(A^{\dag}|\psi\rangle,|\varphi\rangle)$$ for all
		$|\varphi\rangle, |\psi\rangle\in\mathcal{H}$. The operator
		$A^{\dag}$ is called the adjoint of $A$. Given the matrix form of $A$, $A^\dag$ is the conjugate transpose of $A$.
		
		Following are frequently used sets of operators:
		\begin{enumerate}
			\item \emph{Hermitian operator}: An operator $M$ on
			$\mathcal{H}$ is said to be Hermitian if $M^{\dag}=M$.
			\item \emph{Positive semi-definite operator}: An Hermitian operator $A$ on $\mathcal{H}$ is said to be positive semi-definite if $\langle
			\psi|A|\psi\rangle\geq 0$ for all states $|\psi\rangle\in
			\mathcal{H}$. 
			\item \emph{Projection}: An Hermitian operator $P$ on
			$\mathcal{H}$ is a projection if $P^2 = P$. There is a one-to-one correspondence between the closed subspaces and projections: given projection $P$, its corresponding closed subspace is $\{|\phi\>\in\cH: P|\phi\> = |\phi\>\}$; and given closed subspace $V\subseteq\cH$, its corresponding projection is $\sum_i|\phi_i\>\<\phi_i|$ where $\{|\phi_i\>\}$ is an orthonormal basis of $V$.
			\item \emph{Unitary operator}: An operator $U$ on $\mathcal{H}$ is unitary if $UU^\dag = U^\dag U = I$ where $I$ is the identity operator on $\mathcal{H}$.
			\item \emph{Density operator}: An positive semi-definite operator with trace one.
		\end{enumerate}
		
		Following are frequently used concepts of operators:
		\begin{enumerate}
			\item \emph{Trace}: The trace of an operator $A$ on $\mathcal{H}$ is given by $$\tr(A)=\sum_{i}\langle \psi_i|A|\psi_i\rangle$$
			where $\{|\psi_i\>\}$ is an orthonormal basis of $\mathcal{H}$. In particular, trace is independent of the choice of the orthonormal basis. Given the matric form of $A$, $\tr(A)$ is exactly the summation of diagonal entries of $A$.
			\item \emph{Support}: The \emph{support} of a Hermitian operator $M$ on $\mathcal{H}$ is the (topological) closure of subspace spanned by its eigenvectors with nonzero eigenvalues. One can show that:
			$$\supp(M) = \{|\phi\>\in\cH:\ \<\phi|\rho|\phi\> = 0\}^\bot,$$
			where ${}^\bot$ stands for ortho-complement.
			\item \emph{L\"owner Order}: Given two Hermitian operator $A,B$ on $\mathcal{H}$, we use the \emph{L\"owner order} to compare them which is defined as follows: 
			$$A\sqsubseteq B \text{ if and only if } B-A \text{ is positive semi-definite};$$
			that is, for any $|\phi\>\in\cH$, $\<\phi|A|\phi\> \le \<\phi|B|\phi\>$. Whenever both $A,B$ are projections, their L\"owner order is consistent with the inclusion relation between the subspaces corresponding to $A,B$.
			\item \emph{eigenspaces of eigenvalue 1}: For a Hermitian operator $A$ on $\mathcal{H}$, we define $\proj (A)$ as the eigenspaces of eigenvalue 1:
			$$\proj(A) = \{|\phi\>\in\cH: A|\phi\> = |\phi\>\}.$$
		\end{enumerate}

		\subsection{Basics of Quantum Information, Extended Version}
		\label{sec basic Quantum extended version}
		
		The \emph{state space} of a quantum system is a Hilbert space $\cH$, which is
		essentially a vector space in the finite-dimensional case. A \emph{pure state}
		of the system is a unit column vector $|\psi\>\in\cH$. For example, the state
		space of a quantum bit (aka qubit) is a two-dimensional Hilbert space $\cH_2$ with
		basis states $$|0\> = \left[\begin{array}{c} 1 \\ 0\end{array} \right] \quad\text{and}\quad
		|1\> = \left[\begin{array}{c} 0 \\ 1 \end{array}\right],$$ and any pure state
		of a qubit can be described in the form $$\alpha|0\>+\beta|1\> =
		\left[\begin{array}{c} \alpha \\ \beta \end{array}\right]$$ satisfying
		normalization condition $|\alpha|^2+|\beta|^2 = 1$. The orthonormal basis is not unique, for example, the states $|+\> = \frac{1}{\sqrt{2}}(|0\>+|1\>)$ and $|-\> = \frac{1}{\sqrt{2}}(|0\>-|1\>)$ is another orthonormal basis of $\cH_2$.

		When the state is not
		completely known but could be in one of some pure states $|\psi_i\>$ with
		respective probabilities $p_i$, we call $\{(p_i,|\psi_i\>)\}$ an
		\emph{ensemble} of pure states or a \emph{mixed state}, and the system is
		fully described by the \emph{density operator} $\rho =
		\sum_ip_i|\psi_i\>\<\psi_i|$, which mathematically, the positive semi-definite operator with unit trace. For example, the completely mixed state of a
		qubit can be seen as ensemble $\{(0.5,|0\>), (0.5,|1\>)\}$ (i.e. the state is
		either $|0\>$ or $|1\>$ with the same probability 0.5) or density matrix
		$$\frac{1}{2}(|0\>\<0|+|1\>\<1|) = \left[\begin{array}{cc} 0.5 & 0 \\ 0 & 0.5
		\end{array}\right];$$ 
		if a state is in $|0\rangle$ with probability $\frac{2}{3}$ and in $|+\rangle$ with probability $\frac{1}{3}$, then it can be described by density operator
		\begin{equation}\label{ex-mix}\rho=\frac{2}{3}|0\rangle\langle 0|+\frac{1}{3}|+\rangle\langle +|=\frac{1}{6}\left (\begin{array}{cc}5 & 1\\ 1&1\end{array}\right ).\end{equation}
		
		The evolution of a quantum system is modelled by a \emph{unitary operator}
		$U$; i.e.  a complex matrix with $UU^\dag=U^\dag U$ being the identity
		operator, where $\dag$ is conjugate transpose. In quantum computing, operators
		are often called \emph{quantum gates}. For example, the Hadamard gate
		$H=\frac{1}{\sqrt{2}}\left[\begin{array}{cc} 1 & 1 \\ 1 & -1
		\end{array}\right]$ maps $|0\>, |1\>$ to their superpositions
		$|+\>$ and $|-\>$ respectively:
		\begin{align*}
		&H|0\> = \frac{1}{\sqrt{2}}\left[\begin{array}{cc} 1 & 1 \\ 1 & -1
		\end{array}\right]\cdot\left[\begin{array}{c} 1 \\ 0\end{array} \right] = \frac{1}{\sqrt{2}}\left[\begin{array}{c} 1 \\ 1\end{array} \right] = |+\>,\\
		&H|1\> = \frac{1}{\sqrt{2}}\left[\begin{array}{cc} 1 & 1 \\ 1 & -1
		\end{array}\right]\cdot\left[\begin{array}{c} 0 \\ 1\end{array} \right] = \frac{1}{\sqrt{2}}\left[\begin{array}{c} 1 \\ -1\end{array} \right] = |-\>.
		\end{align*}
		
		Unlike a classical system which can be observed directly without changing its
		state, we need to perform a quantum measurement to extract information from a
		quantum state which inevitably leads to state collapse. Formally, a
		\emph{projective quantum measurement} consists of a set of \emph{projections},
		$M_0,M_1,\dots, M_n$ that satisfies the completeness condition:
		$$\sum_{i = 0}^n M_i^\dag M_i = \sum_{i = 0}^n M_i = I$$
		where $I$ is the identity operator. When such a measurement is applied to a quantum state
		$\rho$, we obtain one of the classical outcome $i\in\{0,1,\dots,n\}$ with
		probability $p_i = \mathrm{tr}(M_i\rho)$, and the post-measurement state of
		the system is then $\frac{M_i\rho M_i}{p_i}$.
		For instance, consider the measurement defined by $\cM=\{M_0 = |+\>\<+|,M_1 = |-\>\<-|\}$, and if we perform $\cM$ on a qubit in (mixed) state $\rho$ given in equation~(\ref{ex-mix}), then the probability that we get outcome \textquotedblleft$1$\textquotedblright\ is $$p_1=\mathit{tr}(M_1\rho)=tr\left(\frac{1}{2}\left[\begin{array}{cc}1&-1\\-1&1\end{array}\right]
		\cdot\frac{1}{6}\left[\begin{array}{cc}5& 1\\ 1&1\end{array}\right]
		\right)=\frac{1}{12}\cdot \tr\left[\begin{array}{cc}4&0\\-4&0\end{array}\right]=\frac{1}{3}$$ and after that, the qubit's state will change to $|-\>\<-|$:
		$$M_1\rho M_1/p_1 = \frac{1}{2}\left[\begin{array}{cc}1&-1\\-1&1\end{array}\right]\cdot\frac{1}{6}\left[\begin{array}{cc}5& 1\\ 1&1\end{array}\right]\cdot\frac{1}{2}\left[\begin{array}{cc}1&-1\\-1&1\end{array}\right]\div\frac{1}{3} = \frac{1}{2}\left[\begin{array}{cc}1&-1\\-1&1\end{array}\right] = |-\>\<-|.$$
		Similarly, the probability of outcome \textquotedblleft$0$\textquotedblright\ is $p_0=\frac{2}{3}$, and then the state changes to $|+\>\<+|$.

		We use variables $p,q,r, ...$ to denote quantum systems. Operations in quantum
		computing are often performed on a composite system consisting of multiple
		qubits. To indicate which system a state describes or an operation acts on, we
		use subscripts; for example, $\cH_p$ is the state space of system $p$,
		$|0\>_{p}$ is the pure state $|0\>$ of the system $p$ and $|1\>_{q}\<1|$ is
		the density matrix of the system $q$. The composite system is described by the
		tensor product of its subsystems; for example, a composite system $pq$ with $p,q$ being single qubit systems has the
		state space $\cH_p\otimes\cH_q$, and 
		$$|0\>_p\otimes|1\>_q = \left[\begin{array}{c} 1 \\ 0\end{array} \right]_p\otimes\left[\begin{array}{c} 0 \\ 1\end{array} \right]_q = \left[\begin{array}{c} 0 \\ 1 \\ 0 \\ 0\end{array} \right]_{pq}$$
		(or, $|0\>_p|1\>_q$
		for short) is a pure state in which subsystem $p$ is in state $|0\>$ and
		subsystem $q$ is in state $|1\>$. Due to the superposition principle, there
		exist states like 
		\begin{equation}
		\label{ex-ent}
		|\Phi\>_{pq} =
		\frac{1}{\sqrt{2}}(|0\>_p|0\>_q+|1\>_p|1\>_q) = \frac{1}{\sqrt{2}}\left[\begin{array}{c} 1 \\ 0 \\ 0 \\ 1\end{array} \right]_{pq}
		\end{equation} that cannot be written in the
		simple tensor form $|\phi\>_p|\psi\>_q$, which are called \emph{entangled
			states}. These states play a crucial role in applications of quantum
		computation and quantum communication. 
		
		The state of a composite system fully determines the state of each subsystem.
		Formally, given composite system $pq$ in state $\rho$, subsystem $q$ is then
		in state $\tr_p(\rho)$, where the partial trace $\tr_p(\cdot)$ over $p$ is a
		mapping from operators on $\cH_p\otimes\cH_q$ to operators on $\cH_q$ defined
		by: $$\tr_p(|\phi_p\>_p\<\psi_p|\otimes|\phi_q\>_q\<\psi_q|) =
		\<\psi_p|\phi_p\>\cdots|\phi_q\>_q\<\psi_q|$$  for all
		$|\phi_p\>,|\psi_p\>\in\cH_p$ and $|\phi_q\>,|\psi_q\>\in\cH_q$ together with
		linearity. The state $\tr_q(\rho)$ of subsystem $q$ can be defined
		symmetrically.  We often use the notations $\rt{\rho}{p} \triangleq
		\tr_p(\rho)$ and $\rt{\rho}{q} \triangleq \tr_q(\rho)$ in order to explicitly
		indicate that $\rt{\rho}{p}$ and  $\rt{\rho}{q}$ are states of $p,q$,
		respectively. For example, if the composite system $pq$ is in state $
		|\Phi\rangle_{pq}$ defined in Eqn. \ref{ex-ent} or equivalently represented by density operator $\Phi_{pq}$ \begin{equation}\label{max-ent}\Phi_{pq} = |\Phi\>_{pq}\<\Phi| = \frac{1}{2}(|0\>_p\<0|\otimes|0\>_q\<0|+ |0\>_p\<1|\otimes|0\>_q\<1|+|1\>_p\<0|\otimes|1\>_q\<0|+|1\>_p\<1|\otimes|1\>_q\<1|)
		= \frac{1}{2}\left[\begin{array}{cccc} 1 & & & 1\\ & 0 & & \\ & & 0 & \\ 1 & & & 1\end{array} \right]_{pq}
		\end{equation} then the partial traces $\rt{\Phi_{pq}}{q} = \tr_p(|\Phi\>_{pq}\<\Phi|) = \frac{1}{2}(|0\>_q\<0|+|1\>_q\<1|) \text{\ and\ } \rt{\Phi_{pq}}{p} = \tr_q(|\Phi\>_{pq}\<\Phi|) = \frac{1}{2}(|0\>_p\<0|+|1\>_p\<1|)$
		describe states of $q$ and $p$, respectively.
		
		\vspace{0.2cm}
		\noindent\textbf{Summary of Notations.}
		Let $\vars$ be  the set of all quantum variables. A \emph{quantum register} is a list of \emph{distinct} variables $\qbar = q_1,\dots,q_n$. Each quantum variable $q$ has a type $\cH_q$, which is the state Hilbert space of quantum system denoted by $q$.
		For a set of quantum variables $S = \{q_1,\dots,q_n\}\subseteq \vars$ (or a quantum register $\qbar = q_1,\dots,q_n$), we fix following notations:
		\begin{itemize}
			\item $\cH_S = \bigotimes_{i=1}^n\cH_{q_i}$: the Hilbert space of $S$.
			\item $\dim(S)$: the dimension of $\cH_S$.
			\item $\cD(S)$: the set of all mixed quantum states (i.e. density matrices)
			of $S$. In particular, for any $\rho\in\cD(S)$, its \emph{domain} is
			defined as $\dom{\rho}\triangleq S$; we write $\cD \triangleq
			\bigcup_{S\subseteq \vars}\cD(S)$ for the set of all states.
			\item $\cP(S)$: the set of projections on $\cH_S$. In particular, for any
			$P\in\cP(S)$, its domain is defined as $\free{P} \triangleq S$.  Since
			there is a one-to-one correspondence between projections and closed
			subspaces, we sometimes called closed subspaces of $\cH_S$ projections.
			We write $\cP\triangleq \bigcup_{S\subseteq \vars}\cP(S)$ for the set of
			all projections.
			\item $
			\rt{\rho}{S} \triangleq 
			\tr_{\dom{\rho}\backslash S}(\rho)
			$: the \emph{restriction} of state $\rho$ on $S$,  defined as a reduced density operator over $S\cap\dom{\rho}$. 
		\end{itemize}

		\vspace{0.3cm}
		
		\noindent\textbf{Permutations} of variables are frequently used in quantum computing, e.g., in qubit allocation~\cite{SSV18}. We use 
		$\perm(\qbar\mapsto\qbar^\prime)$ to denote the operator  
		that permutes a list $\qbar := q_1,q_2,\cdots,q_n$ of quantum variables to $\qbar^\prime = q_{i_1},q_{i_2},\cdots,q_{i_n}$. 	For example, if $q_1$ and $q_2$ are two different variables with same type and $\{|i\>\}$ is an arbitrary orthonormal basis of $\cH_{q_1}$ (and $\cH_{q_2}$), then the swap gate 
		$\swap[q_1,q_2]\triangleq\sum_{i,j}|i\>_{q_1}\<j|\otimes|j\>_{q_2}\<i|$ is the simplest permutation from $\qbar=q_1, q_2$ to $\qbar^\prime=q_2,q_1$, that is, for any $m,n$:
		$$\swap[q_1,q_2](|m\>_{q_1}|n\>_{q_2}) = \sum_{i,j}|i\>_{q_1}\<j|\otimes|j\>_{q_2}\<i|(|m\>_{q_1}|n\>_{q_2}) =   |n\>_{q_1}|m\>_{q_2}.$$
		Indeed, any permutation can be decomposed into a sequence of swap gates.

		\vspace{0.3cm}
		
		\noindent{\bf Meet and Join of Projections}: There is a one-to-one correspondence between the closed subspaces of a Hilbert space and projections in it, and moreover, the inclusion between closed subspaces is coincident with the L\"owner order between their projections.
		So, we do not distinguish a closed subspace from the projection onto it. Furthermore, let $^\perp$ stands for the orthocomplement, and for any $P, Q\in\cP(S)$, we define the \emph{meet} $\qwedge$ and \emph{join} $\qvee$: $$P\qwedge Q=P\cap Q,\quad\quad P\qvee Q=\overline{\operatorname{\spa}(P\cup Q)}$$ where $\overline{T}$ stands for the closure of $T$ and $\operatorname{\spa}(T)$ for the subspace spanned by $T$.
		It is well-known that $(\cP(S),\qwedge,\qvee,^\perp)$ is an orthomodular lattice (or quantum logic) \cite{BvN36, KA83}, with inclusion $\subseteq$ as its order.
		
		\section{BI and its quantum interpretation, Deferred Proofs for Section \ref{sec-q-interpret}}
		\label{sec app BI}
		
		\subsection{Hilbert-style rules for BI.}
		\label{sec app sub Hilbert-style rules for BI}
		Hilbert-style rules for BI is shown in Fig. \ref{fig HR for BI}.
		\begin{figure}[h]
			\small
			\begin{align*}
			&1. \quad \frac{}{\phi\vdash\phi}  \qquad
			&&2. \quad \frac{}{\phi\vdash\top} \qquad 
			&&3. \quad \frac{}{\bot\vdash\phi} \qquad
			&&4. \quad \frac{\mu\vdash\phi\quad\mu\vdash\psi}{\mu\vdash\phi\wedge\psi} \\
			&5. \quad \frac{\phi\vdash\psi_1\wedge\psi_2}{\phi\vdash\psi_i} \quad
			&&6. \quad \frac{\phi\vdash\psi}{\mu\wedge\phi\vdash\psi} \qquad
			&&7. \quad \frac{\mu\vdash\psi\quad\phi\vdash\psi}{\mu\vee\phi\vdash\psi} \qquad
			&&8. \quad \frac{\phi\vdash\psi_i}{\phi\vdash\psi_1\vee\psi_2} \\
			&9. \quad \frac{\mu\vdash\phi\rightarrow\psi\quad\mu\vdash\phi}{\mu\vdash\psi} \quad
			&&10. \quad \frac{\mu\wedge\phi\vdash\psi}{\mu\vdash\phi\rightarrow\psi} \quad 
			&&11. \quad \frac{\xi\vdash\phi\quad\mu\vdash\psi}{\xi\ast\mu\vdash\phi\ast\psi} \qquad
			&&12. \quad \frac{\mu\ast\phi\vdash\psi}{\mu\vdash\phi\sepimp\psi} \\
			&13. \quad \frac{\xi\vdash\phi\sepimp\psi\quad\mu\vdash\phi }{\xi\ast\mu\vdash\psi} \qquad
			&&14. \quad \frac{}{\phi\ast\psi\vdash\psi\ast\phi} \quad
			&&15. \quad \frac{}{(\phi\ast\psi)\ast\xi\vdash\phi\ast(\psi\ast\xi)} \qquad
			&&16. \quad \frac{}{\phi\ast\top\dashv\vdash\phi} 
			\end{align*}
			\caption{Hilbert-style rules for BI\cite{Pym02,Doc19}. $i=1$ or $2$ for rules 5 and 8. 
			}
			\label{fig HR for BI}
		\end{figure}
		
		\subsection{Proposition \ref{pro-rt}}
		
		\begin{proposition}[Properties of Partial Trace]
			\label{pro-rt}
			\begin{enumerate}
				\item $\forall\ S_1,S_2\subseteq \vars \text{\ and\ }\rho\in\cD:\ \rt{(\rt{\rho}{S_2})}{S_1} = \rt{(\rt{\rho}{S_1})}{S_2} = \rt{\rho}{S_1\cap S_2}$;
				\item $\forall\,S\subseteq\vars,\rho_1,\rho_2\in\cD$ with $\type{\rho_1}\cap\type{\rho_2}=\emptyset$:  $\rt{(\rho_1\otimes\rho_2)}{S} = \rt{\rho_1}{S}\otimes\rt{\rho_2}{S}.$
			\end{enumerate}
		\end{proposition}
		\begin{proof}
			Trivial.
		\end{proof}
		
		\subsection{Proof of Proposition \ref{prop quantum BI frame}}
		
		\begin{proposition}
			$(\cD,\circ, \preceq,1)$ forms a BI frame, where scalar number $1$ is understood as a state over the empty register. 
		\end{proposition}
		\begin{proof}
			It is straightforward to check all the properties defined in Definition \ref{def BI frame}.
			\begin{itemize}
				\item (Unit Existence): for all $\rho\in\cD$, note that scalar $1$ is of domain $\emptyset$, so $\rho\circ 1 = \rho\otimes 1 = \rho$, and $1\circ\rho = 1\otimes \rho = \rho$.
				\item (Commutativity): for all $\rho,\sigma\in\cD$, if their domains are overlap, i.e., $\dom {\rho}\cap\dom {\sigma}\neq\emptyset$, then neither $\rho\circ\sigma$ nor $\sigma\circ\rho$ is defined; if their domains are disjoint, i.e., $\dom {\rho}\cap\dom {\sigma}=\emptyset$, then $\rho\circ\sigma = \sigma\circ\rho$ since both of them denote the tensor product state over system $\dom {\rho}\cup\dom {\sigma}$ with reduced state $\rho$ over subsystem $\dom {\rho}$ and $\sigma$ over subsystem $\dom {\sigma}$.
				\item (Associativity): for all $\rho,\sigma, \delta\in\cD$, if their domains are pairwise disjoint, then $\rho\circ(\sigma\circ\delta) = (\sigma\circ\rho)\circ\delta$ since standard tensor product are associative; otherwise, neither $\rho\circ(\sigma\circ\delta)$ nor $(\sigma\circ\rho)\circ\delta$ is defined.
				\item (Compatible with $\preceq$): it follows from the property of partial trace. Formally, for any $\rho_1\preceq\rho_1^\prime$ and $\rho_2\preceq\rho_2^\prime$ and both $\rho\circ\rho_2$ and $\rho_1^\prime\circ\rho_2^\prime$ are defined, then we know:
				\begin{itemize}
					\item $\dom{\rho_1}\subseteq\dom{\rho_1^\prime}$, $\dom{\rho_2}\subseteq\dom{\rho_2^\prime}$, $\dom{\rho_1^\prime}\cap\dom{\rho_2^\prime} = \emptyset$;
					
					let us use notations: $S_1 \triangleq \dom{\rho_1}$, $S_1^\prime \triangleq \dom{\rho_1^\prime}\backslash S_1$, $S_2 \triangleq \dom{\rho_2}$, $S_2^\prime \triangleq \dom{\rho_2^\prime}\backslash S_2$;
					\item $\rho_1 = \tr_{\dom{\rho_1^\prime}\backslash\dom{\rho_1}}(\rho_1^\prime) = \tr_{S_1^\prime}(\rho_1^\prime)$ and $\rho_2 = \tr_{\dom{\rho_2^\prime}\backslash\dom{\rho_2}}(\rho_2^\prime) = \tr_{S_2^\prime}(\rho_2^\prime)$;
				\end{itemize}
				Rewrite $\rho_1^\prime$ and  $\rho_2^\prime$ in the explicit forms:
				$$\rho_1^\prime = \sum_{ii^\prime jj^\prime}\lambda_{ii^\prime jj^\prime}|i\>_{S_1}\<i^\prime|\otimes|j\>_{S_1^\prime}\<j^\prime|, \quad \rho_2^\prime = \sum_{mm^\prime nn^\prime}\gamma_{mm^\prime nn^\prime}|m\>_{S_2}\<m^\prime|\otimes|n\>_{S_2^\prime}\<n^\prime|$$
				where $\{|i\>\}, \{|j\>\},\{|m\>\},\{|n\>\}$ are orthonormal basis of system $S_1$, $S_1^\prime$, $S_2$ and $S_2^\prime$ respectively, and $\lambda_{ii^\prime jj^\prime}, \gamma_{mm^\prime nn^\prime}$ are complex numbers. By the definition of partial trace, we have:
				\begin{align*}
				&\rho_1 = \tr_{S_1^\prime}(\rho_1^\prime) = \tr_{S_1^\prime}\left(\sum_{ii^\prime jj^\prime}\lambda_{ii^\prime jj^\prime}|i\>_{S_1}\<i^\prime|\otimes|j\>_{S_1^\prime}\<j^\prime|\right) = \sum_{ii^\prime jj^\prime}\lambda_{ii^\prime jj^\prime}|i\>_{S_1}\<i^\prime|\cdot\<j^\prime|j\> = \sum_{ii^\prime j}\lambda_{ii^\prime jj}|i\>_{S_1}\<i^\prime| \\
				&\rho_2 = \tr_{S_2^\prime}(\rho_2^\prime) = \tr_{S_2^\prime}\left(\sum_{mm^\prime nn^\prime}\gamma_{mm^\prime nn^\prime}|m\>_{S_2}\<m^\prime|\otimes|n\>_{S_2^\prime}\<n^\prime|\right) = \sum_{mm^\prime n}\gamma_{mm^\prime nn}|m\>_{S_2}\<m^\prime|
				\end{align*}
				and we can calculate $\tr_{S_1^\prime\cup S_2^\prime}(\rho_1^\prime\circ\rho_2^\prime)$ directly:
				\begin{align*}
				\tr_{S_1^\prime\cup S_2^\prime}(\rho_1^\prime\circ\rho_2^\prime) &= \tr_{S_1^\prime\cup S_2^\prime}\left(\sum_{ii^\prime jj^\prime}\lambda_{ii^\prime jj^\prime}|i\>_{S_1}\<i^\prime|\otimes|j\>_{S_1^\prime}\<j^\prime| \otimes  \sum_{mm^\prime nn^\prime}\gamma_{mm^\prime nn^\prime}|m\>_{S_2}\<m^\prime|\otimes|n\>_{S_2^\prime}\<n^\prime|\right) \\
				&= \sum_{ii^\prime jj^\prime}\sum_{mm^\prime nn^\prime}\lambda_{ii^\prime jj^\prime}\gamma_{mm^\prime nn^\prime}|i\>_{S_1}\<i^\prime|\otimes|m\>_{S_2}\<m^\prime|\cdot \left(\<j^\prime|j\>\<n^\prime|n\>\right) \\
				&= \sum_{ii^\prime j}\sum_{mm^\prime n}\lambda_{ii^\prime jj}\gamma_{mm^\prime nn}|i\>_{S_1}\<i^\prime|\otimes|m\>_{S_2}\<m^\prime| \\
				&= \rho_1\otimes\rho_2 = \rho_1\circ\rho_2
				\end{align*}
				which leads to $\rho_1\circ\rho_2\preceq\rho_1^\prime\circ\rho_2^\prime$.
			\end{itemize}
		\end{proof}
		
		\subsection{Proposition \ref{pro mon res ato prop}}
		
		\begin{proposition}[Monotonicity and restriction of atomic proposition]
			\label{pro mon res ato prop}
			For any $p\in\AP$ (atomic propositions defined in Sec. \ref{sec free choice of AP}) and $\rho,\sigma\in \cD$ such that $\rho\preceq \sigma$ and $\LTypeF{p}\subseteq\LTypeE{\rho}$, 
			$\rho\models p$ if and only if $\sigma\models p$.
		\end{proposition}
		
		\begin{proof}
			Trivial by the definition and interpretation of atomic propositions defined in Sec. \ref{sec free choice of AP}.
		\end{proof}

		\subsection{Proof of Proposition \ref{prop axiom projection}}
		
		\begin{proposition}[Proposition \ref{prop axiom projection}, Extended Version]
			\label{prop axiom projection extend}
			\begin{enumerate}
				\item For all $S\subseteq \vars$ and identity operator $I_S$ over $\cH_S$, $$\models \bD[S]\leftrightarrow \id_{S}.$$ 
				\item For all $P,Q\in\cP$ with same domain, $\models P\rightarrow Q$ if and only if $P \sqsubseteq Q$;
				\item For all $P,Q\in\cP$ with disjoint domains, then
				$\models  P\wedge Q \leftrightarrow (P\otimes Q)$ ;   	
				\item If $S_1\subseteq S_2$, then $\models\unia[S_2]\rightarrow\unia[S_1]$.
				\item If $S_1, S_2$ are disjoint, then:
				$ \models(\unia[S_1]\ast\unia[S_2])\leftrightarrow\unia[S_1\cup S_2] $. 
				\item For all $P\in\cP$, $S_1\subseteq\free{P}$ and $S_2\cap\free{P} = \emptyset$, if $\models P\rightarrow\unia[S_1]$, then $\models (P\wedge\unia[S_2])\rightarrow \unia[S_1\cup S_2]$. 
			\end{enumerate}
		\end{proposition}
		
		\begin{proof}
			\begin{itemize}
				\item By definition, $\sem{\bD[S]} \triangleq \left\{\rho\in\cD: S\subseteq\dom{\rho}\right\}$. On the other hand, $\free {I_S} = S$ and for any $\rho$ with $\dom {\rho}\supseteq S$, $\supp(\rt {\rho}{S})\subseteq I_S$, so 
				$\sem{I_S} = \left\{\rho\in\cD: S\subseteq\dom{\rho}\right\}$. Therefore, 
				$\models \bD[S]\leftrightarrow \id_{S}.$
				
				\item Suppose $\free{P} = \free{Q} = S$. Then we have: $\models P\rightarrow Q$ iff $\forall\,\rho\in\cD(S),$ $\rho\models P$ implies $\rho\models Q$ iff $\forall\,\rho\in\cD(S),$ $\supp(\rho)\subseteq P$ implies $\supp(\rho)\subseteq Q$ (regarded as subspaces) iff subspaces $P$ and $Q$ have inclusion relation $P\subseteq Q$ iff $P\sqsubseteq Q$ (regarded as projections).
				
				\item Suppose $\rho\models P\otimes Q$, then $\type{\rho}\supseteq\typef{P}\cup\typef{Q}$ and $\supp(\rt{\rho}{\typef{P}\cup\typef{Q}})\subseteq P\otimes Q$. Note that $\typef{P}\cap\typef{Q}=\emptyset$, so $\supp(\rt{\rho}{\typef{P}})\subseteq P$ and $\supp(\rt{\rho}{\typef{Q}})\subseteq Q$; that is,  $\rt{\rho}{\typef{P}}\models P$ and $\rt{\rho}{\typef{Q}}\models Q$, and thus $\rho\models P\wedge Q$.
				
				Suppose $\rho\models P\wedge Q$, then $\rt{\rho}{\typef{P}}\models P$ and $\rt{\rho}{\typef{Q}}\models Q$. As $\typef{P}\cap\typef{Q}=\emptyset$, $\supp(\rt{\rho}{\typef{P}\cup\typef{Q}})\subseteq P\otimes Q$, or equivalently, $\rho\models P\otimes Q$.
				
				\item For any $\rho\models\unia[S_2]$, we must have: $\dom{\rho}\supseteq S_2\supseteq S_1$, and $\rt {\rho}{S_2} = \frac{I_{S_2}}{\dim (S_2)}$. Take the partial trace over $S_1$, we obtain $\rt {\rho}{S_1} = \frac{I_{S_1}}{\dim (S_1)}$ and thus $\rho\models\unia[S_1]$.
				
				\item For any $\rho\models(\unia[S_1]\ast\unia[S_2])$, $\dom{\rho}\supseteq S_1\cup S_2$ and by Proposition \ref{pro mon res ato prop}, $\rt{\rho}{S_1\cup S_2}\models(\unia[S_1]\ast\unia[S_2])$. Therefore, $\rt{\rho}{S_1}\models\unia[S_1]$ and $\rt{\rho}{S_2}\models\unia[S_2]$ and $\rt{\rho}{S_1\cup S_2} = \rt{\rho}{S_1}\otimes\rt{\rho}{S_2}$, and thus
				$$\rt{\rho}{S_1\cup S_2} = \frac{\id_{S_1}}{\dim(S_1)}\otimes\frac{\id_{S_2}}{\dim(S_2)} = \frac{\id_{S_1\cup S_2}}{\dim(S_1\cup S_2)}$$
				which leads to $\rho\models\unia[S_1\cup S_2]$.
				
				If $\rho\models\unia[S_1,S_2]$, then 
				$$\rt{\rho}{S_1\cup S_2} = \frac{\id_{S_1\cup S_2}}{\dim(S_1\cup S_2)} = \frac{\id_{S_1}}{\dim(S_1)}\otimes\frac{\id_{S_2}}{\dim(S_2)}\models \unia[S_1]\ast\unia[S_2]$$
				and so $\rho\models\unia[S_1]\ast\unia[S_2]$.
				
				\item Suppose $\models P\rightarrow\unia[S_1]$. Assume $\{|e_i\>\}$ is an orthonormal basis of $P$, $\{|k\>\}$ an orthonormal basis of $\cH_{\free{P}\backslash S_1}$. First, it is trivial to realize for any $i$, $|e_i\>\<e_i|\models P$, so it must satisfy $\unia[S_1]$, that is,
				\begin{align*}
				\rt{|e_i\>\<e_i|}{S_1} &= \sum_k\<k|e_i\>\<e_i|k\> = \frac{I_{S_1}}{\dim(S_1)}.
				\end{align*}
				Next, for any $i\neq i^\prime$, choose two states $\frac{1}{\sqrt{2}}(|e_i\>+|e_{i^\prime}\>)$ and $\frac{1}{\sqrt{2}}(|e_i\>+\bi|e_{i^\prime}\>)$ which also satisfy $P$ and so $\unia[S_1]$, then 
				\begin{align*}
				&\frac{1}{2}\sum_k\<k|e_i\>\<e_{i^\prime}|k\> + \frac{1}{2}\sum_k\<k|e_{i^\prime}\>\<e_i|k\> = 0 \\
				&\frac{-\bi}{2}\sum_k\<k|e_i\>\<e_{i^\prime}|k\> + \frac{\bi}{2}\sum_k\<k|e_{i^\prime}\>\<e_i|k\> = 0
				\end{align*}
				which lead to $\sum_k\<k|e_i\>\<e_{i^\prime}|k\> = 0$. Now, for any $\rho\in\cD(\free{P}\cup S_2)$ that satisfy $P\wedge\unia[S_2]$, it can be written in the form
				$$\rho = \sum_m\left(\sum_i|e_i\>|h_{im}\>\right)\left(\sum_{i^\prime}\<e_{i^\prime}|\<h_{i^\prime m}|\right) = \sum_{ii^\prime m}|e_i\>\<e_{i^\prime}|\otimes|h_{im}\>\<h_{i^\prime m}|$$
				where the states $|h_{im}\>$ may not be unit vectors. By restriction, its reduced state $\rt{\rho}{S_2} = \sum_{im}|h_{im}\>\<h_{im}|\models \unia[S_2]$.
				We observe that:
				\begin{align*}
				\rt{\rho}{S_1\cup S_2} &= \sum_k\<k|\rho|k\> = \sum_{kii^\prime m}\<k|e_i\>\<e_{i^\prime}|k\>\otimes|h_{im}\>\<h_{i^\prime m}| \\
				&= \sum_{kim}\<k|e_i\>\<e_{i}|k\>\otimes|h_{im}\>\<h_{im}|  \\
				&= \frac{I_{S_1}}{\dim(S_1)}\otimes\sum_{im}|h_{im}\>\<h_{im}|, 
				\end{align*}
				and thus, $\rt{\rho}{S_1\cup S_2}\models \unia[S_1]\ast\unia[S_2]$. By $\mathit{4}$, we know that 
				$\rt{\rho}{S_1\cup S_2}\models \unia[S_1\cup S_2]$. Finally, by monotonicity and restriction, this conclusion holds for all $\rho\in\cD$ and thus finishes the proof.
				
			\end{itemize}
			
		\end{proof}
		
		\subsection{Nonexistence of Extension}
		
		As is well-known, the frame rule plays an essential role in separation logic, and in turn it heavily relies on the restriction property that satisfaction only depends on the free variables appearing in a BI formula $\phi$. The restriction property and frame rule were successfully generalised into probabilistic separation logic in \cite{BHL19}. Essentially, the validity of the restriction property in the probabilistic setting can be attributed to a fundamental fact in probabilistic theory -- existence of extension: for any three random variables $x,y,z$, if joint distributions $\mu_{xy}$ and $\mu_{yz}$ coincide on $y$, then there exists a joint distribution $\mu_{xyz}$ with $\mu_{xy}$ and $\mu_{yz}$ as its marginals. Unfortunately, existence of extension is not true for quantum systems as shown in the following:
		
		\begin{example}[Non-existence of Extension]
			\label{prop noexist extension}
			Consider three qubits $q_1,q_2,q_3$ and states $\rho_{12}\in\cD(q_1q_2)$, $\rho_{23}\in\cD(q_2q_3)$:
			$$\rho_{12} = {\small\left[\begin{array}{cccc}0.5 & & & \\ & 0 & & \\ & & 0 & \\ & & & 0.5\end{array}\right]},\quad
			\rho_{23} = {\small\left[\begin{array}{cccc}0.25 & & 0.25 & \\ & 0.25 & & -0.25\\ 0.25 & & 0.25 & \\ & -0.25 & & 0.25\end{array}\right]}.
			$$
			It is easy to see that $\tr_{q_1}(\rho_{12}) = \tr_{q_3}(\rho_{23})$. But by SDP (Semi-definite Programming), we can prove that there is no  $\rho_{123}\in\cD(q_1q_2q_3)$ such that $\tr_{q_3}(\rho_{123}) = \rho_{12}$ and $\tr_{q_1}(\rho_{123}) = \rho_{23}$. This shows that existence of extension does not hold even for separable states $\rho_{12}$ and/or $\rho_{23}$.
		\end{example}
		
		
		\subsection{Failure of the Restriction Property}
		As a consequence, the restriction property: $\rho\models\phi\Rightarrow \rt{\rho}{\V{\phi}}\models\phi$ where $\V{\phi}$ stands for the free variables occurring in $\phi$, does not hold, even for the ordinary implication $\phi=\phi_1\rightarrow\phi_2$ (see Definition \ref{def satisfaction BI} for its semantics). 
		
		\begin{example}[Failure of the Restriction Property] 
			\label{exam failure restriction implication}
			Consider three qubits $q_1,q_2,q_3$ and maximally entanglement (Bell states) $|\Psi^\pm_{ij}\> =  \frac{1}{\sqrt{2}}(|0\>_{q_i}|0\>_{q_j}\pm|1\>_{q_i}|1\>_{q_j})$ between $q_i$ and $q_j$ for $1\le i\neq j\le 3$. Their density matrices are $\Psi^\pm_{ij} = |\Psi^\pm_{ij}\>\<\Psi^\pm_{ij}|.$
			Set $\rho = \Psi^+_{12}\in\cD(q_1,q_2)$. Let $\phi_1 = \Psi^+_{23}\in\cP(q_2,q_3),\ \phi_2 = \Psi^-_{23}\in\cP(q_2,q_3)$. Then: 
			\begin{itemize}
				\item $\rho\models\phi_1\rightarrow\phi_2$ is valid because there does not exist $\rho^\prime\succeq\rho$ such that $\rho^\prime\models\phi_1$; that is, no extension of $\Psi^+_{12}$ and $\Psi^+_{23}$ exists. 
				\item It is easy to see that   $\rt{\rho}{\LTypeF{\phi_1\rightarrow\phi_2}} = \frac{\id_{q_2}}{2}\in\cD(q_2)$. Choose $\rho^{\prime\prime} = \Psi^+_{23}\in\cD(q_2,q_3)$. It holds that $\rho^{\prime\prime}\succeq\rt{\rho}{\LTypeF{\phi_1\rightarrow\phi_2}}$. Note that $\rho^{\prime\prime}\models\phi_1$, but $\rho^{\prime\prime}\models\phi_2$ is not true. Therefore, $\rt{\rho}{\LTypeF{\phi_1\rightarrow\phi_2}}\not\models\phi_1\rightarrow\phi_2$.\end{itemize}
		\end{example}
		
		\vspace{0.5cm}
		
		\noindent\textbf{Problem in Program logic without Restriction} The following example shows that, without the domain assumption and restriction property, local reasoning is not sound in program logic.
		\begin{problem}
			For classical assignment rule $\{\phi[e/x]\}x:=e\{\phi\}$, if $\phi$ do not contain free variable $x$, then $\{\phi\}x:=e\{\phi\}$. However, such simple rule doesn't hold for quantum case.
			
			Here is a simple example. Consider a three qubits system $q_1,q_2,q_3$, and let $|\Phi^\pm\> = |00\>\pm|11\>$, $\Phi^\pm = |\Phi^\pm\>\<\Phi^\pm|$. Now, the state $\Phi^+_{12}$ indeed satisfies the formula $\Phi^+_{23}\rightarrow\Phi^-_{23}$ because for any $\rho\succeq\Phi^+_{12}$, $\rho\not\models\Phi^+_{23}$. However, if we do an initialization on $q_1$, which is disjoint of the domain of $\Phi^+_{23}\rightarrow\Phi^-_{23}$ (its domain is $\{q_2,q_3\}$). Now the state $\Phi^+_{12}$ is changed to $|0\>_1\<0|\otimes I_2$, which violate $\Phi^+_{23}\rightarrow\Phi^-_{23}$ ($|0\>_1\<0|\otimes\Phi^+_{23} \succeq|0\>_1\<0|\otimes I_2$, $|0\>_1\<0|\otimes\Phi^+_{23}\models\Phi^+_{23}$ but $|0\>_1\<0|\otimes\Phi^+_{23}\not\models\Phi^-_{23}$)!
		\end{problem}

		\subsection{Proof of Proposition \ref{pro res set}}
		
		\begin{proposition}
			Any formula $\phi\in\res$ is restrictive, i.e., for any $\rho\models\phi$, $\rt{\rho}{\free{\phi}}\models\phi$.
		\end{proposition}
		\begin{proof}
			It is straightforward to prove it by induction on the structure of $\phi$.
			\begin{itemize}
				\item $\phi\equiv p\in\AP$. By Proposition \ref{pro mon res ato prop}.
				\item $\phi \equiv \top$ or $\bot$. Trivial.
				\item $\phi\equiv \phi_1\wedge (\vee)\ \phi_2$. If $\rho\models\phi_1\wedge (\vee)\ \phi_2$, then $\rho\models\phi_1$ and(or) $\rho\models\phi_2$, by induction hypotheses, we know $\rt{\rho}{\free{\phi_1}}\models\phi_1$ and(or) $\rt{\rho}{\free{\phi_2}}\models\phi_2$ and by monotonicity, $\rt{\rho}{\free{\phi_1\wedge (\vee)\, \phi_2}}\models\phi_1$ and(or) $\rt{\rho}{\free{\phi_1\wedge (\vee)\, \phi_2}}\models\phi_2$, and thus, $\rt{\rho}{\free{\phi_1\wedge (\vee)\, \phi_2}}\models\phi_1\wedge (\vee)\ \phi_2$.
				\item $\phi\equiv \phi_1\ast \phi_2$. If $\rho\models\phi_1\ast \phi_2$, then there exist $\rho_1$ and $\rho_2$ such that $\rho\succeq\rho_1\otimes\rho_2$ and $\rho_1\models\phi_1$, $\rho_2\models\phi_2$. By Proposition \ref{pro res set} and a careful treatment of variable sets, we know that $\rt{\rho}{\free{\phi_1\ast \phi_2}} = \rt{\rho_1}{\free{\phi_1}}\otimes \rt{\rho_2}{\free{\phi_2}}$, and by induction hypotheses, $\rt{\rho}{\free{\phi_1}}\models\phi_1$ and $\rt{\rho}{\free{\phi_2}}\models\phi_2$, thus $\rt{\rho}{\free{\phi_1\ast \phi_2}}\models\phi_1\ast \phi_2$.
			\end{itemize}
		\end{proof}
		
		\subsection{Proposition \ref{pro BI res mon}}
		
		\begin{proposition}
			\label{pro BI res mon}
			For any $\phi\in\res$ and $\rho,\sigma\in \cD$ such that $\rho\preceq \sigma$ and $\LTypeF{\phi}\subseteq\LTypeE{\rho}$, 
			$\rho\models \phi$ if and only if $\sigma\models \phi$.
		\end{proposition}
		\begin{proof}
			By monotonicity and Proposition \ref{pro res set}.
		\end{proof}
		
		\subsection{Proof of Proposition \ref{pro modification}}
		
		\begin{proposition}
			Let $\prog$ be unitary transformation $\qU$ or initialisation $\qI$, and $\phi$ be any BI formula. If its modification  $\phi[\prog]$ is defined, then:
			\begin{enumerate}
				\item $\phi$ and $\phi[\prog]$ have the same domain: $\free{\phi} = \free{\phi[\prog]}$;
				\item for all $\rho\in\cD(\free{\phi}\cup\var(\prog))$, if $\rho\models \phi[\prog]$, then $\sem{\prog}(\rho)\models \phi$.
			\end{enumerate}
		\end{proposition}

		\begin{proof}
			
			\noindent (1). Induction on the structure of $\phi$.
			
			\vspace{0.2cm}
			
			\noindent (2). 	We will introduce following lemmas which can be realized easily, and set variable set $\vars = \cD(\free{\phi}\cup\var(\prog))$.	
			
			\begin{lemma}
				\label{lem sound proof 1}
				For any $\rho\in\cD(\vars)$ and terminating program $\prog$, for any variable set $S\subseteq\vars$:
				\begin{enumerate}
					\item if $S\cap\var(\prog) = \emptyset$, then $\rt{\rho}{S} = \rt{\sem{\prog}(\rho)}{S}$;
					\item if $S\supseteq\var(\prog)$, then $\sem{\prog}(\rt{\rho}{S}) = \rt{\sem{\prog}(\rho)}{S}$.
				\end{enumerate} 
			\end{lemma}
			\begin{lemma}
				\label{lem sound proof 5}
				For any command $\prog\equiv \qU$ or $\qI$, and for any $\rho_0,\rho_1\in\cD$ with disjoint domains and $\var(\prog)\subseteq\dom{\rho_0}$: $\sem{\prog}(\rho_0)\otimes\rho_1 = \sem{\prog}(\rho_0\otimes\rho_1)$
			\end{lemma}
			\begin{lemma}
				\label{lem sound proof 6}
				For any $\rho\in\cD(\vars)$ and any command $\prog\equiv \qI$ and two disjoint sets $S_1,S_2\subseteq\vars$, $\sem{\prog}(\rho)\succeq\rt{\sem{\prog}(\rho)}{S_1}\otimes\rt{\sem{\prog}(\rho)}{S_2}$ if and only if $\rho\succeq\rt{\rho}{S_1\backslash q}\otimes\rt{\rho}{S_2\backslash q}$.
			\end{lemma}
			\begin{lemma}
				\label{lem sound proof 7}
				For any $\rho\in\cD(\vars)$ and any command $\prog\equiv \qU$ and two disjoint sets $S_1,S_2\subseteq\vars$ such that $\qbar\subseteq S_1$ or $\qbar\subseteq S_2$ or $\qbar\cap(S_1\cup S_2) = \emptyset$, $\sem{\prog}(\rho)\succeq\rt{\sem{\prog}(\rho)}{S_1}\otimes\rt{\sem{\prog}(\rho)}{S_2}$ if and only if $\rho\succeq\rt{\rho}{S_1}\otimes\rt{\rho}{S_2}$.
			\end{lemma}	
			\begin{lemma}
				\label{lem sound proof 8}
				$\rho\models\phi_1\ast\phi_2$ iff $\free{\phi_1}\cap\free{\phi_2} = \emptyset$, $\rho\models\phi_1$, $\rho\models\phi_2$ and $\rho\succeq\rt{\rho}{\free{\phi_1}}\otimes\rt{\rho}{\free{\phi_2}}$.
			\end{lemma}	
			
			Now we start to prove (2) by following two statements:
			
			{\bf Statement 1:} For any $\rho\in\cD(\vars)$, if $\rho\models\phi[\qI]$, then  $\sem{\qI}(\rho)\models\phi$.
			\begin{enumerate}
				\item $\phi\equiv \bD[S]$. By definition, $\bD[S][\qI] = \bD[S]$. If $\rho\models\bD[S]$, then $\dom {\rho}\supseteq S$. Trivially, $\dom {\sem{\qI}(\rho)} = \dom {\rho}\supseteq S$, so $\sem{\qI}(\rho)\models\bD[S]$.
				\item $\phi\equiv P\in\cP$. There are two cases: 
				
				Case 1: $q\in\free{P}$, $ P[\qI] = \id_q\wedge\lceil P\rceil_q$. First observe that for any $\rho\in\cD(\vars)$,
				\begin{align*}
				\rt{\sem{\qI}(\rho)}{\free{P}} &= \rt{\left[(\cE_{\qI}\otimes\cI_{\vars\backslash\var(\qI)})(\rho)\right]}{\free{P}} \\
				&= (\cE_{\qI}\otimes\cI_{\free{P}\backslash q)})(\rt{\rho}{\free{P}}) \\
				&= \sum_n(|0\>_q\<n|\otimes\id_{\free{P}\backslash q)})\left(\rt{\rho}{\free{P}}\right)(|n\>_q\<0|\otimes\id_{\free{P}\backslash q)}) \\
				&= |0\>_q\<0|\otimes \rt{\rho}{\free{P}\backslash q}.
				\end{align*}
				If $\rho\models P[\qI]$, then $\rt{\rho}{\free{P}}\models\id_q\wedge\lceil P\rceil_q$, so $\supp(\rt{\rho}{\free{P}\backslash q})\subseteq\lceil P\rceil_q$ and $\supp(|0\>_q\<0|\otimes\rt{\rho}{\free{P}\backslash q})\subseteq P$ by definition of $\lceil P\rceil_q$, which implies $\rt{\sem{\qI}(\rho)}{\free{P}}\models P$ and thus $\sem{\qI}(\rho)\models P$ as desired. 
				
				Case 2: $q\notin\free{P}$, $ P[\qI] = P$. For any $\rho\in\cD(\vars)$, note that $\sem{\qI}$ is trace preserving and only applies on $q$, so $\free{P}\cap\var(\qI) = \emptyset$ and therefore, 
				$$\rt{\rho}{\free{P}} = \rt{\sem{\qI}(\rho)}{\free{P}}.$$
				Thus, $\rho\models P[\qI]$ iff $\rt{\rho}{\free{P}}\models P$ iff $\rt{\sem{\qI}(\rho)}{\free{P}}\models P$ iff $\sem{\qI}(\rho)\models P$.
				
				\item $\phi\equiv \unia[S]$. The only case $\unia[S][\qI]$ being defined is that $q\notin S$ and $\unia[S][\qI] = \unia[S]$. For any $\rho\in\cD(\vars)$, since $S\cap\var(\qI) = \emptyset$, so $\rt{\rho}{S} = \rt{\sem{\qI}(\rho)}{S}$. Therefore, $\rho\models\unia[S][\qI]$ iff $\rt{\rho}{S}\models \unia[S]$ iff $\rt{\sem{\qI}(\rho)}{S}\models \unia[S]$ iff $\sem{\qI}(\rho)\models \unia[S]$.

				\item $\phi \equiv \top$ or $\bot$. Trivial.
				
				\item $\phi \equiv \phi_1\wedge\!(\vee)\ \phi_2$. For any $\rho\in\cD(\vars)$, first by induction hypothesis, $\rho\models\phi_i[\qI]\Rightarrow\sem{\qI}(\rho)\models\phi_i$ for $i = 1,2$. Thus, 
				\begin{align*}
				&\rho\models\phi[\qI] \equiv \phi_1[\qI]\wedge\!(\vee)\ \phi_2[\qI] \\
				\Longrightarrow\ &\rho\models\phi_1[\qI] \text{\ and(or)\ }\rho\models\phi_2[\qI] \\
				\Longrightarrow\ &\sem{\qI}(\rho)\models\phi_1 \text{\ and(or)\ }\sem{\qI}(\rho)\models\phi_2 \\
				\Longrightarrow\ &\sem{\qI}(\rho)\models\phi_1\wedge\!(\vee)\ \phi_2.
				\end{align*}
				
				\item $\phi \equiv \phi_1\ast\phi_2$. For any $\rho\in\cD(\vars)$, first by induction hypothesis, $\rho\models\phi_i[\qI]\Rightarrow\sem{\qI}(\rho)\models\phi_i$ for $i = 1,2$. 
				\begin{itemize}
					\item[$\cdot$]Case 1: $q\notin\free{\phi_1}\cup\free{\phi_2}$, and $\phi[\qI]\equiv \phi_1[\qI] \ast \phi_2[\qI]$. So, $\rt{\rho}{\free{\phi_1}\cup\free{\phi_2}} = \rt{\sem{\qI}(\rho)}{\free{\phi_1}\cup\free{\phi_2}}$, then using induction hypothesis and Proposition \ref{pro res set} we have:
					\begin{align*}
					&\rho\models\phi_1[\qI] \ast \phi_2[\qI] \\
					\Longrightarrow\ &\free{\phi_1[\qI]}\cap\free{\phi_2[\qI]} = \emptyset,\ \rho\models\phi_1[\qI],\ \rho\models\phi_2[\qI] \text{\ and\ }\\
					& \rho\succeq\rt{\rho}{\free{\phi_1}}\otimes\rt{\rho}{\free{\phi_2}} \\
					\Longrightarrow\ &\free{\phi_1}\cap\free{\phi_2} = \emptyset,\ \sem{\qI}(\rho)\models\phi_1, \sem{\qI}(\rho)\models\phi_2 \text{\ and\ }\\
					&\sem{\qI}(\rho)\succeq\rt{\sem{\qI}(\rho)}{\free{\phi_1}}\otimes\rt{\sem{\qI}(\rho)}{\free{\phi_2}} \\
					\Longrightarrow\ &\sem{\qI}(\rho)\models\phi_1\ast\phi_2. 
					\end{align*}	
					\item[$\cdot$]Case 2: $q\in\free{\phi_1}$ and $q\notin\free{\phi_2}$, and  $\phi[\qI]\equiv (\phi_1[\qI] \wedge \phi_2[\qI])\wedge(\bD(\free{\phi_1}\backslash q)\ast\bD(\free{\phi_2})$. Following by Lemma \ref{lem sound proof 6}, we have :				
					\begin{align*}
					&\rho\models(\phi_1[\qI] \wedge \phi_2[\qI])\wedge(\bD(\free{\phi_1}\backslash q)\ast\bD(\free{\phi_2}) \\
					\Longrightarrow\ &(\free{\phi_1}\backslash q)\cap\free{\phi_2} = \emptyset,\ \rho\models\phi_1[\qI],\ \rho\models\phi_2[\qI] \text{\ and\ }\\
					& \rho\succeq\rt{\rho}{\free{\phi_1}\backslash q}\otimes\rt{\rho}{\free{\phi_2}} \\
					\Longrightarrow\ &\free{\phi_1}\cap\free{\phi_2} = \emptyset,\ \sem{\qI}(\rho)\models\phi_1, \sem{\qI}(\rho)\models\phi_2 \text{\ and\ }\\
					&\sem{\qI}(\rho)\succeq\rt{\sem{\qI}(\rho)}{\free{\phi_1}}\otimes\rt{\sem{\qI}(\rho)}{\free{\phi_2}} \\
					\Longrightarrow\ &\sem{\qI}(\rho)\models\phi_1\ast\phi_2. 
					\end{align*}			
				\end{itemize}

			\end{enumerate}

			{\bf Statement 2:} For any $\rho\in\cD(\vars)$, if $\rho\models\phi[\qU]$, then $\sem{\qU}(\rho)\models\phi$.
			\begin{enumerate}
				\item $\phi\equiv \bD[S]$. Similar to Statement 1 (1).
				
				\item $\phi\equiv P\in\cP$. There are two cases: 
				
				Case 1: $\qbar\subseteq\free{P}$, $ P[\qU] = (U^{\qbar\dag}\otimes\id_{\free{P}\backslash \qbar})P(U^{\qbar}\otimes\id_{\free{P}\backslash \qbar})$. First observe that for any $\rho\in\cD(\vars)$,			
				\begin{align*}
				\rt{\sem{\qU}(\rho)}{\free{P}} &= \rt{\left[(\cE_{\qU}\otimes\cI_{\vars\backslash\var(\qU)})(\rho)\right]}{\free{P}} \\
				&= (\cE_{\qU}\otimes\cI_{\free{P}\backslash \qbar)})(\rt{\rho}{\free{P}}) \\
				&= (U^{\qbar}\otimes\id_{\free{P}\backslash \qbar})(\rt{\rho}{\free{P}}) (U^{\qbar\dag}\otimes\id_{\free{P}\backslash \qbar})
				\end{align*}
				Therefore, we have
				\begin{align*}
				&\rho\models (U^{\qbar\dag}\otimes\id_{\free{P}\backslash \qbar})P(U^{\qbar}\otimes\id_{\free{P}\backslash \qbar}) \\
				\Longrightarrow\ &\rt{\rho}{\free{P}}\models (U^{\qbar\dag}\otimes\id_{\free{P}\backslash \qbar})P(U^{\qbar}\otimes\id_{\free{P}\backslash \qbar}) \\
				\Longrightarrow\ &\supp(\rt{\rho}{\free{P}})\subseteq (U^{\qbar\dag}\otimes\id_{\free{P}\backslash \qbar})P(U^{\qbar}\otimes\id_{\free{P}\backslash \qbar}) \\
				\Longrightarrow\ &\supp\big((U^{\qbar}\otimes\id_{\free{P}\backslash \qbar})(\rt{\rho}{\free{P}}) (U^{\qbar\dag}\otimes\id_{\free{P}\backslash \qbar})\big)\subseteq P \\
				\Longrightarrow\ &\rt{\sem{\qU}(\rho)}{\free{P}}\models P \\
				\Longrightarrow\ &\sem{\qU}(\rho)\models P
				\end{align*}
				
				Case 2: $\qbar\cap\free{P} = \emptyset$, $ P[\qU] = P$. For any $\rho\in\cD(\vars)$, note that $\sem{\qU}$ is trace preserving and $\var(\qU) = \qbar$, so $\free{P}\cap\var(\qU) = \emptyset$ and therefore, 
				$$\rt{\rho}{\free{P}} = \rt{\sem{\qU}(\rho)}{\free{P}}.$$
				Thus, $\rho\models P[\qU]$ iff $\rt{\rho}{\free{P}}\models P$ iff $\rt{\sem{\qU}(\rho)}{\free{P}}\models P$ iff $\sem{\qU}(\rho)\models P$.

				\item $\phi\equiv \unia[S]$. There are two cases:
				
				Case 1: $\qbar\subseteq S$, $\unia[S][\qU] = \unia[S]$. Similar to above arguments, for any $\rho\in\cD(\vars)$,			
				$$\rt{\sem{\qU}(\rho)}{\free{P}} = (U^{\qbar}\otimes\id_{\free{P}\backslash \qbar})(\rt{\rho}{\free{P}}) (U^{\qbar\dag}\otimes\id_{\free{P}\backslash \qbar}),$$
				and therefore,
				\begin{align*}
				&\rho\models \unia[S][\qU] \\
				\Longrightarrow\ &\rt{\rho}{S} = \frac{I_S}{\dim(S)} \\
				\Longrightarrow\ &(U^{\qbar}\otimes\id_{S\backslash \qbar})(\rt{\rho}{S}) (U^{\qbar\dag}\otimes\id_{S\backslash \qbar}) = \frac{I_S}{\dim(S)}\\
				\Longrightarrow\ &\rt{\sem{\qU}(\rho)}{S}\models \unia[S] \\ 
				\Longrightarrow\ &\sem{\qU}(\rho)\models \unia[S].
				\end{align*}
				
				Case 2: $\qbar\cap S = \emptyset$. For any $\rho\in\cD(\vars)$, as $S\cap\var(\qU) = \emptyset$, we have
				$\rt{\rho}{S} = \rt{\sem{\qU}(\rho)}{S}$, which leads to: $\rho\models\unia[S][\qU]$ iff $\rt{\rho}{S}\models \unia[S]$ iff $\rt{\sem{\qU}(\rho)}{S}\models \unia[S]$ iff $\sem{\qU}(\rho)\models \unia[S]$.
				
				\item $\phi \equiv \top$ or $\bot$. Trivial.
				
				\item $\phi \equiv \phi_1\wedge\!(\vee)\ \phi_2$. Similar to Statement 1 (5).
				
				\item $\phi \equiv \phi_1\ast\phi_2$. Either $\qbar\subseteq\free{\phi_1}$ or $\qbar\subseteq\free{\phi_2}$ or $\qbar\cap(\free{\phi_1}\cup\free{\phi_2}) = \emptyset$. So according to Lemma \ref{lem sound proof 7} and induction hypothesis we have:
				\begin{align*}
				&\rho\models\phi_1[\qU] \ast \phi_2[\qU] \\
				\Longrightarrow\ &\free{\phi_1[\qU]}\cap\free{\phi_2[\qU]} = \emptyset,\ \rho\models\phi_1[\qU],\ \rho\models\phi_2[\qU] \text{\ and\ }\\
				& \rho\succeq\rt{\rho}{\free{\phi_1[\qU]}}\otimes\rt{\rho}{\free{\phi_2[\qU]}} \\
				\Longrightarrow\ &\free{\phi_1}\cap\free{\phi_2} = \emptyset,\ \sem{\qU}(\rho)\models\phi_1, \sem{\qU}(\rho)\models\phi_2 \text{\ and\ }\\
				&\sem{\qU}(\rho)\succeq\rt{\sem{\qU}(\rho)}{\free{\phi_1}}\otimes\rt{\sem{\qU}(\rho)}{\free{\phi_2}} \\
				\Longrightarrow\ &\sem{\qU}(\rho)\models\phi_1\ast\phi_2. 
				\end{align*}
			\end{enumerate}
			
		\end{proof}

		\section{Separation Logic for Quantum Programs, Deferred Proofs for Section \ref{sec QSL}}
		\label{sec app QSL}
		
		\subsection{Proof of Theorem \ref{thm eq glb var set}}
		
		\begin{theorem}[Theorem \ref{thm eq glb var set}]
			For any two sets $\vars$ and $\vars^\prime$ containing all free variables of $\phi, \psi$ and $\prog$,
			$\vars\models\{\phi\}\prog\{\psi\} \text{\ if\ and\ only\ if\ }\vars^\prime\models\{\phi\}\prog\{\psi\}.$
		\end{theorem}
		
		\begin{proof} Suppose $\vars\subseteq \vars^\prime$.
			
			(Extension, $\Rightarrow$ part): For any $\rho\in\cD(\vars^\prime)$, if $\rho\models\phi$, then by Proposition \ref{pro BI res mon}, $\rt{\rho}{\vars}\models\phi$ as $\free{\phi}\subseteq\vars$.
			Note that $\rt{\rho}{\vars}\in\cD(\vars)$, so by assumption, 
			$\sem{\prog}(\rt{\rho}{\vars})\models\psi$. Observe that:
			\begin{align*}
			\sem{\prog}(\rt{\rho}{\vars}) &= 
			(\cE_\prog\otimes\cI_{\vars\backslash\var(\prog)})(\rt{\rho}{\vars}) \\
			&= \rt{\left[(\cE_\prog\otimes\cI_{\vars\backslash\var(\prog)}\otimes\cI_{\vars^\prime\backslash\vars})(\rho)\right]}{\vars} \\
			&= \rt{\left[(\cE_\prog\otimes\cI_{\vars^\prime\backslash\var(\prog)})(\rho)\right]}{\vars} \\
			&= \rt{\sem{\prog}(\rho)}{\vars}
			\end{align*}
			by using Proposition \ref{prop sem qo}, therefore $\sem{\prog}(\rt{\rho}{\vars})\preceq\sem{\prog}(\rho)$ and $\sem{\prog}(\rho)\models\psi$ by Kripke monotonicity.
			
			(Restriction, $\Leftarrow$ part): For any $\rho\in\cD(\vars)$, choose a $\rho^\prime\in\cD(\vars^\prime)$ such that $\rho\preceq\rho^\prime$ (so $\rho = \rt{\rho^\prime}{\vars}$). If $\rho\models\phi$, then $\rho^\prime\models\phi$, and by assumption, 
			$\sem{\prog}(\rho^\prime)\models\psi$. Observe that:
			\begin{align*}
			\sem{\prog}(\rho) &= 
			(\cE_\prog\otimes\cI_{\vars\backslash\var(\prog)})(\rt{\rho^\prime}{\vars}) \\
			&= \rt{\left[(\cE_\prog\otimes\cI_{\vars\backslash\var(\prog)}\otimes\cI_{\vars^\prime\backslash\vars})(\rho^\prime)\right]}{\vars} \\
			&= \rt{\left[(\cE_\prog\otimes\cI_{\vars^\prime\backslash\var(\prog)})(\rho^\prime)\right]}{\vars} \\
			&= \rt{\sem{\prog}(\rho^\prime)}{\vars}
			\end{align*}
			by using Proposition \ref{prop sem qo}, therefore $\sem{\prog}(\rho)\preceq\sem{\prog}(\rho^\prime)$. As $\free{\psi}\subseteq\vars$, and $\sem{\prog}(\rho)\models\psi$  by Proposition \ref{pro BI res mon}.
		\end{proof}

		\subsection{Proof of Proposition \ref{prop CM}}
		
		\begin{proposition}[Proposition \ref{prop CM}, Extended Version] The formulas generated by following grammar are ${\rm CM}$. 
			$$
			\phi,\psi ::= p\in\AP\ |\ \top\ |\ \bot\ |\ \phi\wedge\psi\ |\ \phi\in{\rm SP}\ |\ \mu\ast\phi
			$$
			where $\mu\in {\rm SP}$.
		\end{proposition}
		
		\begin{proof}
			\begin{enumerate}
				\item $p\equiv\bD[S]$. Trivial.
				\item $p\equiv P\in\cP$. Suppose $\rho,\rho^\prime\in\cD$ with same domain and $\rho\models P$ and $\rho^\prime\models P$, then 
				$$\supp\left(\rt{\rho}{\free {P}}\right)\subseteq P, \quad \supp\left(\rt{\rho^\prime}{\free {P}}\right)\subseteq P$$
				then for any $\lambda\in[0,1]$, we have:
				\begin{align*}
				\supp\left(\rt{\left(\lambda\rho + (1-\lambda)\rho^\prime\right)}{\free {P}}\right) &=  \supp\left(\rt{\left(\lambda\rho\right)}{\free {P}}\right)\sqcup\supp\left(\rt{\left((1-\lambda)\rho^\prime\right)}{\free {P}}\right) \\
				&\subseteq \supp\left(\rt{\rho}{\free {P}}\right)\sqcup \supp\left(\rt{\rho^\prime}{\free {P}}\right) \\
				&\subseteq P
				\end{align*}
				\item $p\equiv\unia[S]$. Suppose $\rho,\rho^\prime\in\cD$ with same domain and $\rho\models \unia[S]$ and $\rho^\prime\models \unia[S]$, then 
				$$\rt{\rho}{S} = \rt{\rho^\prime}{S} = \frac{I_S}{\dim(S)},$$
				and thus for any $\lambda\in[0,1]$, we have:
				$$\rt{\left(\lambda\rho + (1-\lambda)\rho^\prime\right)}{S} = \lambda\rt{\rho}{\free {S}} + (1-\lambda)\rt{\rho^\prime}{\free {S}} = \frac{I_S}{\dim(S)},$$
				and so, $\lambda\rho + (1-\lambda)\rho^\prime\models \unia[S]$.
				\item $\top$ or $\bot$. Trivial.
				\item $\phi\wedge\psi$. Suppose $\rho,\rho^\prime\in\cD$ with same domain and $\rho\models \phi\wedge\psi$ and $\rho^\prime\models \phi\wedge\psi$, then by induction hypothesis, for any $\lambda\in[0,1]$, 
				$$\lambda\rho + (1-\lambda)\rho^\prime \models\phi,\quad \lambda\rho + (1-\lambda)\rho^\prime \models\psi$$
				and thus, $\lambda\rho + (1-\lambda)\rho^\prime \models\phi\wedge\psi$.
				\item $\phi\in{\rm SP}$. If $\sem{\phi} = \emptyset$, then trivially $\phi\in{\rm CM}$. Otherwise, suppose $\sigma$ is the least element of $\sem{\phi}$, and $\rho,\rho^\prime\in\cD$ with same domain and $\rho\models \phi$ and $\rho^\prime\models \phi$, we must have: for any $\lambda\in[0,1]$,
				\begin{align*}
				\rt{\rho}{\free{\phi}} = \rt{\rho^\prime}{\free{\phi}} = \sigma\quad\Rightarrow\quad \rt{\left(\lambda\rho + (1-\lambda)\rho^\prime\right)}{\free{\phi}} = \sigma
				\end{align*}
				and so $\lambda\rho + (1-\lambda)\rho^\prime\models\phi$.
				\item $\mu\ast\phi$. Suppose $\sigma$ is the least element of $\sem{\phi}$, $\rho,\rho^\prime\in\cD$ with same domain and $\rho\models \mu\ast\phi$ and $\rho^\prime\models \mu\ast\phi$, then by induction hypothesis and \ref{lem sound proof 8}, $\free{\mu}\cap\free{\psi} = \emptyset$, for any $\lambda\in[0,1]$, 
				\begin{align*}
				&\rt{\rho}{\free{\mu\ast\phi}} = \sigma \otimes \rt{\rho}{\free{\phi}},\quad \rt{\rho^\prime}{\free{\mu\ast\phi}} = \sigma \otimes \rt{\rho^\prime}{\free{\phi}},\quad \rt{\rho}{\free{\phi}},\rt{\rho^\prime}{\free{\phi}}\models\phi \\
				\Rightarrow\ &\rt{\left(\lambda\rho + (1-\lambda)\rho^\prime\right)}{\free{\mu\ast\phi}} = \sigma \otimes \left(\lambda\rt{\rho}{\free{\phi}} + (1-\lambda)\rt{\rho^\prime}{\free{\phi}} \right),\quad \lambda\rt{\rho}{\free{\phi}} + (1-\lambda)\rt{\rho^\prime}{\free{\phi}} \models\phi
				\end{align*}
				and thus, $\lambda\rho + (1-\lambda)\rho^\prime\models\mu\ast\phi$.
			\end{enumerate}
		\end{proof}
		
		\subsection{Proof of Proposition \ref{prop glo imp}}
		
		\begin{proposition}[extended version]
			\label{prop glo imp app}
			1) For all $\phi\in{\rm Res}$ and $S\subseteq \vars$, $\models \phi\gimplr \phi\wedge \bD[S]$.
			
			2) For all $\phi,\psi$, $\models\phi\rightarrow\psi$ implies $\models\phi\gimp\psi$.
		\end{proposition}
		\begin{proof}
			1) By definition, it is sufficient to prove that for all $\rho$ with $\dom{\rho}\supseteq\free{\phi}\cup S$, $\rho\models\phi$ if and only if $\rho\models\phi\wedge\bD[S]$. This is trivial since $\rho\models\bD[S]$.
			
			2) Trivial by definition of global implication $\gimp$.
		\end{proof}
		
		\subsection{Proof of Proposition \ref{prop SP}}
		\label{sec app sub proof prop SP}
		\begin{proposition}
			\label{prop SP}
			The formulas generated by following grammar are ${\rm SP}$: 
			$$
			\phi,\psi ::= \unia[S]\ |\ p\in\cP {\rm\ of\ rank\ } 1\ |\ \top\ |\ \bot\ |\ \phi\ast\psi
			$$
			where $\cP {\rm\ of\ rank\ } 1$ consists all rank 1 projections.
		\end{proposition}
		
		\begin{proof}
			\begin{enumerate}
				\item $\unia[S]$. Trivially, $\frac{I_S}{\dim(S)}$ is the least element of $\sem{\unia[S]}$.
				\item $P\in\cP \text{\ of rank 1}$. Trivially, $P$ itself (interpreted as a pure quantum state) is the least element of $\sem{P}$.
				\item $\top$. Scalar number $1$ is the least element of $\sem{\top}$.
				\item $\bot$. Trivial.
				\item $\phi\ast\psi$. Suppose $\sigma_\phi$ and $\sigma_\psi$ are the least elements of  $\sem{\phi}$ and  $\sem{\psi}$ respectively, then it is straightforward to show $\sigma_\phi\otimes \sigma_\psi$ is the least element of $\phi\ast\psi$.
			\end{enumerate}
		\end{proof}
		
		\subsection{Proposition \ref{prop modification qo}}
		
		\begin{proposition}
			\label{prop modification qo}
			\begin{enumerate}
				\item If $\phi[\cE[\qbar]]\Mexist$, $\free{\phi[\cE[\qbar]]} = \free {\phi}$;
				\item 
				If $\phi[\cE[\qbar]]\Mexist$, then for any state $\rho\in\cD(\free{\phi}\supseteq\qbar)$, $\cE(\rho)\models\phi$ if and only if $\rho\models \phi[\cE[\qbar]]$.
			\end{enumerate}
		\end{proposition}
		\begin{proof}
			\noindent	(1). Induction on the structure of $\phi$.
			
			\noindent	(2). We prove it by induction on the structure of $\phi$.
			\begin{enumerate}
				\item[(a)] $\phi\equiv\top$ or $\bot$. Trivial.
				\item[(b)] $\phi\equiv P\in\cP$,  there are two cases.
				
				Case 1. $\qbar\cap\free{P} = \emptyset$. For any $\rho\in\cD(\qbar\cup\free{P})$, $\rt{\rho}{\free{P}} = \rt{\cE(\rho)}{\free{P}}$ and thus, $\cE(\rho)\models\phi$ if and only if $\rho\models \phi[\cE[\qbar]]$ since $P[\cE[\qbar]] = P$.
				
				Case 2. $\qbar\subseteq\free{P}$. For any $\rho\in\cD(\free{P})$, we observe:
				\begin{align*}
				&\rho\models P[\cE[\qbar]] \\
				\Longleftrightarrow\ & \rho\models\left(\big(\cE^\ast_{\qbar}\otimes\cI_{\free{P}\backslash\qbar}\big)(P^\bot)\right)^\bot \\
				\Longleftrightarrow\ & \tr\left(\rho \left(\big(\cE^\ast_{\qbar}\otimes\cI_{\free{P}\backslash\qbar}\big)(P^\bot)\right)\right) = 0 \\
				\Longleftrightarrow\ & \tr\left(\big(\cE_{\qbar}\otimes\cI_{\free{P}\backslash\qbar}\big)(\rho) P^\bot\right) = 0 \\
				\Longleftrightarrow\ &
				\tr(\cE(\rho) P) = 1\\
				\Longleftrightarrow\ &\cE(\rho)\models P.		
				\end{align*}

				\item[(c)] $\phi\wedge\psi$. By induction hypothesis, for any state $\rho\in\cD(\free{\phi\wedge\psi}\cup\qbar)$, $\cE(\rho)\models\phi\wedge\psi$ iff $\cE(\rho)\models\phi$ and $\cE(\rho)\models\psi$ iff $\rho\models \phi[\cE[\qbar]]$ and $\rho\models \psi[\cE[\qbar]]$ iff $\rho\models \phi[\cE[\qbar]]\wedge\psi[\cE[\qbar]]$ iff $\rho\models (\phi\wedge\psi)[\cE[\qbar]]$.
				\item[(d)] $\phi\vee\psi$. Similar to (c).
			\end{enumerate}
		\end{proof}

		\subsection{Proof of Theorem \ref{thm sound QSL}}
		\label{sec app sub proof thm sound QSL}
		The global variable set is denoted by $\vars$, which contains all variables of programs and formulas. We first introduce following lemma for quantum measurement:	
		\begin{lemma}
			\label{lem sound proof 4}
			For any $\rho\in\cD(\vars)$ and projective measurement $M = \{M_m\}$, if $\rho\models M_m^{\qbar}$, then performing the measurement $M[\qbar]$ will not change the state, and the outcome is $m$ with certainty. As a consequence, for any {\bf if} statement $\prog\equiv\mathbf{if}\ (\square m\cdot M[\qbar] = m \rightarrow \prog_m )\ \mathbf{fi}$, if the global state $\rho\models M_m^{\qbar}$, then $\sem{\prog}(\rho) = \sem{\prog_m}(\rho)$. 
		\end{lemma}

		\begin{proof}[Proof of Theorem \ref{thm sound QSL}]
			
			It is sufficient to show that each of the rules shown in Figure \ref{fig proof system 1}, \ref{fig proof system 2} and \ref{fig proof system 3} is sound.
			
			\begin{figure}[h]\centering
				\begin{equation*}\begin{split}
				&\textsc{Skip}\ \frac{}{\{\phi\}\mathbf{skip}\{\phi\}}\quad 
				\textsc{Init}\ \frac{\phi[\qI]\Mexist}{\left\{\phi[\qI] \right\}\qI\{\phi\}}\quad
				\textsc{Unit}\ \frac{\phi[\qU]\Mexist}{
					\{\phi[\qU]\}\qbar:=U\left[\qbar\right]\{\phi\}}\\[0.1cm]
				&\textsc{Perm}\quad \frac{}{\{\phi[\qbar^\prime\mapsto\qbar]\}\qbar:=\perm(\qbar\mapsto\qbar^\prime)[\qbar]\{\phi\}} \qquad
				\textsc{Seq}\quad
				\frac{\{\phi\}\prog_1\{\psi\}\ \ \ \ \ \ \{\psi\}\prog_2\{\mu\}}{\{\phi\}\prog_1;\prog_2\{\mu\}}\\[0.1cm]
				&\textsc{DIf}\quad
				\frac{\left\{\phi_m\right\}\prog_m\{\psi\}\ {\rm for\ all}\ m}{\big\{\bigvee_m(M_m\wedge \phi_m)\big\}\mathbf{if}\cdots\mathbf{fi}\{\psi\}}\qquad
				\textsc{DLoop}\quad
				\frac{\{\phi\}\prog\{(M_0\wedge \psi)\vee (M_1\wedge \phi)\}}{\{(M_0\wedge \psi)\vee (M_1\wedge \phi)\}\mathbf{while}\{\psi\}} \\[0.1cm]
				&\textsc{RIf}\quad\frac{\left\{\phi\ast M_m\right\}\prog_m\{\psi\}\ {\rm for\ all}\ m\quad \psi\in \text{CM}}{\{\phi\ast\id_{\qbar}\}\mathbf{if}\cdots\mathbf{fi}\{\psi\}}\qquad
				\textsc{RLoop}\quad\frac{\{\phi\ast M_1\}\prog\{\phi\ast\id_{\qbar}\}\quad\phi\in\text{CM}}{\{\phi\ast\id_{\qbar}\}\mathbf{while}\{\phi\wedge M_0\}} 
				\end{split}\end{equation*}
				\caption{Proof System QSL. In \textsc{DIf}, \textsc{DLoop}, \textsc{RIf} and \textsc{RLoop}, $\mathbf{if}\cdots\mathbf{fi}$ and $\mathbf{while}$ are abbreviations of $\mathbf{if}\ (\square m\cdot M[\qbar] = m \rightarrow \prog_m )\ \mathbf{fi}$ and $\mathbf{while}\ M[\qbar]=1\ \mathbf{do}\ \prog\ \mathbf{od}$ respectively, and  $M_0,M_1,M_m$ in assertions are regarded as projective predicates acting on $\qbar$. In \textsc{Perm}, $\perm(\qbar\mapsto\qbar^\prime)[\qbar]$ stands for the unitary transformation which permutes the variables from $\qbar$ to $\qbar^\prime$ (see Section \ref{sec basic Quantum} for details).
				}
			\end{figure}
			
			\noindent -- \textsc{Skip}. Trivial as the state of quantum variables are unchanged when applying ${\bf skip}$.
			
			\vspace{0.4cm}
			
			\noindent -- \textsc{Init}. By Proposition \ref{pro modification}.
			
			\vspace{0.4cm}
			
			\noindent -- \textsc{Unit}. By Proposition \ref{pro modification}.
			
			\vspace{0.4cm}
			
			\noindent -- \textsc{Perm}. For any input $\rho\in\cD(\vars)$ with matrix form $\ol{\rho}[\qbar,\qbar_r]$ ($\qbar_r = \vars\backslash\qbar$; i.e., $\ol{\rho}$ is a purely matrix and $[\qbar,\qbar_r]$ denotes the order of basis; that is, $\ol{\rho}[\qbar,\qbar_r]$ is interpreted as a matrix over $\cH_{\qbar}\otimes\cH_{\qbar_r}$), the output state after performing the $\perm[\qbar\mapsto\qbar^\prime]$: $\qbar:=\perm[\qbar\mapsto\qbar^\prime]$ has the matrix form $\ol{\rho}[\qbar^\prime,\qbar_r]$. Then it is not difficult to show $\ol{\rho}[\qbar,\qbar_r]\models\phi[\qbar^\prime\mapsto\qbar]$ if and only if $\ol{\rho}[\qbar^\prime,\qbar_r]\models\phi$.
			
			\vspace{0.4cm}
			
			\noindent -- \textsc{Seq}. For any $\rho\in\cD(\vars)$, if $\rho\models\phi$, then by assumptions, $\sem{\prog_1}(\rho)\models\psi$ and $\sem{\prog_2}(\sem{\prog_1}(\rho))\models\mu$. Note that $\sem{\prog_2}(\sem{\prog_1}(\rho)) = \sem{\prog_1;\prog_2}(\rho)$ as $\rho$ is a global state, so $\sem{\prog_1;\prog_2}(\rho)\models\mu$.
			
			\vspace{0.4cm}
			
			\noindent -- \textsc{DIf}. For any $\rho\in\cD(\vars)$, if $\rho\models\bigvee_m(M_m^{\qbar}\wedge \phi_m)$, then there exists at least one $m$ such that $\rho\models M_m^{\qbar}\wedge \phi_m$, and we assume it is $n$. As $\rho\models M_n^{\qbar}\wedge \phi_n$, so $\rho\models M_n^{\qbar}$ and $\rho\models\phi_n$, by Lemma \ref{lem sound proof 4}, we have 
			$$\sem{\mathbf{if}\ (\square m\cdot M[\qbar] = m \rightarrow \prog_m )\ \mathbf{fi}}(\rho) = \sem{\prog_n}(\rho),$$
			and by assumption $\{\phi_n\}\prog_n\{\psi\}$, so $\sem{\prog_n}(\rho)\models\psi$, or equivalently, $$\sem{\mathbf{if}\ (\square m\cdot M[\qbar] = m \rightarrow \prog_m )\ \mathbf{fi}}(\rho)\models\psi.$$
			
			\vspace{0.4cm}
			
			\noindent -- \textsc{DLoop}. For any input $\rho$ that satisfies $(M_0^{\qbar}\wedge \psi)\vee (M_1^{\qbar}\wedge \phi)$, with the premise $\{\phi\}\prog\{(M_0^{\qbar}\wedge \psi)\vee (M_1^{\qbar}\wedge \phi)\}$, it is indeed a deterministic loop and the measurement $M$ in guard never changes the current state (see Lemma \ref{lem sound proof 4}), i.e., the number of iterations $N(\rho)$ is deterministic and moreover,
			$$\sem{\mathbf{while}\ M[\qbar]=1\ \mathbf{do}\ \prog\ \mathbf{od}} = \sem{\prog^{N(\rho)}}(\rho)$$
			where $\prog^k \triangleq \prog;\cdots;\prog$ is the $k$-fold sequential composition of $\prog$. Soundness follows by repeatedly using the induction hypothesis.

			\vspace{0.4cm}
			
			\noindent -- \textsc{RIf}. For any input $\rho\in\cD(\vars)$ such that $\rho\models\phi\ast\id_\qbar$, it must have:
			$$\rt{\rho}{\free{\phi}\cup\qbar} = \rt{\rho}{\free{\phi}}\otimes\rt{\rho}{\qbar}.$$
			After the measurement $M$, 
			with probability $p_m = \tr(M_m^\qbar\rho M_m^\qbar)$ the outcome is $m$ and the state changes to 
			$$\rho_m = \frac{M_m^\qbar\rho M_m^\qbar}{p_m}.$$ Observe that $p_m = \tr(M_m^\qbar\rt{\rho}{\qbar} M_m^\qbar)$ and 
			\begin{align*}
			\rt{\rho_m}{\free{\phi}\cup\qbar} &= \frac{M_m^\qbar\rt{\rho}{\free{\phi}\cup\qbar} M_m^\qbar}{p_m} \\
			&= \frac{M_m^\qbar\rt{\rho}{\qbar}\otimes\rt{\rho}{\free{\phi}} M_m^\qbar}{p_m} \\
			&= \frac{M_m^\qbar\rt{\rho}{\qbar} M_m^\qbar}{p_m}\otimes\rt{\rho}{\free{\phi}}.
			\end{align*}
			Realizing that $\frac{M_m^\qbar\rt{\rho}{\qbar} M_m^\qbar}{p_m}\models M_m^\qbar$ and $\rt{\rho}{\free{\phi}}\models \phi$, we have $\rho_m\models\phi\ast M_m^\qbar$. By premise, $\sem{\prog_m}(\rho_m)\models\psi$. Back to the semantics of {\bf if} statement, we know that
			$$\sem{\mathbf{if}\ (\square m\cdot M[\qbar] = m \rightarrow \prog_m )\ \mathbf{fi}}(\rho) = \sum_mp_m\sem{\prog_m}(\rho_m)$$
			and by promise $\psi\in{\rm CM}$, so $\sem{\mathbf{if}\ (\square m\cdot M[\qbar] = m \rightarrow \prog_m )\ \mathbf{fi}}(\rho) \models\psi$. 
			
			\vspace{0.4cm}
			
			\noindent -- \textsc{RLoop}. We here use the notations similar to \cite{Ying16}, Section 3.3. Set quantum operation (and its cylinder extension) $\cE_i(\cdot) = M_i(\cdot) M_i^\dag$ for $i = 0,1$. We first claim:
			$$\textbf{Statement:} \rho\models\phi\ast\id_{\qbar} \text{\ implies\ } \cE_0(\rho)\models\phi\ast M_0,\ \cE_1(\rho)\models\phi\ast M_1, \ \sem{\prog}\circ_{c}\cE_1(\rho)\models\phi\ast I_{\qbar}$$
			by the premises and $\circ_{c}$ denote the composition of quantum operations, i.e., $(\cE\circ_{c}\cF)(\rho) = \cE(\cF(\rho))$.
			Next, by induction and the statement, we have: for all $i\ge 0$:
			$$ \rho\models\phi\ast\id_{\qbar} \text{\ implies\ } \cE_0\circ_{c}(\sem{\prog}\circ_{c}\cE_1)^i(\rho)\models\phi\ast M_0.$$
			Finally, it has been proved that (see \cite{Ying11})
			$$\sem{{\bf while}}(\rho) = \sum_{i=0}^\infty \cE_0\circ_{c}(\sem{\prog}\circ_{c}\cE_1)^i(\rho) $$
			and thus if $\rho\models\phi\ast\id_{\qbar}$, then $\sem{{\bf while}}(\rho)\models\phi$ and $\sem{{\bf while}}(\rho)\models M_0$ since $\phi,M_0\in{\rm CM}$. Therefore, $\sem{{\bf while}}(\rho)\models \phi\wedge M_0$.	
			
			\begin{figure}[h]\centering
				\begin{equation*}\begin{split}
				&\textsc{Weak}\quad  \frac{\phi\gimp\phi^\prime\quad \{\phi^{\prime}\}\prog\{\psi^{\prime}\}\quad 
					\psi^{\prime}\gimp\psi  }{\{\phi\}\prog\{\psi\}} \quad
				\textsc{Conj}\quad \frac{\{\phi_1\}\prog\{\psi_1\}\quad\{\phi_2\}\prog\{\psi_2\}}{\{\phi_1\wedge\phi_2\}\prog\{\psi_1\wedge\psi_2\}} \\[0.1cm]
				&\textsc{Disj}\quad \frac{\{\phi_1\}\prog\{\psi_1\}\quad\{\phi_2\}\prog\{\psi_2\}}{\{\phi_1\vee\phi_2\}\prog\{\psi_1\vee\psi_2\}} \qquad\qquad\
				\textsc{Const}\quad \frac{\{\phi\}\prog\{\psi\}\quad\free{\mu}\cap\var(\prog)=\emptyset}{\{\phi\wedge\mu\}\prog\{\psi\wedge\mu\}} \\[0.1cm]
				&\textsc{Frame}\quad \frac{\{\phi\}\prog\{\psi\}\quad\free{\mu}\cap\var(\prog)=\emptyset\quad\free{\psi}\cup\var(\prog)\subseteq\free{\phi}\text{\ or\ }\psi\in{\rm SP}}{\{\phi\ast\mu\}\prog\{\psi\ast\mu\}} \\[0.1cm]
				\end{split}\end{equation*}
				\caption{Proof System QSL. ${\rm SP}$: supported assertion, i.e., there is at most one element $\rho\in\cD(\free{\psi})$ that satisfies $\psi$.}
			\end{figure}
			
			\vspace{0.4cm}
			
			\noindent -- \textsc{Weak}. By premise $\phi\gimp\phi^\prime \triangleq \bD[\free{\phi}\cup\free{\phi^\prime}]\rightarrow(\phi\rightarrow\phi^\prime)$, we know that for any input $\rho\in\cD(\vars)$ that satisfies $\phi$, it must also satisfy $\phi^\prime$. By another premise $\{\phi^{\prime}\}\prog\{\psi^{\prime}\}$, then $\sem{\prog}(\rho)\models\psi^{\prime}$, and thus $\sem{\prog}(\rho)\models\psi$ by $\psi\gimp\psi^\prime$. The trick here is that $\dom{\rho} = \dom{\sem{\prog}(\rho)} \supseteq \free{\phi}\cup\free{\phi^\prime} \cup \free{\psi}\cup\free{\psi^\prime}$.
			
			\vspace{0.4cm}
			
			\noindent -- \textsc{Conj}. For any input $\rho\in\cD(\vars)$ such that $\rho\models\phi_1\wedge\phi_2$, then it must have $\rho\models\phi_1$ and $\rho\models\phi_2$. By premise and induction hypothesis, we obtain $\sem{\prog}(\rho)\models\psi_1$ and $\sem{\prog}(\rho)\models\psi_2$ and thus $\sem{\prog}(\rho)\models\psi_1\wedge\psi_2$.
			
			\vspace{0.4cm}
			
			\noindent -- \textsc{Case}. For any input $\rho\in\cD(\vars)$ such that $\rho\models\phi_1\vee\phi_2$, it must satisfy $\phi_1$ or $\phi_2$. By premise and induction hypothesis, we know that $\sem{\prog}(\rho)\models\psi_1$ or $\sem{\prog}(\rho)\models\psi_2$, that is, $\sem{\prog}(\rho)\models\psi_1\vee\psi_2$.
			
			\vspace{0.4cm}
			
			\noindent -- \textsc{Const}. For any input $\rho\in\cD(\vars)$ such that $\rho\models\phi\vee\mu$, it must satisfy $\phi$ and thus by premise and induction hypothesis, $\sem{\prog}(\rho)\models\psi$. Moreover, $\rho\models\mu$ implies  $\rt{\rho}{\free{\mu}}\models\mu$, and note that $\free{\mu}\cap\var(\prog) = \emptyset$, so $\rt{\rho}{\free{\mu}} = \rt{\sem{\prog}(\rho)}{\free{\mu}}$ by Lemma \ref{lem sound proof 1}, which leads to $\rt{\sem{\prog}(\rho)}{\free{\mu}}\models\mu$ and $\sem{\prog}(\rho)\models\mu$. Therefore, $\sem{\prog}(\rho)\models\psi\wedge\mu$.
			
			\vspace{0.4cm}
			
			\noindent -- \textsc{Frame}(1), with premise $\free{\psi}\cup\var(\prog)\subseteq\free{\phi}$. For any input $\rho\in\cD(\vars)$ such that $\rho\models\phi\ast\mu$, by Proposition \ref{lem sound proof 8}, then $\free{\phi}\cap\free{\mu}=\emptyset$, $\rho\models\phi\wedge\mu$, $\rt{\rho}{\free{\phi}\cup\free{\mu}} = \rt{\rho}{\free{\phi}} \otimes \rt{\rho}{\free{\mu}}$. Similar to \textsc{Const}, we have $\sem{\prog}(\rho)\models\psi\wedge\mu$ by first two premises. Also $\free{\psi}\cap\free{\mu}\subseteq\free{\phi}\cap\free{\mu} = \emptyset$, so it is sufficient to show $\rt{\sem{\prog}(\rho)}{\free{\psi}\cup\free{\mu}} = \rt{\sem{\prog}(\rho)}{\free{\psi}} \otimes \rt{\sem{\prog}(\rho)}{\free{\mu}}$. Observe following facts:
			\begin{align*}
			\rt{\sem{\prog}(\rho)}{\free{\phi}\cup\free{\mu}} &= \sem{\prog}(\rt{\rho}{\free{\phi}\cup\free{\mu}}) &&\hspace{-1.5cm} \text{Lemma\ \ref{lem sound proof 1}},\ \var(\prog)\subseteq\free{\phi} \\
			&= (\cE_\prog\otimes\cI_{\free{\phi}\backslash\var(\prog)}\otimes\cI_{\free{\mu}})(\rt{\rho}{\free{\phi}} \otimes \rt{\rho}{\free{\mu}}) && \\
			&= (\cE_\prog\otimes\cI_{\free{\phi}\backslash\var(\prog)})(\rt{\rho}{\free{\phi}})\otimes\rt{\rho}{\free{\mu}} &&\\
			&= \sem{\prog}(\rt{\rho}{\free{\phi}})\otimes\rt{\sem{\prog}(\rho)}{\free{\mu}} && \hspace{-1.5cm} \text{Lemma\ \ref{lem sound proof 1}},\ \var(\prog)\cap\free{\mu}=\emptyset\\
			&= \rt{\sem{\prog}(\rho)}{\free{\phi}}\otimes\rt{\sem{\prog}(\rho)}{\free{\mu}} && \hspace{-1.5cm} \text{Lemma\ \ref{lem sound proof 1}},\ \var(\prog)\subseteq\free{\phi}
			\end{align*}
			and by the downwards closed property of $\circ$ ($\otimes$), using $\free{\psi}\subseteq\free{\phi}$, we obtain
			\begin{align*}
			\rt{\sem{\prog}(\rho)}{\free{\psi}\cup\free{\mu}} &= \rt{\left[\rt{\sem{\prog}(\rho)}{\free{\phi}\cup\free{\mu}}\right]}{\free{\psi}\cup\free{\mu}} \\
			&= \rt{\left[\rt{\sem{\prog}(\rho)}{\free{\phi}}\otimes\rt{\sem{\prog}(\rho)}{\free{\mu}}\right]}{\free{\psi}\cup\free{\mu}} \\
			&= \rt{\left[\rt{\sem{\prog}(\rho)}{\free{\phi}}\right]}{\free{\psi}}\otimes\rt{\sem{\prog}(\rho)}{\free{\mu}} \\
			&= \rt{\sem{\prog}(\rho)}{\free{\psi}}\otimes\rt{\sem{\prog}(\rho)}{\free{\mu}}.
			\end{align*}

			\vspace{0.4cm}
			
			\noindent -- \textsc{Frame}(2), with premise $\psi\in{\rm SP}$. Unlike the previous proofs, this rule is highly nontrivial, at least in the sense of proof of itself. Given the output a singleton, there are many unreleased properties of the program $\prog$. One technique we used here is the \emph{purification}, which allows us to associate pure states with mixed states. 
			
			\begin{fact}
				Given any density operator $\rho_A$ of the system $A$, and introduce another system $R$, often called the reference system. If the dimension of $R$ is larger than or equal to $A$, then there exists a pure state $|\psi\>_{AR}$ over the composite system $AR$, such that:
				$$\tr_R(|\psi\>_{AR}\<\psi|) = \rho_A.$$
				Generally, such purifications are not unique, but they are related by a local unitary of reference system $R$. In detail, for any purifications $|\psi\>_{AR}$ and $|\psi^\prime\>_{AR}$ of $\rho_A$, there exists a unitary transformation $U_R$ acting on system $R$, such that:
				$$|\psi^\prime\>_{AR} = (\id_A\otimes U_R)|\psi\>_{AR}.$$
			\end{fact}
			
			\noindent\textbf{Step 1}: Let us first reveal some variable information from the rule itself. If there exists some input satisfies $\phi\ast\mu$, then obviously, $\free{\phi}\cap\free{\mu} = \emptyset$; otherwise, the rule is trivially sound. From the promise $\free{\mu}\cap\var(\prog) = \emptyset$, we know that $\free{\mu}\cap(\free{\phi}\cup\var(\prog)) = \emptyset$, thus without loss of generality, we can assume $\free{\phi}\subseteq\var(\prog)$, as we can always add all the variables in $\free{\phi}\backslash\var(\prog)$ to the program and left them unchanged.
			
			Moreover, as $\psi$ is a singleton formula, we must have $\free{\psi}\subseteq\free{\phi}\cup\var(\prog)$. To see this, suppose $q\in\free{\psi}\backslash(\free{\phi}\cup\var(\prog))$, then the input state is free on $q$ and the state of $q$ remains unchanged after executing $\prog$, so the output state on $q$ is not unique, which is contradictory to the premise that $\psi$ is a singleton formula.
			
			In summary, it is sufficient to prove the soundness when $\free{\phi},\free{\psi}\subseteq\var(\prog)$. To simply the representation, we use $A$ to denote $\free{\phi}$, $A^\prime$ for $\free{\psi}$, $B$ for $\var(\prog)\backslash \free{\phi}$, $C$ for $\free{\mu}$, as illustrated in Figure \ref{fig proof of FrameS 1}.
			
			\begin{figure}
				\centering
				\includegraphics[width=0.8\linewidth]{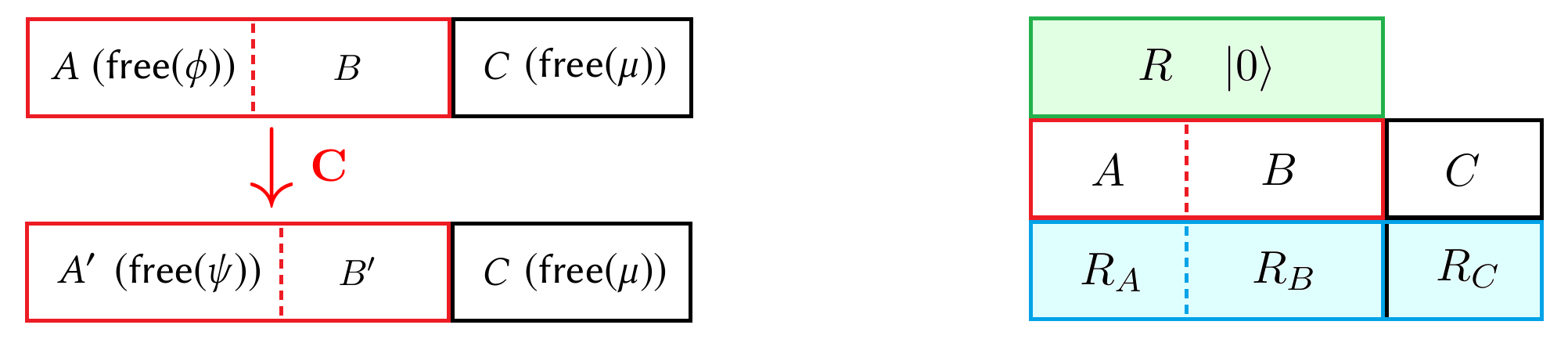}
				\caption{Variables. The left one illustrates the relations of variables, where $\var(\prog)$ is the area in red frame, see step 1. The right one shows the auxiliary systems, including $R_A$ used for purify $\rho_A$, $R_C$ used for purify $\rho_C$, $R_AR_BR_C$ used for purify the input $\rho_{ABC}$, and $R$ is the ancilla system with initial state $|0\>$ used to purify the quantum operation $\cE_{AB}$, see step 2 and step 3. }
				\label{fig proof of FrameS 1}
			\end{figure}
			
			\noindent\textbf{Step 2}: Extract hidden information from premise $\{\phi\}\prog\{\psi\}$.  In detail, we prove the following Lemma \ref{lem proof of FrameS 1}.
			
			Let us first extend the system $ABC$ with several mathematically ancilla system $R_A,R_B,R_C$ and $R$. The dimension of $R_A$ is the same as $A$ and it is used for purify density matrix of $A$, and the similar for $R_B$ and $R_C$. System $R$ is used for related the semantics function $\sem{\prog}$ (a quantum operation acting on $AB$) to a unitary transformation $U_{RAB}$ acting on $RAB$; in detail, for any input density operator $\rho_{AB}\in\cD(AB)$, the output $\sem{\prog}(\rho_{AB})$ can be obtained by following step:
			1. initial system $R$ in $|0\>_R$; 2. apply unitary transformation $U_{RAB}$ on $RAB$; 3. trace out the system $R$; or equivalently:
			$$\sem{\prog}(\rho_{AB}) = \tr_{R}\left( U_{RAB} (|0\>_R\<0|\otimes\rho_{AB}) U_{RAB}^\dag \right).$$
			
			Suppose $\rho_{A}\in\cD(A)$ is any state that satisfies $\phi$, and $\sigma_{A^\prime}\in\cD(A^\prime)$ with diagonal decomposition $\sigma_{A^\prime} = \sum_i\lambda_i|\alpha_i\>_{A^\prime}\<\alpha_i|$ is the only state on $A^\prime$ that satisfies $\psi$.
			Assuming $|\Psi_{AR_A}\> \in \cH(AR_A)$ is a purification of $\rho_{A}$. We now prove that, 
			\begin{lemma}
				\label{lem proof of FrameS 1}
				For any pure state $|\Psi_{BR_BCR_C}\> \in \cH(BR_BCR_C)$, and any unitary transformation $U_{R_ABR_BR_C}$ acting on $R_ABR_BR_C$:
				\begin{align}
				\label{eqn proof of FrameS 1}
				&\rt{\left\{\sem{\prog}\left[U_{R_ABR_BR_C}(|\Psi_{AR_A}\>\<\Psi_{AR_A}|\otimes|\Psi_{BR_BCR_C}\>\<\Psi_{BR_BCR_C}|)U^\dag_{R_ABR_BR_C}\right]\right\}}{A^\prime C} \\
				=\ &\rt{\left\{U_{RAB}U_{R_ABR_BR_C}(|0\>_R\<0|\otimes|\Psi_{AR_A}\>\<\Psi_{AR_A}|\otimes|\Psi_{BR_BCR_C}\>\<\Psi_{BR_BCR_C}|)U^\dag_{R_ABR_BR_C}U_{RAB}^\dag\right\}}{A^\prime C} \\
				=\ &\sigma_{A^\prime}\otimes\sigma_C
				\end{align}
				where $$\sigma_C = \rt{|\Psi_{BR_BCR_C}\>\<\Psi_{BR_BCR_C}|}{C}.$$
			\end{lemma}
			Realize that the input state indeed satisfies $\phi$ because
			\begin{align*}
			\rt{\left[U_{R_ABR_BR_C}(|\Psi_{AR_A}\>\<\Psi_{AR_A}|\otimes|\Psi_{BR_BCR_C}\>\<\Psi_{BR_BCR_C}|)U^\dag_{R_ABR_BR_C}\right]}{A} = \rt{|\Psi_{AR_A}\>\<\Psi_{AR_A}|}{A} = \rho_A,
			\end{align*}
			so after applying $\prog$, the reduced state over $A^\prime$ of the output must be $\sigma_{A^\prime}$. That is: for fixed $U_{R_ABR_BR_C}$ and $|\Psi_{AR_A}\>$, the output pure state must have the form
			\begin{align*}
			U_{RAB}U_{R_ABR_BR_C}(|0\>_R|\Psi_{AR_A}\>|\Psi_{BR_BCR_C}\> = \sum_i\sqrt{\lambda_i}|\alpha_i\>_{A^\prime} X_i\otimes I_C |\Psi_{BR_BCR_C}\>
			\end{align*}
			where $X_i$ are isometries mapping from $BR_BR_C$ to $RR_AB^\prime R_BR_C$, and satisfies:
			\begin{align*}
			\forall\, i:\ \|X_i\otimes I_C |\Psi_{BR_BCR_C}\>\| = 1,\quad 
			\forall\, i\neq j:\ \<\Psi_{BR_BCR_C}|(X_j^\prime \otimes I_C)(X_i\otimes I_C) |\Psi_{BR_BCR_C}\>.
			\end{align*}
			Since $|\Psi_{BR_BCR_C}\>$ are freely chosen, so it can range over all pure state over $\cH(BR_BCR_C)$, we must have:
			\begin{align}
			\label{eqn proof of FrameS 2}
			\forall\, i:\ X_i^\dag X_i = I_{BR_BR_C},\quad 
			\forall\, i\neq j:\ X_j^\dag X_i = 0.
			\end{align}
			Let $\{|e_k\>\}$ be an orthonormal basis of $BR_BR_C$, then $X_i$ has the explicit form:
			$$X_i = \sum_k|\beta_{ik}\>\<e_k|$$
			where $|\beta_{ik}\>\in\cH(RR_AB^\prime R_BR_C)$ and may not be normalized. However, by Eqn. (\ref{eqn proof of FrameS 2}), there are many constrains of $|\beta_{ik}\>$:
			\begin{align*}
			&\forall\, i:\ X_i^\dag X_i = \sum_{kk^\prime} |e_{k^\prime}\>\<e_{k}|\cdot\<\beta_{ik^\prime}|\beta_{ik}\> = \sum_k|e_k\>\<e_k| \quad\Rightarrow\quad \forall\, i:\ \<\beta_{ik^\prime}|\beta_{ik}\> = \left\{\begin{array}{ll}1 & k = k^\prime \\ 0 & k\neq k^\prime\end{array}\right. \\
			&\forall\, i\neq j:\ X_j^\dag X_i = \sum_{kk^\prime} |e_{k^\prime}\>\<e_{k}|\cdot\<\beta_{jk^\prime}|\beta_{ik}\> = 0 \quad\Rightarrow\quad \forall\, i\neq j,\ \forall\, k,k^\prime:\ \<\beta_{jk^\prime}|\beta_{ik}\> = 0.
			\end{align*}
			As a consequence, $\{|\beta_{ik}\>\}_{ik}$ is a orthonormal set, and we may extend it as an orthonormal basis of $RR_AB^\prime R_BR_C$: $\{|\beta_{ik}\>, |f_m\>\}.$ Now, let us start to calculate the explicit form of output:
			\begin{align*}
			&\rt{\left\{U_{RAB}U_{R_ABR_BR_C}(|0\>_R\<0|\otimes|\Psi_{AR_A}\>\<\Psi_{AR_A}|\otimes|\Psi_{BR_BCR_C}\>\<\Psi_{BR_BCR_C}|)U^\dag_{R_ABR_BR_C}U_{RAB}^\dag\right\}}{A^\prime C} \\
			=\ &\tr_{RR_AB^\prime R_BR_C}\left\{
			\sum_{ii^\prime}\sum_{kk^\prime}\sqrt{\lambda_i\lambda_{i^\prime}}|\alpha_i\>_{A^\prime}\<\alpha_{i^\prime}|\otimes |\beta_{ik}\>\<e_k|\Psi_{BR_BCR_C}\>\<\Psi_{BR_BCR_C}|e_{k^\prime}\>\<\beta_{i^\prime k^\prime}|
			\right\} \\
			=\ &\sum_{i^{\prime\prime}k^{\prime\prime}}\<\beta_{i^{\prime\prime}k^{\prime\prime}}|\left\{
			\sum_{ii^\prime}\sum_{kk^\prime}\sqrt{\lambda_i\lambda_{i^\prime}}|\alpha_i\>_{A^\prime}\<\alpha_{i^\prime}|\otimes |\beta_{ik}\>\<e_k|\Psi_{BR_BCR_C}\>\<\Psi_{BR_BCR_C}|e_{k^\prime}\>\<\beta_{i^\prime k^\prime}|
			\right\}|\beta_{i^{\prime\prime}k^{\prime\prime}}\> \\
			& + \sum_{m}\<f_m|\left\{
			\sum_{ii^\prime}\sum_{kk^\prime}\sqrt{\lambda_i\lambda_{i^\prime}}|\alpha_i\>_{A^\prime}\<\alpha_{i^\prime}|\otimes |\beta_{ik}\>\<e_k|\Psi_{BR_BCR_C}\>\<\Psi_{BR_BCR_C}|e_{k^\prime}\>\<\beta_{i^\prime k^\prime}|
			\right\}|f_m\> \\
			=\ &\sum_{ik}\left\{
			\sqrt{\lambda_i\lambda_{i}}|\alpha_i\>_{A^\prime}\<\alpha_{i}|\otimes \<e_k|\Psi_{BR_BCR_C}\>\<\Psi_{BR_BCR_C}|e_{k}\>\right\} + 0 \\
			=\ &\left(\sum_i\lambda_i|\alpha_i\>_{A^\prime}\<\alpha_{i}|\right)\otimes \left(\sum_k\<e_k|\Psi_{BR_BCR_C}\>\<\Psi_{BR_BCR_C}|e_{k}\>\right)\\
			=\ &\sigma_{A^\prime}\otimes\sigma_C
			\end{align*}
			
			\noindent\textbf{Step 3}: purification of all possible input. In detail, we show that for any input state $\rho_{ABC}\models\phi\ast\mu$, it can must be written into the form of Lemma \ref{lem proof of FrameS 1}; that is, there must exists $|\Psi_{AR_A}\>\in\cH(AR_A)$, $|\Psi_{BR_BCR_C}\> \in \cH(BR_BCR_C)$ and unitary transformation $U_{R_ABR_BR_C}$ acting on $R_ABR_BR_C$, such that \begin{equation*}
			\rt{\left[U_{R_ABR_BR_C}(|\Psi_{AR_A}\>\<\Psi_{AR_A}|\otimes|\Psi_{BR_BCR_C}\>\<\Psi_{BR_BCR_C}|)U^\dag_{R_ABR_BR_C}\right]}{ABC} = \rho_{ABC}.
			\end{equation*}
			
			This step is relatively simple if we realize the fact of freedom of purification.
			We use the notations $\rho_A = \tr_{BC}(\rho_{ABC}),\ \rho_C = \tr_{AB}(\rho_{ABC})$ and $\rho_{AC} = \tr_{B}(\rho_{ABC})$, and trivially $\rho_A\models\phi$ and $\rho_C\models\mu$ by restriction.
			First, there always exists pure state $|\Psi_{ABCR_AR_BR_C}\>$ that purify $\rho_{ABC}$. Next, let us focus on system $AC$: note that $\rho_{AC}$ is a product state between $A$ and $C$ as $\rho_{ABC}\models\phi\ast\mu$ (so $\rho_{AC}\models\phi\ast\mu$), thus $\rho_{AC} = \rho_A\otimes\rho_C$, and set $|\Psi_{AR_A}\>$ and $|\Psi_{CR_C}\>$ being the purification of $\rho_A$ and $\rho_C$, then $|\Psi_{AR_A}\>|\Psi_{CR_C}\>$ is also a purification of $\rho_{AC}$. If we add $|0\>_{BR_B}$, trivially $|\Psi_{AR_A}\>|0\>_{BR_B}|\Psi_{CR_C}\>$ is still a purification of $\rho_{AC}$. Since $|\Psi_{ABCR_AR_BR_C}\>$ is a purification of $\rho_{ABC}$, it is also a purification of $\rho_{AC}$. Now, we have two purifications $|\Psi_{AR_A}\>|0\>_{BR_B}|\Psi_{CR_C}\>$ and $|\Psi_{ABCR_AR_BR_C}\>$ with the same reference system $R_ABR_BR_C$, and due to the freedom of purification, there exists a local unitary transformation $U_{R_ABR_BR_C}$ that related these two purifications, i.e., 
			$$
			|\Psi_{ABCR_AR_BR_C}\> = U_{R_ABR_BR_C}|\Psi_{AR_A}\>|0\>_{BR_B}|\Psi_{CR_C}\>.
			$$
			Set $|\Psi_{BR_BCR_C}\> = |0\>_{BR_B}|\Psi_{CR_C}\>$ and we will obtain:
			\begin{align*}
			\rho_{ABC} &= \tr_{R_AR_BR_C}\left[|\Psi_{ABCR_AR_BR_C}\>\<\Psi_{ABCR_AR_BR_C}|\right] \\
			&= \rt{\left[U_{R_ABR_BR_C}(|\Psi_{AR_A}\>\<\Psi_{AR_A}|\otimes|\Psi_{BR_BCR_C}\>\<\Psi_{BR_BCR_C}|)U^\dag_{R_ABR_BR_C}\right]}{ABC}
			\end{align*}
			as we desired.
			
			\noindent\textbf{Step 4}: Combine Step 2 and 3 to conclude the soundness. For any $\rho_{ABC}\models\phi\ast\mu$, we have the following equations:
			\begin{align*}
			&\rt{\sem{\prog}(\rho_{ABC})}{A^\prime C} \\
			=\ &\rt{\left\{\sem{\prog}\left[\rt{\left(U_{R_ABR_BR_C}(|\Psi_{AR_A}\>\<\Psi_{AR_A}|\otimes|\Psi_{BR_BCR_C}\>\<\Psi_{BR_BCR_C}|)U^\dag_{R_ABR_BR_C}\right)}{ABC}\right]\right\}}{A^\prime C} \\
			=\ &\rt{\left\{\sem{\prog}\left[U_{R_ABR_BR_C}(|\Psi_{AR_A}\>\<\Psi_{AR_A}|\otimes|\Psi_{BR_BCR_C}\>\<\Psi_{BR_BCR_C}|)U^\dag_{R_ABR_BR_C}\right]\right\}}{A^\prime C} \\
			=\ &\sigma_{A^\prime}\otimes\rho_C
			\end{align*}
			by using Lemma \ref{lem sound proof 1} \textit{2}, $A^\prime C\subseteq ABC$ and 
			$$\sigma_C = \rt{|\Psi_{BR_BCR_C}\>\<\Psi_{BR_BCR_C}|}{C} = \rho_C.$$
			Since $\sigma_{A^\prime} \models \psi$ and $\rho_C\models\mu$, so $\sem{\prog}(\rho_{ABC})\models \psi\ast\mu$ as we desired.

			\begin{figure}[h]\centering
				\begin{equation*}\begin{split}
				&\textsc{UnCR}\quad \frac{\{\phi\}\prog\{\psi\}\quad\qbar\cap\var(\prog)=\emptyset\quad\phi[\cE[\qbar]]\Mexist\quad\psi[\cE[\qbar]]\Mexist}{\{\phi[\cE[\qbar]]\}\prog\{\psi[\cE[\qbar]]\}} 		
				\end{split}
				\end{equation*}
				\caption{Proof rule for dealing with entangled predicates. 
				}
			\end{figure}
			
			\vspace{0.4cm}
			
			\noindent -- \textsc{UnCR}. For $\rho\in\cD(\vars)$ such that $\rho\models\phi[\cE[\qbar]]$, by Proposition \ref{prop modification qo} (2), $\cE(\rho)\models\phi$, by premise $\{\phi\}\prog\{\psi\}$, $\sem{\prog}(\cE(\rho))\models\psi$. By the premise $\qbar\cap\var(\prog)=\emptyset$, 
			$\sem{\prog}(\cE(\rho)) = \cE(\sem{\prog}(\rho))$, and thus $\cE(\sem{\prog}(\rho))\models\psi$. Using Proposition \ref{prop modification qo} (2) again, we have $\sem{\prog}(\rho)\models\psi[\cE[\qbar]]$.
			
		\end{proof}

		\subsection{More explanations for reasoning about entangled predicates}
		\label{sec app LR}
		
		The technique for reasoning about entangled predicates proposed in \cite{YZL18} can be  described in following three steps:
		\begin{enumerate}
			\item \textit{\textbf{Pushing out}}: Introduce auxiliary variables for local reasoning;
			To capture the behavior how a program affects the entanglement relations between other systems, it is needed to introduce auxiliary variables, at most the fresh copy of each quantum variable as auxiliary variables, in both pre- and post-conditions\footnote{It is not surprising a fresh copy of all variables is enough, since the semantic function is a quantum operation -- quantum operation can be realized by a unitary transformation acting on both prime system and environment as large as the prime system. }. Thanks to Theorem \ref{thm eq glb var set}, this step can be safely down in our logic.
			\item \textit{\textbf{Modification}}: Choose appropriate quantum operation $\mathcal{E}$ and apply rule \textsc{UnCR} on the auxiliary variables;
			We use frame rule to glue all preconditions (postcondition) derived by local reasoning together to obtain a valid judgment $\{\phi\}\prog\{\psi\}$. Remember that $\phi$ and $\psi$ contains the same set of auxiliary variables. Suppose $\cE$ is an arbitrary quantum operation acting on auxiliary registers $\qbar^\prime$ and both $\cE[\qbar^\prime](\phi), \cE[\qbar^\prime](\psi)$ are defined, for any $\rho\models\cE[\qbar^\prime](\phi)$, we know that: $\cE(\rho)\models\phi$ and so $\sem{\prog}(\cE(\rho))\models\psi$. Since $\sem{\prog}$ and $\cE[\qbar^\prime]$ are two quantum operations acting on disjoint registers (i.e., the prime system and auxiliary system), $\sem{\prog}(\cE(\rho)) = \cE(\sem{\prog}(\rho))$, and thus $\cE(\sem{\prog}(\rho))\models\psi$ and hence $\sem{\prog}(\rho)\models\cE[\qbar^\prime](\psi)$. In summary, the judgment $\{\cE[\qbar^\prime](\phi)\}\prog\{\cE[\qbar^\prime](\psi)\}$ is valid, which we named it rule \textsc{UnCR}.
			\item \textit{\textbf{Pulling back}}: Using certain equivalence predicates to link the principal and auxiliary variables and then trace out the auxiliary variables. 
			There exist some BI formulas which have a globally equivalent form but with less variables, e.g. the cases of Proposition \ref{prop glo imp}. Generally, this step is not feasible for all predicates; however, with a proper choice of $\cE[\qbar^\prime]$, this step is suitable for lots of scenarios, in particular, it is feasible for all postcondition $\psi$ being projection or observable.
		\end{enumerate}
		
		\subsection{\textsc{PEPR}: simplified rule for projections}
		\label{rule PERP}
		
		When BI formula $\phi$ and $\psi$ appeared in \textsc{UnCR} are projections, we can derive the following rule by combining the ideas of modification and pulling back: 
		\begin{equation*}
		\textsc{PEPR}\ \frac{\{\Psi\}\prog\{\Phi\}\quad \var(\prog)\subseteq \qbar\quad \ptype(\qbar) = \ptype(\qbar^\prime)\quad \qbar\cap\qbar^\prime=\emptyset \quad \Psi\in\cP(\qbar,\qbar^\prime)\quad\Phi\in{\rm MES}(\qbar,\qbar^\prime)\quad Q\in\cP(\qbar)}{\left\{\proj\left[\dim(\qbar)\cdot\tr_{\qbar^\prime}\left((Q_{\qbar^\prime}\otimes\id_{\qbar})\Psi(Q_{\qbar^\prime}\otimes\id_{\qbar})\right)\right]\right\}\prog\{Q\}}	
		\end{equation*}
		as an instance of rule \textsc{UnCR}, where ${\rm MES}(\qbar,\qbar^\prime)$ stands for the set of maximally entangled states of two disjoint registers $\qbar,\qbar^\prime$ with the same type (i.e., the same dimension of their Hilbert space), and $\proj$ maps an observable to the projection onto its eigenspace of eigenvalue 1 (see Section \ref{sec basic Quantum app}). 
		\begin{proof}
			The simplified rule \textsc{PEPR} is indeed a combination of Modification and Pulling back discussed in Section \ref{sec app LR} whenever postcondition $Q\in\cP(\qbar)$. Suppose $Q$ has the diagonal decomposition $Q = \sum_i\lambda_i|\beta_i\>_\qbar\<\beta_i|$, and set $Q_{\qbar^\prime} = \sum_i\lambda_i|\beta_i\>_{\qbar^\prime}\<\beta_i|$, $N = \dim(\qbar)$ and $\Phi$ being the maximally entangled state $|\Phi\> = \frac{1}{\sqrt{N}}\sum_j|j\>_\qbar|j\>_{\qbar^\prime}$ with $\{|j\>\}$ being orthonormal basis of $\cH_\qbar$ and choose quantum operation \footnote{$\cE$ defined here is not a trace-preserving quantum operation, similar problems happens in this paragraph. However, it is always possible to add a scalar factor and this does not affect the conclusions. }
			$$\cE(\rho) = N\cdot \sum_{ij}\lambda_i|\ol{\beta_i}\>_{\qbar^\prime}\<j|\rho|j\>_{\qbar^\prime}\<\ol{\beta_i}|,$$
			then $\Phi[\cE[\qbar^\prime]] = Q_\qbar\otimes\id_{\qbar^\prime}$. First, follows by \cite{YZL18}, $\models\{\sem{\prog}^\ast(\Phi)\}\prog\{\Phi\}$ and moreover, $\sem{\prog}^\ast(\Phi)$ is the weakest precondition to make it valid. By employing rule \textsc{UnCR}, we obtain
			$$\{(\sem{\prog}^\ast(\Phi))[\cE[\qbar^\prime]]\}\prog\{\Phi[\cE[\qbar^\prime]]\}\qquad\text{or equvalently}\qquad
			\{(\sem{\prog}^\ast(\Phi))[\cE[\qbar^\prime]]\}\prog\{Q_\qbar\otimes\id_{\qbar^\prime}\}
			$$
			On the other hand, by premise $\{\Psi\}\prog\{\Phi\}$, we have
			$$\Psi\sqsubseteq\sem{\prog}^\ast(\Phi)$$
			since $\sem{\prog}^\ast(\Phi)$ is the weakest precondition and thus,
			\begin{align*}
			&\proj\left[N\cdot\tr_{\qbar^\prime}\left((Q_{\qbar^\prime}\otimes\id_{\qbar})\Psi(Q_{\qbar^\prime}\otimes\id_{\qbar})\right)\right]\otimes\id_{\qbar^\prime} \\
			\sqsubseteq\ & \proj\left[N\cdot\tr_{\qbar^\prime}\left((Q_{\qbar^\prime}\otimes\id_{\qbar})\sem{\prog}^\ast(\Phi)(Q_{\qbar^\prime}\otimes\id_{\qbar})\right)\right]\otimes\id_{\qbar^\prime} \\
			=\ & N\cdot(Q_{\qbar^\prime}\otimes\id_{\qbar})\sem{\prog}^\ast(\Phi)(Q_{\qbar^\prime}\otimes\id_{\qbar}) = (\sem{\prog}^\ast(\Phi))[\cE[\qbar^\prime]].
			\end{align*}
			or equivalently, $\proj\left[N\cdot\tr_{\qbar^\prime}\left((Q_{\qbar^\prime}\otimes\id_{\qbar})\Psi(Q_{\qbar^\prime}\otimes\id_{\qbar})\right)\right]\otimes\id_{\qbar^\prime}\rightarrow(\sem{\prog}^\ast(\Phi))[\cE[\qbar^\prime]]$ by Proposition \ref{prop axiom projection extend} (2). Then by rule \textsc{Weak} and Proposition \ref{prop glo imp app}, we have
			$$\vdash \{\proj\left[N\cdot\tr_{\qbar^\prime}\left((Q_{\qbar^\prime}\otimes\id_{\qbar})\Psi(Q_{\qbar^\prime}\otimes\id_{\qbar})\right)\right]\otimes\id_{\qbar^\prime}\}\prog\{Q_\qbar\otimes\id_{\qbar^\prime}\}$$
			According to Proposition \ref{prop axiom projection extend}(3) and \ref{prop axiom projection extend}(1) and applying
			rule \textsc{Weak} again, we conclude $$\vdash\left\{\proj\left[N\cdot\tr_{\qbar^\prime}\left((Q_{\qbar^\prime}\otimes\id_{\qbar})\Psi(Q_{\qbar^\prime}\otimes\id_{\qbar})\right)\right]\right\}\prog\{Q\}.$$
		\end{proof}
		
		\subsection{Verification of Example \ref{exam verify entanglement}}
		\label{sec app sub verification exam verify entanglement}
		First let us define the maximally entangled states $\Phi^\pm_{pq}$ as $|\Phi\>_{pq} = \frac{1}{\sqrt{2}}(|0\>_p|0\>_q \pm |1\>_p|1\>_q$. The program $\prog$ defined in Example \ref{exam verify entanglement} is:
		$$\prog\equiv {\sqrt{Z}}[q_1]; {\sqrt{Z}}[q_2].$$
		We aim to prove $\{\Phi^+_{q_1q_2}\}\prog\{\Phi^-_{q_1q_2}\}$. As discussed above, our proof has following steps:
		\begin{itemize}
			\item Local reasoning, pushing out. For subprogram ${\sqrt{Z}}[q_1]$, let us introduce an auxiliary qubit $q_1^\prime$. Using rule \textsc{Unit}, we have:
			$$\{\Psi_{q_1q_1^\prime}\}{\sqrt{Z}}[q_1]\{\Phi^+_{q_1q_1^\prime}\}\qquad\text{with}\qquad |\Psi\>_{pq} = \frac{1}{\sqrt{2}}(|0\>_p|0\>_q - \bi |1\>_p|1\>_q).$$
			Similarly, we have $\{\Psi_{q_2q_2^\prime}\}{\sqrt{Z}}[q_2]\{\Phi^+_{q_2q_2^\prime}\}$. With \textsc{Frame}, \textsc{Seq} and \textsc{Weak} and Proposition \ref{prop axiom projection}, we obtain:
			\begin{align*}
			&\vdash\{\Psi_{q_1q_1^\prime}\otimes\Psi_{q_2q_2^\prime}\}\{\Psi_{q_1q_1^\prime}\wedge\Psi_{q_2q_2^\prime}\}{\sqrt{Z}}[q_1]\{\Phi^+_{q_1q_1^\prime}\wedge\Psi_{q_2q_2^\prime}\}{\sqrt{Z}}[q_2]\{\Phi^+_{q_1q_1^\prime}\wedge\Phi^+_{q_2q_2^\prime}\}\{\Phi^+_{q_1q_1^\prime}\otimes\Phi^+_{q_2q_2^\prime}\}.
			\end{align*}
			\item modification and pulling back. Note that $\Phi^+_{q_1q_1^\prime}\otimes\Phi^+_{q_2q_2^\prime}\in{\rm MES}[q_1q_2, q_1^\prime q_2^\prime]$, we apply rule \textsc{PEPR} (an instance of \textsc{UnCR}, see Section \ref{rule PERP}) to obtain:
			$$\left\{\proj\left[4\cdot\tr_{q_1^\prime q_2^\prime}\left((\Phi^-_{q_1^\prime q_2^\prime}\otimes\id_{q_1q_2})\Psi(\Phi^-_{q_1^\prime q_2^\prime}\otimes\id_{q_1q_2})\right)\right]\right\}\prog\{\Phi^-_{q_1q_2}\}.$$
			A careful calculation shows that the precondition is exactly $\Phi^+_{q_1q_2}$ as we desired.
		\end{itemize}
		
		\section{Local Reasoning: Analysis of Variational
			Quantum Algorithms, details for Section \ref{sec-vqa-exam}}
		\label{sec app vqa}
		
		\subsection{VQA in the Tutorial of Cirq}
		\label{sec app sub VQA in the Tutorial of Cirq}
		
		The VQA presented in the tutorial of Cirq \footnote{\url{https://quantumai.google/cirq/tutorials/variational_algorithm}} deals with a 2D $+/-$ Ising model of size $N\times N$ with objective Hamiltonian (observable)
		$$H = \sum_{(i,j)}h_{ij}Z_{ij} + \sum_{(i,j;i^\prime,j^\prime)\in S}J_{ij;i^\prime j^\prime}Z_{ij}Z_{i^\prime j^\prime},$$
		where each index pair $(i,j)$ is associated with a vertex in a the $N\times N$ grid, $S$ is the set of all neighboring vertices in the grid, and all $h_{ij}$ and $J_{ij;i^\prime j^\prime}$ are either $+1$ or $-1$.
		The algorithm for preparing the ansatz state with real parameters $(\alpha,\beta,\gamma)$ given in the tutorial of Cirq can be rewritten in the quantum-{\bf while} language as follows:
		\begin{align*}
		{\rm VQA}(N)\equiv\ &{\bf for\ } j = 1,\cdots,N {\bf\ do}\ {\rm ProcC}(j)\ {\bf od};
		{\bf for\ } i = 1,\cdots,N {\bf\ do}\ {\rm ProcR}(i)\ {\bf od}
		\end{align*}
		with following subprograms:
		\begin{align*} 
		{\rm ProcC}(j)\equiv\ &{\bf for\ } i = 1,\cdots,N {\bf\ do} & {\rm ProcR}(i)\equiv\ &{\bf for\ } j = 1,\cdots,N-1 {\bf\ do} \\		
		&\quad q_{ij} = X^\alpha[q_{ij}]; & &\quad q_{ij} = X^{(Jc_{ij} = -1)}[q_{ij}]; \\		
		&\quad q_{ij} = \big(Z^\beta\big){}^{(h_{ij} = 1)}[q_{ij}] \quad {\bf od}; & &\quad q_{i(j+1)} = X^{(Jc_{ij} = -1)}[q_{i(j+1)}];\\		
		&{\bf for\ } i = 1,\cdots,N-1 {\bf\ do} & &\quad q_{ij},q_{i(j+1)} = CZ^\gamma[q_{ij},q_{i(j+1)}]; \\		
		&\quad q_{ij} = X^{(Jr_{ij} = -1)}[q_{ij}]; & &\quad q_{ij} = X^{(Jc_{ij} = -1)}[q_{ij}]; \\		
		&\quad q_{(i+1)j} = X^{(Jr_{ij} = -1)}[q_{(i+1)j}]; & &\quad q_{i(j+1)} = X^{(Jc_{ij} = -1)}[q_{i(j+1)}]; \\		
		&\quad q_{ij},q_{(i+1)j} = CZ^\gamma[q_{ij},q_{(i+1)j}]; & &{\bf od}; \\		
		&\quad q_{ij} = X^{(Jr_{ij} = -1)}[q_{ij}]; & &\\		
		&\quad q_{(i+1)j} = X^{(Jr_{ij} = -1)}[q_{(i+1)j}] \quad {\bf od}; & &
		\end{align*}
		where for simplicity, we write some logical judgments as superscripts; for example  $X^{(Jc_{ij} = -1)}$ means that if $Jc_{ij}=-1$, apply the $X$ gate, and otherwise skip. Since the parameters $h, Jc$ and $Jr$ are given a priori, this notation should not lead to any confusion. 
		
		\subsection{Specifying Incorrectness in Quantum Separation Logic}
		
		As pointed out at the beginning of this section, we can use our quantum separation logic to show that algorithm ${\rm VQA}(N)$ is indeed incorrect. Let us first describe its incorrectness in our logical language.  
		Suppose the Hamiltonian $H$ has eigenvalues $E_0,E1,...$ ranged in increasing order, with corresponding eigenspaces (projections) $Q_0, Q_1...$. 
		If for each $i\leq n$, we can find a  precondition $P_i\in\cP$ such that $\models\{P_i\} {\rm VQA}(N) \{1-\sum_{k=0}^{i}Q_i\}\ (i=0,1,...,n)$, then by showing that $|0\>$ (the initial state of quantum circuit) is close to $P_i$; that is, $\<0|P_i|0\> \ge \delta_i$, we can conclude that the approximate ground energy computed by ${\rm VQA}(N)$ is at least
		\begin{equation}E_0 + \sum_{i=1}^{n}(E_{i+1}-E_{i})\delta_i.\end{equation}
		\begin{proof}
			For input $\rho = |0\>\<0|$, the output is $\sigma = \sem{{\rm VQA}(N)}(\rho)$. We have following observations of the energy of output by realizing $E_i$ is an increasing sequence and $\sum_{i\ge 0}Q_i = \id$:
			\begin{align*}
			\tr(H\sigma) &= \tr\left(\sum_{i\ge 0}E_iQ_i\sigma\right) \ge \tr\left(\sum_{i= 0}^nE_iQ_i\sigma\right) + \tr\left(E_{n+1}\sum_{i\ge n+1}Q_i\sigma\right) \\
			&= \tr(E_0\sigma) + \sum_{i=0}^n\tr\left((E_{i+1}-E_i)\left(1-\sum_{k=0}^iQ_i\right)\sigma\right).
			\end{align*}
			On the other hand, suppose $\models\{P_i\} {\rm VQA}(N) \{1-\sum_{k=0}^{i}Q_i\}\ (i=0,1,...,n)$. According to the lifting principle (see \cite{ZYY19} Theorem 3.2), we know that, for any two projections $P$ and $Q$, if $\{P\}\prog\{Q\}$ is valid, then for any input state $\rho$, $\tr(P\rho)\le\tr(Q\sem{\prog})(\rho)$. Thus, we obtain:
			$$\tr(P_i\rho)\le\tr\left(\left(1-\sum_{k=0}^{i}Q_i\right)\sigma\right)$$
			which implies that the energy of output
			$$\tr(H\sigma) \ge \tr(E_0\sigma) + \sum_{i=0}^n\tr\left((E_{i+1}-E_i)P_i\rho\right) = E_0 + \sum_{i=0}^n(E_{i+1}-E_i)\<0|P_i|0\>\ge E_0 + \sum_{i=0}^n(E_{i+1}-E_i)\delta_i$$
			if $\<0|P_i|0\>\ge\delta_i$ for all $0\le i\le n$, as we desired.
		\end{proof} 
		Therefore, whenever the quantity in (\ref{VQA-meters}) is far away from the real ground energy $E_0$, then ${\rm VQA}(N)$ is incorrect.
		
		To illustrate our idea more explicitly, let us consider the simplest case of $2\times 2$ grid ($N = 2$) with parameters:
		$$h = {\small\left[\begin{array}{cc}-1 & -1 \\ 1 & 1\end{array}\right]},\quad Jc = {\small\left[\begin{array}{c}-1 \\ -1\end{array}\right]},\quad Jr = \left[\begin{array}{cc}-1 & 1\end{array}\right]$$
		and $J_{ij;(i+1)j} = Jr_{ij}$ and $J_{ij;i(j+1)} = Jc_{ij}$. 
		The eigenvalues of the Hamiltonian $H$ in this case are $E_0,\cdots,E_5 = -6,-4,-2,0,2,4$ with corresponding eigenspaces $Q_0,Q_1,\cdots,Q_5$, respectively.  If we can find preconditions $P_0$ and $P_1$ satisfying: 
		\begin{equation}\label{VQA-preconditions}\begin{cases}&\models\{P_i\} {\rm VQA}(N) \left\{1-\sum_{k=0}^{i}Q_i\right\}\ (i=0,1),\\ 
		&\<0|P_0|0\> = 1-\frac{1}{16}sin(\alpha\pi)^4 \ge \frac{15}{16}, \\
		&\<0|P_1|0\> = 1-\frac{1}{32}(7+\cos(2\alpha\pi))\sin^2(\alpha\pi)\ge \frac{13}{16},
		\end{cases}\end{equation}
		then it follows from (\ref{VQA-meters}) that the approximate ground energy of VQA is at least $-2.5$, which is much higher than the real ground energy $E_0 = -6$.
		
		\subsection{Verifying Incorrectness in Quantum Separation Logic}
		\label{sec app sub Verifying Incorrectness in QSL}
		Now we use the inference rules of quantum separation logic to derive preconditions $P_0$ and $P_1$ required in (\ref{VQA-preconditions}) and thus show that ${\rm VQA}(N)$ is indeed incorrect. 
		The derivation is given following the three steps outlined at the end of Subsection \ref{sec Local Reasoning Entangle}. 
		
		\vspace{0.15cm}
		
		\noindent\textbf {Pushing out}: This step is essentially \textit{local reasoning}. For each subprogram ${\rm ProcC}(j)$ or ${\rm ProcR}(i)$, We derive a certain precondition of it with (the projection onto the one-dimensional subspace spanned by) the \emph{maximally entanglement} as their postconditions, which plays the role of pushing out to connect the variables in the subprogram with some auxiliary variables. For example, consider: $${\rm ProcR}(1)\equiv q_{11} = X[q_{11}];\ q_{12} = X[q_{12}];\ q_{11},q_{12} = CZ^\gamma[q_{11},q_{12}];\ q_{11} = X[q_{11}];\ q_{12} = X[q_{12}]$$
		We introduce auxiliary variables $q_{11}^\prime,q_{12}^\prime$ with the same types as $q_{11},q_{12}$. For postcondition $\Phi^{R1} = |\Phi^{R1}\>\<\Phi^{R1}|$ with 
		$|\Phi^{R1}\> = \frac{1}{2}\sum_{i,j\in\{0,1\}}|i\>_{q_{11}}|j\>_{q_{12}}|i\>_{q_{11}^\prime}|j\>_{q_{12}^\prime},$ 
		we can use rules \textsc{Unit} and \textsc{Seq} to derive precondition $\Psi^{R1} = |\Psi^{R1}\>\<\Psi^{R1}|$ with 
		$|\Psi^{R1}\> = X_{q_{11}}X_{q_{12}}CZ^\gamma_{q_{11},q_{12}}X_{q_{11}}X_{q_{12}}|\Phi^{R1}\>$
		such that $\vdash\{\Psi^{R1}\}{\rm ProcR}(1)\{\Phi^{R1}\}.$ Similarly, we can derive preconditions $\Psi^{R2}, \Psi^{C1}, \Psi^{C2}$ such that  $$\vdash\{\Psi^{R2}\}{\rm ProcR}(2)\{\Phi^{R2}\},\qquad \vdash\{\Psi^{C1}\}{\rm ProcC}(1)\{\Phi^{C1}\},\qquad \vdash\{\Psi^{C2}\}{\rm ProcC}(2)\{\Phi^{C2}\}.$$ 	
		\noindent\textbf {Modification and pulling back: } 
		Now we apply \textsc{PEPR} (which is an instance of \textsc{UnCR}, see Section \ref{rule PERP}) to VQA. With \textsc{Frame}, \textsc{Seq} and \textsc{Weak} and Proposition \ref{prop axiom projection}, we obtain:
		\begin{align*}
		&\vdash\{\Psi^{R1}\otimes\Psi^{R2}\}\{\Psi^{R1}\wedge\Psi^{R2}\}{\rm ProcR}(1)\{\Phi^{R1}\wedge\Psi^{R2}\}{\rm ProcR}(2)\{\Phi^{R1}\wedge\Phi^{R2}\}\{\Phi^{R1}\otimes\Phi^{R2}\},
		\end{align*}
		Note that $\otimes$ is the tensor product in mathematics and $\Psi^{R1}\otimes\Psi^{R2}$ and $\Phi^{R1}\otimes\Phi^{R2}$ are still projections.
		Similarly, it holds that $\vdash\{\Psi^{C1}\otimes\Psi^{C2}\}{\rm ProcC}(1);{\rm ProcC}(2)\{\Phi^{C1}\otimes\Phi^{C2}\}$.
		Realizing the fact that the tensor product of maximally entangled state is still a maximally entangled state, i.e., $\Phi_1\in{\rm MES}(\qbar_1,\qbar_1^\prime)$ and $\Phi_2\in{\rm MES}(\qbar_2,\qbar_2^\prime)$ imply $\Phi_1\otimes\Phi_2\in{\rm MES}(\qbar_1\qbar_2,\qbar_1^\prime\qbar_2^\prime)$, we can use rule \textsc{PEPR} to derive:
		$$\vdash\{R\}{\rm ProcR}(1);{\rm ProcR}(2)\{Q\},\quad R\triangleq \proj\left[16\cdot\tr_{\qbar^\prime}\left((Q_{\qbar^\prime}\otimes\id_{\qbar})(\Psi^{R1}\otimes\Psi^{R2})(Q_{\qbar^\prime}\otimes\id_{\qbar})\right)\right]$$
		where $\qbar = q_{11},q_{12},q_{21},q_{22}$ and $\qbar^\prime = q_{11}^\prime,q_{12}^\prime,q_{21}^\prime,q_{22}^\prime$. Similarly, we have:
		$$\vdash\{P\}{\rm ProcC}(1);{\rm ProcC}(2)\{R\},\quad P\triangleq \proj\left[16\cdot\tr_{\qbar^\prime}\left((R_{\qbar^\prime}\otimes\id_{\qbar})(\Psi^{C1}\otimes\Psi^{C2})(R_{\qbar^\prime}\otimes\id_{\qbar})\right)\right].$$
		
		The explicit expressions of $P,Q,R$ are involved. Here, we only display the closed form of $(P_0)_{11}$ and $(P_1)_{11}$ for $Q = \id - Q_0$ and $\id - Q_0 - Q_1$ respectively:
		$$(P_0)_{11} = 1-\frac{1}{16}sin(\alpha\pi)^4,\quad (P_1)_{11} =1-\frac{1}{32}(7+\cos(2\alpha\pi))\sin^2(\alpha\pi),$$
		since $\<0|P_i|0\> = (P_i)_{11}$ is what actually needed in (\ref{VQA-preconditions}).

		\section{Scalable Reasoning: Verification of Security, Details for Section \ref{sec-security-exam}}
		
		\subsection{Security of Quantum Secret Sharing}
		
		We here prove the validity of $\models P_S[p,q,r]\rightarrow (\unia[p]\wedge \unia[q]\wedge \unia[r])$ (see Section \ref{sec QSS sec uni}). First, for any pure state $|\psi\>_{pqr}\in P_S[p,q,r]$, it can be written as:
		$$|\psi\>_{pqr} = \sum_{i=0}^2\lambda_i|e_i\>_{pqr},$$
		with complex numbers $\lambda_i$ satisfies $\sum_{i=0}^2|\lambda_i|^2 = 1$. Then a straightforward calculation shows that:
		$$\rt{\left(|\psi\>_{pqr}\<\psi|\right)}{p} = \frac{\id_p}{3},$$
		which implies $|\psi\>_{pqr}\<\psi|\models\unia[p]$. Next, note that $\unia[p]\in{\rm CM}$, so for any $\rho\models P_S[p,q,r]$, $\rho\models\unia[p]$, or equivalently, $\models P_S[p,q,r]\rightarrow \unia[p]$. Similarly for $\unia[q]$ and $\unia[r]$.

		\subsection{Security against Eavesdropper} As one can imagine, verification of quantum secret sharing with eavesdroppers is harder. Let us consider a slightly more complicated situation than its original design in \cite{HBB99,KKI99,CGL99}: the quantum secret is unknown for the sender and thus cannot be re-prepared by the sender. A protocol for secret transmission in this case was recently proposed in \cite{YLZ19}, and an instance of it can be written as the following program:  
		\begin{align*}
		{\rm QSS\_ E}(n) \equiv\ &{\bf for}\ i=1,\cdots, n\ {\bf do} & \prog_0[p,q,r,h_i]\equiv\ & p,h_i := \perm(p,h_i\mapsto h_i,p)[p,h_i];\\
		&\quad {\bf Enc}[p,q,r]; & & q,r := U_{\rm rec}[q,r];\\
		&\quad c:= |0\>; c:= H[c]; & & p,q := \perm(p,q\mapsto q,p)[p,q];\\
		&\quad {\bf if}\ M[c] = 0\ \rightarrow\ \prog_0[p,q,r,h_i] & & \\
		&\quad \qquad\quad\ \square\,  1\ \rightarrow\ \prog_1[p,q,r,h_i];\ {\bf fi} & \qquad\prog_1[p,q,r,h_i]\equiv\ & q,h_i := \perm(q,h_i\mapsto h_i,q)[q,h_i];  \\
		&{\bf od} & & p,r := \perm(p,r\mapsto r,p)[p,r];
		\end{align*}
		This is a $n+1$-round protocol: in each round, Alice encodes a qutrit to $p,q,r$, and Eva tosses the fresh coin $c$ by Hadamard gate $H$ and measures it by computational basis $M$ to decide which qutrit he is going to steal; Alice first tries to send $p$ to Bob and then $q$ to Charlie. If the coin is head (0), then Eva steals $p$ and stores it in her own register $h_i$, and Alice recovers the message from $q,r$ and sets it as the secret for next round; and if the coin is tail (1), then Eva steals $q$ and stores in $h_i$, and Alice sets $r$ as the secret for next round. It can be shown at the end if Eva doesn't steal $p$ and $q$ at the $n+1$ round, the qutrit(s) that Bob and Charlie get are indeed what they want (anyone has no information about secret but they together can recover the secret).
		
		The security of ${\rm QSS\_ E}(n)$ can be expressed as the uniformity:
		\begin{equation}\vdash\label{qss-e-n}\{\top\}{\rm QSS\_ E}(n)\{\unia[h_1,\cdots,h_n]\},
		\end{equation}
		which means that all qutrits Eva stolen are in fact useless.
		We show how (\ref{qss-e-n}) can be proved in our quantum separation logic. First, by the frame rule \textsc{Const} and (\ref{eqn QSS single}), we have for all $i=1,\cdots,n$:
		$$\vdash\{\unia[h_1,\cdots,h_{i-1}]\}\ p,q,r = {\bf Enc}[p,q,r]\ \{P_S[p,q,r]\wedge\unia[h_1,\cdots,h_{i-1}]\}.$$
		Next, we consider the first branch of the ${\bf if}$-statement in ${\rm QSS\_ E}(n)$ and obtain:
		\begin{align}
		\vdash 
		&\{P_S[p,q,r]\wedge\unia[h_1,\cdots,h_{i-1}]\} && \nonumber\\
		&\bullet\quad \perm(p,h_i\mapsto h_i,p)[p,h_i]; && \textsc{Perm} \nonumber\\
		&\{P_S[h_i,q,r]\wedge\unia[h_1,\cdots,h_{i-1}]\} &&  \textsc{Weak} \nonumber\\
		&\{\unia[h_1,\cdots,h_i]\} && \label{qss-no}\\
		&\bullet\quad q,r := U_{\rm rec}[q,r];p,q := \perm(p,q\mapsto q,p)[p,q]; &&  \textsc{Unit},\ \textsc{Frame} \nonumber\\
		&\{\unia[h_1,\cdots,h_i]\} &&  \nonumber
		\end{align}
		Note that $(P_S[h_i,q,r]\wedge\unia[h_1,\cdots,h_{i-1}])\rightarrow\unia[h_1,\cdots,h_i]$ in assertion logic is derived from Proposition \ref{prop axiom projection extend} (5). 
		Similarly we have: \begin{equation}\label{qss-yes}\vdash\{P_S[h_i,q,r]\wedge\unia[h_1,\cdots,h_{i-1}]\} \prog_1[p,q,r,h_i] \{\unia[h_1,\cdots,h_i]\}\end{equation}
		the second branch of the ${\bf if}$-statement in ${\rm QSS\_ E}(n)$.
		Now, using rule \textsc{RIf} we can combine (\ref{qss-no}) and (\ref{qss-yes}) to derived the following for the ${\bf if}$-statement in ${\rm QSS\_ E}(n)$:
		$$\vdash\{(P_S[p,q,r]\wedge\unia[h_1,\cdots,h_{i-1}])\ast\id_c\}{\bf if}\cdots{\bf fi}\{\unia[h_1,\cdots,h_i]\}.$$
		Finally, we use \textsc{Weak} and \textsc{Seq} repeatedly to glue the above  judgments together:
		\begin{align*}\vdash 
		&\{\top\} &&\\
		&\bullet\quad {\bf for}\ i=1,\cdots, n\ {\bf do} && \textsc{Seq}\\
		&\qquad\qquad \{\unia[h_1,\cdots,h_{i-1}]\} && \\
		&\qquad\qquad \bullet\quad p,q,r = {\bf Enc}[p,q,r]; && \\
		&\qquad\qquad \{P_S[p,q,r]\wedge\unia[h_1,\cdots,h_{i-1}]\} && \textsc{Weak}\\
		&\qquad\qquad \!\left\{\big[(P_S[p,q,r]\wedge\unia[h_1,\cdots,h_{i-1}])\wedge\id_c\big]\wedge(\id_{pqrh_1\cdots h_{i-1}}\ast \id_{\emptyset})\right\} &&\\ 
		&\qquad\qquad \bullet\quad c:= |0\>;\ c:= H[c];  && \textsc{Unit},\textsc{Frame} \\ 
		&\qquad\qquad \{(P_S[p,q,r]\wedge\unia[h_1,\cdots,h_{i-1}])\ast\id_c\} \\
		&\qquad\qquad \bullet\quad {\bf if}\ M[c] = 0\ \rightarrow\ \prog_0[p,q,r,h_i]\ \square \  1\ \rightarrow\ \prog_1[p,q,r,h_i]\ {\bf fi} && \\
		&\qquad\qquad \{\unia[h_1,\cdots,h_i]\} \\
		&\bullet\quad {\bf od} && \\
		&\{\unia[h_1,\cdots,h_n]\}. && 
		\end{align*}
		
		\section{BI with domain - Constructing 2-BID Logic}
		\label{sec 2BID}
		
		As discussed in Section \ref{sec res}, the failure of existence of extensions makes some BI formulas nonrestrictive -- satisfaction relation $\rho\models\phi$ depends on the variables outside $\free{\phi}$. On the other hands, restriction property is so important in program logic that we need to focus on those restrictive BI formulas when we establish QSL. 
		
		Pointed out in Section \ref{sec:discussion}, a possible way is to modify the BI frame to make restriction property intrinsic. Here and the following context focus on the aim, along the line that first introduce 2-BID logic, and then construct QSL based on 2-BID. 
		
		As the first step, we define the assertion language as an extension of BI-logic, tailored for specifying properties of quantum states, with a special consideration of accommodating entanglement and separation together.

		\subsection{BI-Logic with nondeterministic composition}
		
		Here, we give an alternative definition of BI frame with nondeterministic composition, i.e.,  $x\circ y$ is a set of worlds rather than a single world.
		
		\begin{definition}[BI frame \cite{OP99,Doc19}]
			A BI frame is a tuple $\cX = (X,\circ,\preceq,e)$, where $X$ is a set equipped with a preorder $\preceq$, and $\circ: X\times X\rightarrow \wp(X)$ is a binary operation mapping to the power set of $X$ with an unit element $e$ and satisfying the following conditions:
			\begin{enumerate}
				\item (Commutativity) $z\in x\circ y\Rightarrow z\in y\circ x$;
				\item (Unit Existence) $x\in x\circ e$;
				\item (Coherence)  $x\in y\circ z\Rightarrow x\succeq y$;
				\item (Associativity) $t^\prime\succeq t\in x\circ y\wedge w\in t^\prime\circ z\Rightarrow \exists s,s^\prime,w^\prime (s^\prime\succeq s\in y\circ z\wedge w\succeq w^\prime\in x\circ s^\prime)$.
			\end{enumerate}
			Moreover, a BI frame is said to be Downwards Closed (DC) if it satisfies 
			\begin{enumerate}
				\item[(5)] (Downwards Closed) $z\in x\circ y\wedge x^\prime\preceq x\wedge y^\prime\preceq y\Rightarrow \exists z^\prime (z^\prime\in x^\prime\circ y^\prime \wedge z^\prime\preceq z)$
			\end{enumerate}
		\end{definition}
		
		The downwards closed property was identified in \cite{CCA17} to simplify the semantics of magic wand $\sepimp$; we also find it useful for proving the restriction property of $\ast$.
		
		A \emph{valuation} is a mapping $\cV: \AP \rightarrow \wp(X)$, and it is monotonic if $x\in\cV(p)$ and $y \succeq x$ implies $y \in \cV(p)$. A BI frame $\cX$ together with a monotonic valuation $\cV$ gives a BI model $\cM$.
		
		\begin{definition}[Satisfaction in BI models \cite{OP99,Doc19}]
			\label{def satisfaction BI nondeterministic}
			Given a BI formula $\phi$ and a BI model $\cM = (X,\circ,\preceq,e,\cV)$. For each $x\in X$, satisfaction relation $x\models_\cM \phi$ is defined by induction on $\phi$:
			\begin{align*}
			&  x\models_\cM p \ \text{iff}\ x\in\cV(p)\qquad\qquad
			x\models_\cM \top :\  \text{always} \qquad\qquad
			x\models_\cM \bot :\  \text{never} \\
			&  x\models_\cM \phi_1\wedge\phi_2 \  \text{iff}\ x\models_\cM\phi_1\ \text{and}\ x\models_\cM\phi_2 \qquad\qquad
			x\models_\cM \phi_1\vee\phi_2   \  \text{iff}\ x\models_\cM\phi_1\ \text{or}\ x\models_\cM\phi_2\\
			&  x\models_\cM \phi_1\rightarrow\phi_2  \  \text{iff\ for\ all}\ x^\prime\succeq x,\ x^\prime\models_\cM\phi_1\ \text{implies}\ x^\prime\models_\cM\phi_2\\
			&  x\models_\cM \phi_1\ast\phi_2  \  \text{iff\ exists}\ x^\prime,y,z\ \text{s.t.}\ x\succeq x^\prime\in y\circ z,\ y\models_\cM \phi_1\ \text{and}\ z\models_\cM\phi_2 \\
			&  x\models_\cM \phi_1\sepimp\phi_2  \  \text{iff\ for\ all}\ x^\prime,y,z\ \text{s.t.}\ x^\prime\succeq x\ \text{and}\ z\in x^\prime\circ y, y\models_\cM \phi_1\ \text{implies}\ z\models_\cM \phi_2. 
			\end{align*}		
		\end{definition}
		
		The judgment $\phi\models_\cM\psi$ asserts that for every $x\in X$, whenever $x\models_\cM \phi$, it follows $x\models_\cM \psi$. We write $\phi\models\psi$ iff $\phi\models_\cM\psi$ holds for all models, and we say $\phi$ is valid, written $\models \phi$, iff $\top\models\phi$.
		
		Though the nondeterministic composition is considered and the definition of BI frame and BI model are somewhat different from standard ones (see Section \ref{sec brief review of BI}), it enjoys the same proof system in Hilbert-style is presented in Fig. \ref{fig HR for BI}. In particular, it not only sound but also complete. 
		We use $\phi\vdash\psi$ to denote provability. In particular, we say $\phi$ is provable if $\top\vdash\phi$ is provable.
		
		\begin{theorem}[Soundness and Completeness of BI, c.f. \cite{OP99,Doc19}]
			\label{thm sound complete BI}
			For BI formulas $\phi$ and $\psi$, $\phi\models\psi$ iff $\phi\vdash\psi$ in Hilbert system shown in Fig. \ref{fig HR for BI}.
		\end{theorem}
		
		\subsection{BID: BI with Domains}
		
		As discussed before, restriction property is not intrinsic in a standard BI logic. 
		To preserve the restriction property of implication in the quantum setting, we need to introduce domains for both states (i.e. elements of a BI-frame) and atomic propositions in order to explicitly specify the (quantum) variables under consideration.    
		\begin{definition}[BID frame]
			\label{def BID frame}
			A BID frame (a BI frame with domains) is a \emph{downwards closed} BI frame $\cX = (X,\circ,\preceq,e)$ together with a domain function $\LTypeEs: X\rightarrow \wp(Y)$, where $Y$ is a nonempty set of (quantum) variables, that satisfies: 
			\begin{enumerate}
				\item (Monotonicity) $x\preceq y$ implies $\LTypeE{x}\subseteq\LTypeE{y}$;
				\item (Restriction) For any $S\subseteq Y$ and $x\in X$, there is exactly one $\rt{x}{S}\in X$ such that $\rt{x}{S}\preceq x$ and $\LTypeE{\rt{x}{S}} = S\cap\LTypeE{x}$;
				\item (Extension) For any $S\subseteq Y$ and $x\in X$ such that $\LTypeE{x}\subseteq S$, there exists $y\in X$ such that $x\preceq y$ and $\LTypeE{y} = S$;
				\item (Union) $z\in x\circ y$ implies $\LTypeE{z} = \LTypeE{x}\cup \LTypeE{y}$.
			\end{enumerate}
		\end{definition}	
		
		\begin{definition}[BID model]
			\label{def BID model}
			A BID model $\cM$ is a tuple $\cM = (X,\circ,\preceq,e,\LTypeEs, \LTypeFs,\cV)$, where $(X,\circ,\preceq, e,\LTypeEs,\cV)$ is a BID model, and  $\LTypeFs:\AP\rightarrow \wp(Y)$ is a domain function for atomic propositions, such that for any $x,y\in X$ and $p\in\AP$,
			\begin{enumerate}
				\item (Monotonicity) $x\preceq y$ and $x\in\cV(p)$ implies $y\in\cV(p)$;
				\item (Restriction) $x\in\cV(p)$ implies $\LTypeF{p}\subseteq\LTypeE{x}$ and $\rt{x}{\LTypeF{p}}\in\cV(p)$.
			\end{enumerate}
		\end{definition}
		
		Intuitively, $\LTypeFs(p)$ defines the minimal domain of atomic proposition $p$, in the sense that the domain of any state $x\in \cV(p)$ must contain $\LTypeF{p}$. The domain function $\LTypeFs$ can be extended from atomic propositions to all BI formula as follows:
		\begin{enumerate}
			\item if $\phi\equiv \top$ or $\bot$, then $\LTypeF{\phi} = \emptyset$;
			\item if $\phi\equiv \phi_1\wedge\phi_2,\ \phi_1\vee\phi_2,\ \phi_1\rightarrow\phi_2,\ \phi_1\ast\phi_2$, then $\LTypeF{\phi} = \LTypeF{\phi_1}\cup \LTypeF{\phi_2}$;
			\item if $\phi\equiv \phi_1\sepimp\phi_2$, then $\LTypeF{\phi} = \LTypeF{\phi_2}\backslash\LTypeF{\phi_1}$.
		\end{enumerate}
		
		Now we can define satisfaction of BI-formulas in BID models. Here, we adopt a basic idea in classical separation logic \cite{Bro07,IO01,ORY01,OHe07}: satisfaction $x\models\phi$ is only defined when $\LTypeE{x}\supseteq\LTypeF{\phi}$. 
		
		\begin{definition}[Satisfaction in BID models]
			\label{def satisfaction BID}
			Given a BID model $\cM = (X,\circ,\preceq,e,\cV,\LTypeEs,\LTypeFs)$. Let $x\in X$ and $\phi$ be a BI-formula with $\LTypeE{x}\supseteq \LTypeF{\phi}$. Then satisfaction relation $x\models_\cM \phi$ is defined by induction on $\phi$:
			\begin{enumerate}
				\item $x\models_\cM \phi$ is defined in the same way as in Definition \ref{def satisfaction BI nondeterministic} if $\phi=p, \top, \bot, \phi_1\wedge\phi_2, \phi_1\vee\phi_2, \phi_1\rightarrow\phi_2, \phi_1\ast\phi_2$; 
				\item $x\models_\cM \phi_1\sepimp\phi_2$ iff for all $x^\prime,y,z$ s.t. $x^\prime\succeq \rt{x}{\LTypeF{\phi_1\sepimp\phi_2}}$ and $z\in x^\prime\circ y$, $y\models_\cM \phi_1$ implies $z\models_\cM \phi_2$.
			\end{enumerate}
		\end{definition}
		
		Accordingly, judgment $\phi\models_\cM\psi$ now asserts that for every $x\in X$ such that $\LTypeE{x}\supseteq \LTypeF{\phi}\cup\LTypeF{\psi}$, whenever $x\models_\cM \phi$, it holds that $x\models_\cM \psi$. Moreover, $\phi\models\psi$ means that $\phi\models_\cM\psi$ holds for all BID models, and  $\phi$ is valid iff $\top\models\psi$ holds.
		
		It is particularly important to see that monotonicity and the restriction property hold for satisfaction in BID models with the downwards closed   property.
		
		\begin{proposition}[Monotonicity and Restriction]
			\label{prop BI mon and res}
			Given a BID model $\cM$, 
			for all $x,y\in X$ and BI-formula $\phi$ such that $x\preceq y$ and $\LTypeF{\phi}\subseteq\LTypeE{x}$, 
			$x\models_\cM \phi$ if and only if $y\models_\cM \phi$.
		\end{proposition}
		
		A proof system (for reasoning about validity of BI-logical formulas in BID models) can be obtained by modifying the proof system of BI-logic with appropriate domain assumptions. More precisely, it consists of rules 1--10 and 14--16 in Fig. \ref{fig HR for BI} together with rules 11$^\prime$--13$^\prime$ and 17 in Fig. \ref{fig HR for BID}. Note that rule 17 is introduced so that the ordinary conjunction $\wedge$ can be pulled into the separation conjunction $\ast$ (under certain domain conditions).
		\begin{figure}
			\small
			\begin{align*}
			&11^\prime. \ \frac{\xi\vdash\phi\quad\mu\vdash\psi\quad \LTypeF{\phi}\subseteq\LTypeF{\xi}\quad\LTypeF{\psi}\subseteq\LTypeF{\mu}}{\xi\ast\mu\vdash\phi\ast\psi} \quad
			12^\prime. \  \frac{\mu\ast\phi\vdash\psi\quad \LTypeF{\mu}\subseteq\LTypeF{\psi}\backslash\LTypeF{\phi}}{\mu\vdash\phi\sepimp\psi} \\[0.1cm]
			&13^\prime. \  \frac{\xi\vdash\phi\sepimp\psi\quad\mu\vdash\phi \quad \LTypeF{\psi}\backslash\LTypeF{\phi}\subseteq\LTypeF{\xi}\quad \LTypeF{\phi}\subseteq\LTypeF{\mu} }{\xi\ast\mu\vdash\psi} \quad
			17. \  \frac{\LTypeF{\mu}\subseteq\LTypeF{\phi}}{(\phi\ast\psi)\wedge\mu\vdash(\phi\wedge\mu)\ast\psi}
			\end{align*}
			\caption{Hilbert-style rules for BID. 
			}
			\label{fig HR for BID}
		\end{figure}
		
		\begin{theorem}[Soundness of BID]
			\label{thm sound BID}
			For any BI-formulas $\phi$ and $\psi$, if $\phi\vdash\psi$ is provable in the BID proof system, then $\phi\models\psi$ for BID models.
		\end{theorem}
		
		It should be pointed out that the interpretation of separation implication $\sepimp$ in BID (see Definition \ref{def satisfaction BID}) is different from that in BI (see Definition \ref{def satisfaction BI nondeterministic}). Thus, rules 11--13 are in general not sound for BID, and the domain assumptions in 11$^\prime$--13$^\prime$ are necessary for soundness of these rules. Moreover, soundness of rule 17 is ensured by the restriction property.
		

		\subsection{2-BID}
		\label{sec BI extension}
		
		The BID models defined in the above subsection are still not strong enough for description of quantum states with entanglement resides between subsystems. 
		In order to distinguish separable quantum states and entangled quantum states, we use $\ast$ in BI-logic as an uncorrelated conjunction and introduce $\sd$ as a general (possibly entangled) conjunction (a detailed discussion why $\sd$ is employed can be found in Section \ref{sec app sub Quantum Interpretation of 2-BID Connectives}):
		
		\begin{definition}[Syntax of 2-BID]
			The 2-BID formulas are generated by the following syntax:
			$$
			\phi,\psi ::= p\in\AP\ |\ \top\ |\ \bot\ |\ \phi\wedge\psi\ |\ \phi\vee\psi\ |\ \phi\rightarrow\psi\ |\ \phi\ast\psi\ |\ \phi\sepimp\psi\ |\ \phi \sd \psi\ |\ \phi \sdimp \psi.
			$$
		\end{definition}
		Entanglement can now be expressed as a 2-BID formula of the form $(\phi\sd\psi)\wedge\neg\chi$, where $\chi$ describes the probabilistic combination of a family of formulas $\phi_i\ast\psi_i$ with $\phi_i$ and $\psi_i$ depicting certain properties of the subsystems.  
		To define the semantics of 2-BID formulas, we propose a 2-BID frame as a tuple $\cX = (X,\circ,\bullet,\preceq,e,\LTypeEs: X\rightarrow \wp(Y))$, where both $(X,\circ,\preceq,e,\LTypeEs)$ and $(X,\bullet,\preceq,e,\LTypeEs)$ are BID frames and they are related by the following condition: 
		\begin{itemize}
			\item[] (Weakening) $z\in x\circ y\rightarrow z\in x\bullet y$.
		\end{itemize}
		
		Various syntactic and semantic notions defined in the previous subsections can be straightforwardly generalised to 2-BID. First, a domain function $\LTypeFs: \AP \rightarrow \wp(\vars)$ for atomic propositions can be extended to all 2-BID formulas with the following additional clauses for $\sd$ and $\sdimp$:
		\begin{enumerate}
			\item if $\phi\equiv \phi_1\sd\phi_2$, then $\LTypeF{\phi} = \LTypeF{\phi_1}\cup \LTypeF{\phi_2}$;
			\item if $\phi\equiv \phi_1\sdimp\phi_2$, then $\LTypeF{\phi} = \LTypeF{\phi_2}\backslash\LTypeF{\phi_1}$.
		\end{enumerate}
		Next, the satisfaction relation $x\models_\cM\phi$ for states $x$ in a 2-BID model $\cM$ with $\LTypeE{x}\supseteq\LTypeF{\phi}$ can be added by introducing the following additional clauses:
		\begin{enumerate}
			\item $x\models \phi_1\sd\phi_2$ iff exists $x^\prime,x_1,x_2$ s.t. $x\succeq x^\prime\in x_1\bullet x_2$, $x_1\models \phi_1$ and $x_2\models\phi_2$;
			\item $x\models \phi_1\sdimp\phi_2$ iff for all $x^\prime,x_1,x_2$ s.t. $x^\prime\succeq \rt{x}{\LTypeF{\phi_1\sdimp\phi_2}}$ and $x_2\in x^\prime\bullet x_1$, $x_1\models \phi_1$ implies $x_2\models \phi_2$.
		\end{enumerate}
		We write $\models\phi$ when $\phi$ holds for all state $x$ such that $\LTypeE{x}\supseteq\LTypeF{\phi}$. As $(\cD,\circ,\preceq,1,\LTypeEs)$ and $(\cD,\bullet,\preceq,1,\LTypeEs)$ are both BID frames, all properties of BID models automatically hold for 2-BID models. In particular, Proposition \ref{prop BI mon and res} (monotonicity and the restriction property) is still true for 2-BID.
		
		A Hilbert-style proof system for 2-BID can be introduced as a combination of two subsystems, one for $\ast,\sepimp$ and one for $\sd,\sdimp$, related by a rule defining the entailment between $\ast$ and $\sd$. More precisely, it comprises all BID rules for (i.e. rules 1-10, 11$^\prime$--13$^\prime$ and 14--17 in Fig. \ref{fig HR for BI}  and \ref{fig HR for BID}) and their variants for $\sd,\sdimp$ as well as the following additional rule: 
		$$
		({\rm Conjunction\ Weakening}).\quad \frac{}{\phi\ast\psi\vdash\phi\sd\psi}.
		$$ 
		The soundness of this proof system for 2-BID is a direct corollary of Theorem \ref{thm sound BID} together with a trivial checking of the rule (Conjunction Weakening). For convenience, we present several useful derived rules in the following:
		\begin{proposition}
			\label{prop useful HR}
			\begin{enumerate}
				\item If $\models\phi\rightarrow\psi$ and $\models\psi\rightarrow\mu$, then $\models\phi\rightarrow\mu$.
				\item $\models (\phi\wedge\psi)\ast\mu\rightarrow(\phi\ast\mu)\wedge(\psi\ast\mu)$, \quad
				$\models (\phi\wedge\psi)\sd\mu\rightarrow(\phi\sd\mu)\wedge(\psi\sd\mu);$
				\item $\models \phi\ast\psi\rightarrow\phi\sd\psi, \quad \models \phi\sd\psi\rightarrow\phi\wedge\psi;$
				\item $\models \phi\ast\psi\rightarrow\phi, \quad \models \phi\sd\psi\rightarrow\phi$;
				\item If $\models\phi\leftrightarrow\psi$ and $\LTypeF{\phi} = \LTypeF{\psi}$, then for any $\mu$, $\models\mu\leftrightarrow\mu[\phi/\psi]$ where $\mu[\phi/\psi]$ is obtained by replacing all $\psi$ in $\mu$ by $\phi$.
			\end{enumerate}
		\end{proposition}
		
		\section{Quantum Interpretation of 2-BID Logic}
		\label{sec-q-interpret 2BID}
		
		As said before, 2-BID logic is designed as the assertion language of our quantum separation logic. More precisely, it is used to describe properties of the states of quantum programs. An abstract semantics of 2-BID was defined in the previous section in terms of 2-BID frames. In this section, this semantics will be concretised by defining a quantum frame.   
		
		\subsection{2-BID Frame of Quantum States}
		
		Basically, we consider the quantum states over specific registers as resources. Then two kinds of conjunction should be introduced to model combinations of spatially separate quantum resources (quantum states over disjoint registers):  a uncorrelated conjunction ``tensor product'' $\circ$ and a general conjunction ``coupling'' $\bullet$. Formally, they are defined as follows:
		
		\begin{definition}
			\label{def tensor coupling 2BID}
			The binary function $\circ$ and $\bullet: \cD\times\cD\rightarrow \wp(\cD)$ mapping each  pair of quantum states to a set of quantum states are defined by: 
			\begin{enumerate}
				\item $\rho_1\circ\rho_2\triangleq\left\{\rho_1\otimes\rho_2\right\}$ if $\type{\rho_1}\cap \type{\rho_2}=\emptyset$; otherwise, $\rho_1\circ\rho_2\triangleq\emptyset$;
				\item $\rho_1\bullet\rho_2\triangleq\big\{\rho\in\cD(\type{\rho_1}\cup \type{\rho_2})\ \big| \rt{\rho}{\type{\rho_1}} = \rho_1, \rt{\rho}{\type{\rho_2}} = \rho_2 \big\}$ if $\type{\rho_1}\cap \type{\rho_2}=\emptyset$; otherwise, $\rho_1\bullet\rho_2\triangleq\emptyset$;
			\end{enumerate}		
			where $\types$ is the domain function which specifies  the quantum register that a quantum state lies on. 
		\end{definition}	
		
		The functions $\circ$ and $\bullet$ are used to collect the tensor product and the couplings of two quantum states,  respectively, whenever they exist. If the domains of $\rho_1$ and $\rho_2$ have a nonempty overlap, then their tensor product and couplings are not well-defined and thus $\rho_1\circ\rho_2$ and $\rho_1\bullet\rho_2$ return the empty set. As a comparison, $\circ$ defined here and in Definition \ref{def tensor coupling} are the same in principle.  
		
		A partial order over quantum states considered as resources is the same as we defined in main text (Section \ref{sec-q-interpret}, Definition \ref{def partial order quantum state}).
		The partial order $\preceq$ is preserved under restriction:
		\begin{proposition}
			\label{pro-order-rt}
			\begin{enumerate}
				\item For any $S\subseteq \vars$ and $\rho\in\cD$,  $\rt{\rho}{S}\preceq\rho$. Indeed,  $\rt{\rho}{S}$ is the unique quantum state with domain  $S\cap\type{\rho}$ and $\rt{\rho}{S}\preceq\rho$.
				\item For any $S\subseteq \vars$ and $\rho,\rho^\prime\in\cD$, if $\rho\preceq\rho^\prime$, then $\rt{\rho}{S}\preceq\rt{\rho^\prime}{S}$.
			\end{enumerate}
		\end{proposition}
		
		Combining all of the ingredients defined above, we obtain: 
		
		\begin{proposition} $(\cD,\circ,\bullet, \preceq,1,\types)$ forms a 2-BID frame, where scalar number $1$ is understood as a state over the empty register, and $\types$ denotes for the domain of quantum states. \end{proposition}	
		
		\subsection{Atomic Propositions about Quantum States}
		\label{sec free choice of AP 2BID}
		
		Now we can interpret 2-BID logic in the quantum frame $(\cD,\circ,\bullet, \preceq,1,\types)$. As a common practice, we need to choose different sets of atomic propositions in different applications of our 2-BID logic. But the following assumptions about atomic propositions should be satisfied:
		\begin{enumerate}
			\item The domain function for atomic propositions $\Vs:\AP\rightarrow\wp(\vars)$ is defined so that for each atomic proposition $p\in\AP$, $\Vs(p)$ is a family of sets of quantum variables;  
			\item The interpretation $\sem{-}: \AP\rightarrow\wp(\cD)$ of atomic propositions is given so that for each atomic proposition $p\in\AP$, $\sem{p}$ is a set of quantum states that is upward-closed and closed under restriction: for any $\rho,\sigma\in\cD$ and $p\in\AP$,
			\begin{enumerate}
				\item $\rho\preceq \sigma$ and $\rho\in\sem{p}$ implies $\sigma\in\sem{p}$;
				\item $\rho\in\sem{p}$ implies $\V{p}\subseteq\type{\rho}$ and  $\rt{\rho}{\V{p}}\in\sem{p}$.
			\end{enumerate}
		\end{enumerate}
		
		\subsection{Quantum Interpretation of 2-BID Connectives}
		\label{sec app sub Quantum Interpretation of 2-BID Connectives}
		We saw in Section \ref{sec BI extension} that the main difference between BI logic and 2-BID logic comes from spatial (separating) conjunctions and implications. Now we can further examine the difference in terms of their quantum interpretations with the domain assumption.    
		
		{\vskip 3pt}
		
		\noindent\textbf{Spatial Conjunctions}: Only one spatial conjunction is needed in both classical and probabilistic separation logic. However, entanglement between quantum systems forces us to consider two different spatial conjunctions. 
		
		First, independence between registers in probabilistic separation logic \cite{BHL19} can be naturally generalised into the quantum setting: for two formulas $\phi_1$ and $\phi_2$ with disjoint domains, $\phi_1\ast\phi_2$ holds in quantum state $\rho$ if $\rho$ can be split into two uncorrelated states $\rho_1$ and $\rho_2$ that satisfy $\phi_1$ and $\phi_2$, respectively. Formally, $\rho\models\phi_1\ast\phi_2$ if and only if there exist two states $\rho_1$ and $\rho_2$ such that $\rho\succeq\rho_1\circ\rho_2$ and $\rho_i\models\phi_i$ for $i=1,2$ (see Definition \ref{def tensor coupling 2BID} for operation $\circ$). 
		
		To enable local reasoning in the presence of entanglement, we here introduce an additional spatial conjunction $\sd $. One might question why no $\sd$ is employed in main text (see Section \ref{sec-q-interpret}), basically there are two reasons:
		1. $\sd$ can be encoded by $\wedge$ and a side condition for free variables, i.e., $\phi\sd\psi\leftrightarrow\phi\wedge\psi$ if we assume $\free {\phi}\cap\free{\psi} = \emptyset$; 2. the side condition for free variables can be easily checked syntactically if no subscripting/aliasing is allowed in program logic; however, as we point out in Conclusion (Section \ref{sec conclusion}), we aim to verify programs with subscripting/aliasing, which would make checking side conditions for free variables difficult, at least syntacticlly difficult. Therefore, we introduce $\sd$ rather than use $\wedge$.
		For two formulas $\phi_1$ and $\phi_2$ with disjoint domains, a quantum state $\rho\in\cD(S)$ satisfies $\phi_1\sd\phi_2$ if its restrictions on two disjoint subsystems $S_1,S_2\subseteq S$, $\rt{\rho}{S_1}$ and $\rt{\rho}{S_2}$ satisfy $\phi_1$ and $\phi_2$, respectively. 
		Note that unlike in the case of independent conjunction $\ast$, here state $\rho$ can be entangled. This enables $\sd$ to be used in a situation where a program can be divided into several parts acting on different registers, but its input is often entangled between these subsystems. In fact, many of the existing quantum algorithms are designed in such a way. 
		
		\vspace{0.15cm}
		
		\noindent\textbf{Spatial Implications}: Usually, an implication is linked to its corresponding conjuction through a Galois connection. In BI-logic, the semantics of spatial $\sepimp$ corresponding to spatial conjunction $\ast$ is defined as follows:
		\begin{equation}
		\label{def s-implication}
		\rho\models \phi_1\sepimp\phi_2\ \text{iff\ for\ all}\ \rho^\prime,\rho_1,\rho_2\ \text{s.t.}\ \rho\preceq \rho^\prime\ \text{and}\ \rho_2\in \rho^\prime\circ \rho_1:\ \rho_1\models \phi_1\ \text{implies}\ \rho_2\models \phi_2 
		\end{equation}
		In 2-BID logic, however, we have to reconsider the above definition with the domain assumption. It is reasonable to set the domain of spatial implication $\V{\phi_1\sepimp\phi_2} = \V{\phi_2}\backslash \V{\phi_1}$. If we still adopt equation (\ref{def s-implication}) to define the semantics of $\sepimp$, then the   restriction property will be violated. We choose to modify defining equation (\ref{def s-implication}) as follows:
		\begin{equation}
		\label{def ss-implication}
		\rho\models \phi_1\sepimp\phi_2 \ \text{iff\ for\ all}\ \rho^\prime,\rho_1,\rho_2\ \text{s.t.}\ \rho^\prime\succeq \rt{\rho}{\V{\phi_1\sepimp\phi_2}}\ \text{and}\ \rho_2\in \rho^\prime\circ \rho_1, \rho_1\models \phi_1\ \text{implies}\ \rho_2\models \phi_2.
		\end{equation}
		Note that in equation (\ref{def ss-implication}) $\rho^\prime$ is required to range over all states $\succeq \rt{\rho}{\V{\phi_1\sepimp\phi_2}}$. Thus, the restriction property is automatically satisfied by $\sepimp$.
		
		The above discussion also applies to the spatial implication $\sdimp$ corresponding to  conjunction $\sd$.

		
		\subsection{Quantum Modification of 2-BID Formulas}
		
		We conclude this section by presenting a technique of modifying 2-BID formulas, similar to the modification of BI formulas (see Section \ref{sec modification BI formulas}) but much more general since implication and separating implications are considered.
		
		\begin{definition}[Modification of atomic propositions]
			\label{def sub atomic prop 2BID}
			Let $\prog$ be a unitary transformation $\qU$ or an initialisation $\qI$. For any atomic proposition $p\in\AP$, if there exists a 2-BID formula $\phi$ such that:
			\begin{enumerate}
				\item $p$ and $\phi$ have the same domain: $\V{p} = \V{\phi}$;
				\item for all $\rho\in\cD(\V{p}\cup\var(\prog))$, $\rho\models \phi$ if and only if $\sem{\prog}(\rho)\models p$;
			\end{enumerate}
			then we say that $\phi$ is an $\prog$-modification of $p$ and write  $p[\prog]\triangleq \phi$.
		\end{definition}
		
		The modification of some atomic propositions may not exists.
		We write $\phi[\prog]\Mexist$ whenever $\phi[\prog]$ is defined.
		The following examples give the modifications of those atomic propositions needed in the applications of quantum separation logic considered in this paper.

		The notion of modification can be easily extended to all 2-BID formulae: 
		\begin{definition}[Modification of 2-BID formulas]
			\label{def sub 2-BID form} 
			Let $\prog$ be unitary transformation $\qU$ or initialisation $\qI$. The modification $\phi[\prog]$ of 2-BID formula $\phi$ is defined by induction on the structure of $\phi$:
			\begin{enumerate}
				\item if $\phi\equiv \top$ or $\bot$, then $\phi[\prog] \equiv \phi$;
				\item if $\phi\equiv p\in\AP$, then $\phi[\prog]$ is defined according to Definition \ref{def sub atomic prop 2BID};
				\item if $\phi\equiv \phi_1\ \triangle\ \phi_2$ where $\triangle\in\{ \wedge,\vee,\rightarrow, \sd\}$ and $\phi_1[\prog]\Mexist$ and $\phi_2[\prog]\Mexist$, then $\phi[\prog]\equiv\phi_1[\prog] \ \triangle\ \phi_2[\prog]$;
				\item if $\phi\equiv \phi_1\ast\phi_2$, then
				\begin{enumerate}
					\item if $\prog\equiv\qU$ and $\phi_i[\prog]\Mexist$ and $\qbar\subseteq\free{\phi_i}$ or $\qbar\cap\free{\phi_i}=\emptyset$ for $i=1,2$, then $\phi[\prog]\equiv\phi_1[\prog] \ast \phi_2[\prog]$;
					\item if $\prog\equiv\qI$, then 
					$$\phi[\prog]\equiv\left\{\begin{array}{ll}(\phi_1[\prog] \sd \phi_2[\prog])\wedge(\id_{\V{\phi_1}\backslash q}\ast\id_{\V{\phi_2}\backslash q}) & \text{if\ }q\in\V{\phi_1}\cup\V{\phi_2},\ \phi_1[\prog]\Mexist\text{and}\ \phi_2[\prog]\Mexist\\
					\phi &\text{if\ }q\notin\V{\phi_1}\cup\V{\phi_2}
					\end{array}\right.$$\end{enumerate}
				\item if $\phi\equiv \phi_1\ \triangle\ \phi_2$ where $\triangle\in\{\sdimp, \sepimp\}$ and $\prog\equiv\qU$ or $\qI$, then 
				$$\phi[\prog] \equiv \left\{\begin{array}{ll} \phi_1\ \triangle\ \phi_2[\prog]  &\quad \text{if}\ \var(\prog)\subseteq\V{\phi_2}\backslash\V{\phi_1}\ \text{and}\ \phi_2[\prog]\Mexist \footnotemark\\
				\phi &\quad\text{if}\ \var(\prog)\cap\V{\phi_2}\backslash\V{\phi_1}=\emptyset 
				\end{array}\right. $$
			\end{enumerate}
			\footnotetext{ $\phi_1[\prog]\ \triangle\ \phi_2[\prog]$ also works for this case, but it is weaker since additional $\phi_1[\prog]\Mexist$ should be assumed. }
		\end{definition}
		
		The modification of 2-BID formula is not as convention. Since implication is considered, it is necessary to find the weakest precondition of $\phi$, to make the proof rules \textsc{Init} and \textsc{Unit} sound. For example, an initialization command $\qI$ makes $q$ uncorrelated with all other registers. As a consequence, the postcondition $\id_{q}\ast\id_{q^\prime}$ which asserts that two registers $q$ and $q^\prime$ are independent does not imply the precondition should assert the independence of $q$ and $q^\prime$, e.g., $\id_{q}\ast\id_{q^\prime}$, even if $\id_{q}[\qI] = \id_{q}$ and $\id_{q^\prime}[\qI] = \id_{q^\prime}$. In addition, the assumption of command variables and domains is declared for modification of $\sdimp$ and $\sepimp$, since we failed to derive the modified formula for the case $\var(\prog)\supset\V{\phi_2}\backslash\V{\phi_1}$.
		
		A close connection between the semantics of a 2-BID formula $\phi$ and its modification $\phi[\prog]$ is shown in the   following:  
		
		\begin{proposition}
			\label{pro modification 2BID}
			Let $\prog$ be unitary transformation $\qU$ or initialisation $\qI$, and $\phi$ be any 2-BID formula. If its modification  $\phi[\prog]$ is defined, then:
			\begin{enumerate}
				\item $\phi$ and $\phi[\prog]$ have the same domain: $\V{\phi} = \V{\phi[\prog]}$;
				\item for all $\rho\in\cD(\V{\phi}\cup\var(\prog))$, $\rho\models \phi[\prog]$ if and only if $\sem{\prog}(\rho)\models \phi$.
			\end{enumerate}
		\end{proposition}

		We can also generalize the concept of modification to quantum operation. Formally, we define the $\cE$-Modification as follows:
		\begin{definition}[$\cE$-Modification]
			\label{def qo modification 2BID}
			Let $\cE$ be quantum operation on $\qbar$. The $\cE$-Modification of a 2-BID formula $\phi$ is defined inductively:
			\begin{enumerate}
				\item (Atomic Propositions) For atomic proposition $p\in\AP$, if there exists 2-BID formula $\psi$ such that:
				\begin{enumerate}
					\item $p$ and $\psi$ have the same domain: $\V{p} = \V{\psi}$;
					\item for all $\rho\in\cD(\V{p}\cup\qbar)$, $\rho\models \psi$ if and only if $\cE(\rho)\models p$;
				\end{enumerate}
				then we say that $\psi$ is an $\cE$-modification of $p$ and write  $p[\cE[\qbar]]\triangleq\psi$.
				\item (Induction step) We write $\phi[\cE[\qbar]]\Mexist$ if $\phi[\cE[\qbar]]$ is defined.
				\begin{enumerate}			
					\item if $\phi\equiv \top$ or $\bot$, then $\phi[\cE[\qbar]] \equiv \phi$;
					\item if $\phi\equiv p\in\AP$, then $\phi[\cE[\qbar]]$ is defined according to Clause {\rm (1)};
					\item if $\phi\equiv \phi_1\ \triangle\ \phi_2$ where $\triangle\in\{ \wedge,\vee,\rightarrow, \sd\}$ and both $\phi_1[\cE[\qbar]]\Mexist$ and $\phi_2[\cE[\qbar]]\Mexist$, then $\phi[\cE[\qbar]]\equiv\phi_1[\cE[\qbar]] \ \triangle\ \phi_2[\cE[\qbar]].$
				\end{enumerate}
			\end{enumerate}
		\end{definition}
		Intuitively, if $\cE[\qbar](\phi)\Mexist$, then for any state $\rho$, $\cE(\rho)\models\phi$ if and only if $\rho\models \cE[\qbar](\phi)$.
		
		\section{Separation Logic for Quantum Programs with 2-BID as assertion logic}
		\label{sec QSL 2BID}
		
		Now we are ready to present our separation logic for quantum programs with 2-BID logic interpreted in the quantum frame defined in the last section as the assertion language. 
		
		Since all 2-BID formulas are restrictive, in contrast to Section \ref{sec QSL}, now a judgment is a Hoare triple of the form $\{\phi\}\prog\{\psi\}$ with both precondition $\phi$ and postcondition $\psi$ being 2-BID formulas.   
		
		\begin{definition}[Validity] Let $\vars$ be a set of quantum variables with $\V{\phi},\V{\psi},\var(\prog)\subseteq\vars$. Then a correctness formula $\{\phi\}\prog\{\psi\}$ is true in the sense of partial correctness with respect to $\vars$, written 
			$\vars \models\{\phi\}\prog\{\psi\}$, 
			if we have:
			$$\forall \rho\in\cD(\vars),\quad \rho\models\phi\Rightarrow\sem{\prog}_{\vars}(\rho)\models\psi.$$ Here, satisfaction relation $\rho\models\phi$ and $\sem{\prog}_{\vars}(\rho)\models\psi$ are defined according to the quantum interpretation of 2-BID logic given in Section \ref{sec-q-interpret 2BID}.
		\end{definition}
		
		Similarly, satisfaction does not depends on auxiliary variables.
		\begin{theorem}
			\label{thm eq glb var set 2BID}
			For any two sets $\vars$ and $\vars^\prime$ of variables,
			$\vars\models\{\phi\}\prog\{\psi\} \text{\ if\ and\ only\ if\ }\vars^\prime\models\{\phi\}\prog\{\psi\}.$
		\end{theorem}
		
		As a consequence, we can drop $\vars$ from $\vars\models\{\phi\}\prog\{\psi\}$ and simply write $\models\{\phi\}\prog\{\psi\}$.

		\subsection{Inference Rules}
		
		Most of the inference rules shown in main text (see Section \ref{sec QSL}, Figs. \ref{fig proof system 1}, \ref{fig proof system 2} and \ref{fig proof system 3}) are sound. We list the different rules here and comment them in a few words.
		
		\begin{figure}\centering
			\begin{equation*}\begin{split}
			&\textsc{Perm}\quad \frac{}{\{\phi[\qbar^\prime\mapsto\qbar]\}\qbar:=\perm(\qbar\mapsto\qbar^\prime)[\qbar]\{\phi\}} \qquad
			\textsc{RLoop$^\prime$}\quad\frac{\{\phi\ast M_1\}\prog\{\phi\ast\id_{\qbar}\}\quad\phi\in\text{CM}}{\{\phi\ast\id_{\qbar}\}\mathbf{while}\{\phi\sd M_0\}} \\
			&\textsc{Weak$^\prime$}\quad  \frac{\phi\rightarrow\phi^\prime\quad \{\phi^{\prime}\}\prog\{\psi^{\prime}\}\quad 
				\psi^{\prime}\rightarrow\psi  }{\{\phi\}\prog\{\psi\}} \\[0.1cm]
			&\textsc{FrameE}\quad \frac{\{\phi\}\prog\{\psi\}\quad\V{\mu}\cap\var(\prog)=\emptyset\quad\V{\psi}\subseteq\V{\phi}\cup\var(\prog)}{\{\phi\sd\mu\}\prog\{\psi\sd\mu\}} 
			\end{split}\end{equation*}
			\caption{Inference Rules for QSL of 2-BID. In \textsc{RLoop}, $\mathbf{while}$ is the abbreviation of $\mathbf{while}\ M[\qbar]=1\ \mathbf{do}\ \prog\ \mathbf{od}$, and  $M_0,M_1$ in assertions are regarded as projective predicates acting on $\qbar$. In \textsc{Perm}, $\perm(\qbar\mapsto\qbar^\prime)[\qbar]$ stands for the unitary transformation which permutes the variables from $\qbar$ to $\qbar^\prime$ (see Section \ref{sec basic Quantum app} for details).
			}
			\label{fig proof system 4}
		\end{figure}
		
		\begin{itemize}
			
			\item\textbf{Rule} \textsc{Perm}: At the first glance, one may think that this rule is a special case of rule \textsc{Unit} because permutation is a unitary transformation. Indeed, it is strictly stronger that what can be derived from \textsc{Unit} because entanglement is not invariant under a permutation between quantum registers; in particular when the 2-BID formulas describing the involved quantum systems contain independence conjunction $\ast$ and implication $\sepimp$.
			
			\item\textbf{Rules} \textsc{RLoop$^\prime$}: This one is slightly different than the one shown in Fig. \ref{fig proof system 1} since if fact, $\phi$ and $M_0$ have the disjoint domains, which leads to $\phi\wedge M_0$ equivalent to $\phi\sd M_0$.
			
			\item\textbf{Rules} \textsc{Weak$^\prime$}: note that the satisfaction relation for implication in 2-BID is different from it in BI, the $\phi\rightarrow\psi$ is exactly equivalent to $\phi\gimp\psi$. Thus, we can directly use the $\rightarrow$.	
			
			\item\textbf{Rules} \textsc{FrameE}: The conditions $\V{\mu}\cap\var(\prog)=\emptyset$ and $\V{\psi}\subseteq\V{\phi}\cup\var(\prog)$ in the premise ensur that (1) satisfaction of $\mu$ is unchanged after executing $\prog$; and (2) if $\phi\sd\mu$ has a non-empty interpretation, then $\psi\sd\mu$ is well-defined in the sense that the domains of $\psi$ and $\mu$ do not overlap: $\V{\psi}\cap\V{\mu}=\emptyset$.
		\end{itemize}
		
		Since all formulas considered here are 2-BID formulas, the set of CM and SP (see Definition \ref{def CM} and \ref{def SP}) can be generalized to larger sets:
		\begin{proposition}
			\label{prop CM 2BID}
			The formulas generated by following grammar are ${\rm CM}$. 
			$$
			\phi,\psi ::= p\in\cP\cup\cU\cup\cU_p\ |\ \top\ |\ \bot\ |\ \phi\wedge\psi\ |\ \phi\sd\psi\ |\ \mu \sepimp \psi\ |\ \phi\in{\rm SP}\ |\ \mu_1\ast\phi
			$$
			where $\mu$ is an arbitrary 2-BID formula, and $\mu_1\in {\rm SP}$.
		\end{proposition}
		\begin{proposition}
			\label{prop SP 2BID}
			The formulas generated by following grammar are ${\rm SP}$: 
			$$
			\phi,\psi ::= p\in\cU\ |\ p\in\cP \text{of rank 1}\ |\ \top\ |\ \bot\ |\ \phi\ast\psi\ |\ \mu_1 \sepimp \phi\ |\ \mu_1 \sdimp \phi
			$$
			where $\cP \text{of rank 1}$ consists all rank 1 projections, and $\mu_1$ is formula with non-empty interpretation.
		\end{proposition}

		To conclude this section, we show that quantum separation logic QSL consisting of all the proof rules listed in Figure \ref{fig proof system 1}, \ref{fig proof system 2} and \ref{fig proof system 3} and \ref{fig proof system 4} (\textsc{RLoop} and \textsc{Weak} are replaced by \textsc{RLoop$^\prime$} and \textsc{Weak$^\prime$}, respectively) are sound.
		
		\begin{theorem}[Soundness of QSL]
			\label{thm sound QSL 2BID}
			The proof system QSL of 2-BID is sound for terminating programs; that is, 
			$\text{if}\ \prog\ \text{is a terminating program, then}\ \vdash\{\phi\}\prog\{\psi\}\text{\ implies\ }\models\{\phi\}\prog\{\psi\}$.
		\end{theorem}

		\section{Deferred Proofs for Section \ref{sec 2BID} \ref{sec-q-interpret 2BID} and \ref{sec QSL 2BID}}
		
		\noindent\textbf{Most of the proofs in this part are tedious, and some of them are similar to previous proofs. Please find the proofs if needed.}
		
		\begin{proposition}
			\label{pro-unique}
			Suppose $\cX = (X,\circ,\preceq,e,\Ldoms: X\mapsto \wp(Y))$ is a BI frame with domain. Then the following statements hold:
			\begin{enumerate}
				\item for any $x,x^\prime\in X$ and $S\subseteq Y$, if $x^\prime\preceq x$ and $S\subseteq\Ldom{x^\prime}$, then $\rt{x^\prime}{S} = \rt{x}{S}$.
				\item for any $S\subseteq Y$ and $x\in X$ such that $S\subseteq\Ldom{x}$, $\rt{x}{S}$ is the unique least element of set $\{x^\prime\ |\ x^\prime\preceq x \text{\ and\ } S\subseteq\Ldom{x^\prime}\}$.
				\item for any $x,x^\prime\in X$ and $S\subseteq Y$, if $x^\prime\preceq x$, then $\rt{x^\prime}{S} \preceq \rt{x}{S}$.
				\item for any $x\in X$ and $S^\prime\subseteq S\subseteq Y$, $\rt{x}{S^\prime}\preceq \rt{x}{S}$.
			\end{enumerate}
		\end{proposition}
		\begin{proof}
			1. Note that $\rt{x^\prime}{S}\preceq x^\prime\preceq x$ and $\Ldom{\rt{x^\prime}{S}} = S$, so $\rt{x^\prime}{S} = \rt{x}{S}$ according to the uniqueness of domain restriction.
			
			2. For any $x^\prime$ such that $x^\prime\preceq x$ and $S\subseteq\Ldom{x^\prime}$, $\rt{x}{S} = \rt{x^\prime}{S} \preceq x^\prime$, so $\rt{x}{S}$ is a least element of the set. Moreover, suppose $y$ is another least element of the set, then $\rt{x}{S}\preceq y$ and $y\preceq \rt{x}{S}$, by domain monotonicity, $\Ldom{y} = \Ldom{\rt{x}{S}} = S$, so $y = \rt{x}{S}$ according to the uniqueness of domain restriction. Therefore, $\rt{x}{S}$ is the unique least element of the set.
			
			3. $x^\prime\preceq x$ implies $\Ldom{x^\prime}\subseteq\Ldom{x}$. Note that $\Ldom{\rt{x^\prime}{S}} = \Ldom{x^\prime}\cap S \subseteq \Ldom{x}\cap S = \Ldom{\rt{x}{S}}$, and $\rt{x^\prime}{S}\preceq x^\prime\preceq x$, $\rt{x}{S}\preceq x$, so $\rt{x^\prime}{S}\preceq \rt{x}{S}$ according to 2.
			
			4. Note that $\Ldom{\rt{x}{S^\prime}} = \Ldom{x}\cap S^\prime \subseteq \Ldom{x}\cap S = \Ldom{\rt{x}{S}}$, and $\rt{x}{S^\prime}\preceq x$, $\rt{x}{S}\preceq x$, so $\rt{x}{S^\prime}\preceq \rt{x}{S}$ according to 2.
		\end{proof}
		
		\vspace{0.5cm}
		
		\begin{claim}
			$\phi\models_\cM\psi$ if and only if $\models_\cM\phi\rightarrow\psi$.
		\end{claim}
		\begin{proof}
			At first, by the definition of domain for BI formula, $\Lfree{\phi\rightarrow\psi} = \Lfree{\phi}\cup\Lfree{\psi}$.
			
			(if part). For any $x$ such that $\Ldom{x}\supseteq\Lfree{\phi\rightarrow\psi}$ and $x^\prime\succeq x$, note that by domain monotonicity (see Definition \ref{def BID frame} (1)), $\Ldom{x^\prime}\supseteq\Ldom{x}\supseteq\Lfree{\phi\rightarrow\psi}$, thus by assumption, $x^\prime\models_\cM\phi$ implies $x^\prime\models_\cM\psi$, which leads to $\models_\cM \phi\rightarrow\psi$.
			
			(only if part). If $\models_\cM \phi\rightarrow\psi$, then for any $x$ such that $\Ldom{x}\supseteq\Lfree{\phi\rightarrow\psi}$, $x\models_\cM \phi\rightarrow\psi$. Note that $x\succeq x$ by reflexivity of preorder, so by the definition of satisfaction relation, $x\models_\cM\phi$ implies $x\models_\cM\psi$.
		\end{proof}
		
		\vspace{0.5cm}
		
		\noindent\textbf{Proof of Proposition \ref{prop BI mon and res}}
		
		\begin{proposition}[Monotonicity and Restriction, Proposition \ref{prop BI mon and res}]
			Given a BID model $\cM$, 
			for all $x,y\in X$ and BI-formula $\phi$ such that $x\preceq y$ and $\Lfree{\phi}\subseteq\Ldom{x}$, 
			$x\models_\cM \phi$ if and only if $y\models_\cM \phi$.
		\end{proposition}
		
		\begin{proof}
			The monotonicity holds as usual.
			\begin{lemma}[Monotonicity]
				\label{lem BI monotonicity}
				Monotonicity extends to all formulas with respect to BID semantics. That is, for all BI-formula $\phi$ and $x,y\in X$, $x\preceq y$ and $x\models_\cM\phi$ implies $y\models_\cM\phi$.
			\end{lemma}
			
			Moreover, with the downwards closed property, we can prove the restriction lemma for a BID model.
			
			\begin{lemma}[Restriction] 
				\label{lem BI retriction}
				Given a BID model $\cM$, for all $x\in X$ and BI-formula $\phi$,
				$x\models_\cM \phi$ implies for any $x^\prime\preceq x$ such that $\Ldom{x^\prime}\supseteq\Lfree{\phi}$, $x^\prime\models_\cM\phi$.
			\end{lemma}
			
			\noindent Proof of Lemma \ref{lem BI monotonicity}
			
			\vspace{0.2cm}
			
			It is a corollary of the case for original BI \cite{OP99,Pym02}. 
			We prove it here by induction on the structure of the formulas.
			\begin{itemize}
				\item[--] $\phi \equiv p\in\AP.$ $x\models_\cM p$ implies $\Ldom{x}\supseteq\Lfree{p}$ and $x\in\cV(p)$, so $\Ldom{y}\supseteq\Lfree{p}$ and $y\in\cV(p)$ due to the monotonicity of domain and $\cV$, or equivalently, $y\models_\cM p$.
				\item[--] $\phi\equiv \top (\bot)$. Trivial.
				\item[--] $\phi \equiv \phi_1\wedge\!(\vee)\ \phi_2.$ $x\models_\cM \phi_1\wedge\!(\vee)\ \phi_2$ implies $\Ldom{x}\supseteq\Lfree{\phi_1\wedge\!(\vee)\ \phi_2}$ and $x\models_\cM \phi_1$ and(or) $x\models_\cM \phi_2$, by induction hypothesis and monotonicity of domain, $\Ldom{y}\supseteq\Lfree{\phi_1\wedge\!(\vee)\ \phi_2}$ and $y\models_\cM \phi_1$ and(or) $y\models_\cM \phi_2$, so $y\models_\cM \phi_1\wedge\!(\vee)\ \phi_2$.
				\item[--] $\phi \equiv \phi_1\rightarrow\phi_2.$ $x\models_\cM \phi_1\rightarrow\phi_2$ implies that, $\Ldom{x}\supseteq\Lfree{\phi_1\rightarrow\phi_2}$ and for all $x\preceq x^\prime$, $x^\prime\models_\cM\phi_1$ implies $x^\prime\models_\cM\phi_2$. By monotonicity of domain, $\Ldom{y}\supseteq\Lfree{\phi_1\rightarrow\phi_2}$. Moreover, for any $y\preceq y^\prime$, $y^\prime$ must satisfy $x\preceq y^\prime$, therefore, $y^\prime\models_\cM\phi_1$ implies $y^\prime\models_\cM\phi_2$, which concludes $y\models_\cM \phi_1\rightarrow\phi_2$.
				\item[--] $\phi \equiv \phi_1\ast\phi_2.$ $x\models_\cM \phi_1\ast \phi_2$ implies that, $\Ldom{x}\supseteq\Lfree{\phi_1\ast\phi_2}$ and there exists $x^\prime, x_1, x_2$ such that $x\succeq x^\prime\in x_1\circ x_2$, $x_1\models_\cM\phi_1$ and $x_2\models_\cM\phi_2$. 
				Note that $\Ldom{y}\supseteq\Ldom{x}\supseteq\Lfree{\phi_1\ast\phi_2}$ and $y\succeq x\succeq x^\prime$, so $y\models_\cM \phi_1\ast\phi_2$. \\
				\item[--] $\phi \equiv \phi_1\sepimp \phi_2.$ $x\models_\cM \phi_1\sepimp \phi_2$ implies that, $\Ldom{x}\supseteq\Lfree{\phi_1\sepimp \phi_2}$ and for all $x^\prime,x^{\prime\prime},z$ s.t. $x^\prime\succeq \rt{x}{\Lfree{\phi_1\sepimp\phi_2}}$ and $z\in x^\prime\circ x^{\prime\prime}$, $x^{\prime\prime}\models_\cM \phi_1\Rightarrow z\models_\cM \phi_2$. As $y\succeq x$, so $\Ldom{y}\supseteq\Lfree{\phi_1\sepimp \phi_2}$ and $\rt{y}{\Lfree{\phi_1\sepimp \phi_2}}\succeq \rt{x}{\Lfree{\phi_1\sepimp \phi_2}}$, and then trivially $y\models_\cM \phi_1\sepimp \phi_2$. 
			\end{itemize}

			\vspace{0.2cm}
			
			\noindent Proof of Lemma \ref{lem BI retriction}
			
			\vspace{0.2cm}
			
			We prove this by induction on the structure of $\phi$. Suppose $x\models_\cM \phi$, due to Proposition \ref{pro-unique} and monotonicity (Lemma \ref{lem BI monotonicity}), it is sufficient to show $\rt{x}{\Lfree{\phi}}\models_\cM \phi$ (note that $x\models_\cM \phi$ implies $\Ldom{x}\supseteq\Lfree{\phi}$, so $\rt{x}{\Lfree{\phi}}$ is well-defined).
			\begin{enumerate}
				\item[--] $\phi \equiv p\in\AP.$ By the restriction property of $\cV$.
				\item[--] $\phi\equiv \top (\bot)$. Trivial.
				\item[--] $\phi \equiv \phi_1\wedge\!(\vee)\ \phi_2.$ $x\models_\cM \phi_1\wedge\!(\vee)\ \phi_2$ implies $\Ldom{x}\supseteq\Lfree{\phi_1\wedge\!(\vee)\ \phi_2}$ and $x\models_\cM \phi_1$ and(or) $x\models_\cM \phi_2$. By induction hypothesis, $\rt{x}{\Lfree{\phi_1}}\models_\cM \phi_1$ and(or) $\rt{x}{\Lfree{\phi_2}}\models_\cM \phi_2$. By Proposition \ref{pro-unique}, $\rt{x}{\Lfree{\phi_1\wedge(\vee) \phi_2}} \succeq \rt{x}{\Lfree{\phi_1}}$ and $\rt{x}{\Lfree{\phi_1\wedge(\vee) \phi_2}} \succeq \rt{x}{\Lfree{\phi_2}}$, by monotonicity, $\rt{x}{\Lfree{\phi_1\wedge(\vee) \phi_2}}\models_\cM \phi_1$ and(or) $\rt{x}{\Lfree{\phi_1\wedge(\vee) \phi_2}}\models_\cM \phi_2$, or equivalently, $\rt{x}{\Lfree{\phi_1\wedge(\vee) \phi_2}}\models_\cM \phi_1\wedge\!(\vee)\ \phi_2$.
				\item[--] $\phi \equiv \phi_1\rightarrow \phi_2.$ $x\models_\cM \phi_1\rightarrow \phi_2$ implies $\Ldom{x}\supseteq\Lfree{\phi_1\rightarrow\phi_2} = \Lfree{\phi_1}\cup\Lfree{\phi_2}$ and $x\models_\cM \phi_1\Rightarrow x\models_\cM \phi_2$. So $\Ldom{\rt{x}{\Lfree{\phi_1\rightarrow\phi_2}} }\supseteq\Lfree{\phi_1\rightarrow\phi_2}$.
				For any $x^\prime\succeq \rt{x}{\Lfree{\phi_1\rightarrow\phi_2}}$, note that $\rt{x^\prime}{\Lfree{\phi_1}} = \rt{\left(\rt{x}{\Lfree{\phi_1\rightarrow\phi_2}}\right)}{\Lfree{\phi_1}} = \rt{x}{\Lfree{\phi_1}}$ and similarly $\rt{x^\prime}{\Lfree{\phi_2}} = \rt{x}{\Lfree{\phi_1}}$ according to Proposition \ref{pro-unique}. By inductive hypothesis and monotonicity, $x^\prime\models_\cM \phi_1 \Leftrightarrow \rt{x^\prime}{\Lfree{\phi_1}}\models_\cM \phi_1 \Leftrightarrow \rt{x}{\Lfree{\phi_1}}\models_\cM \phi_1 \Leftrightarrow x\models_\cM \phi_1$
				and similarly $x^\prime\models_\cM \phi_2 \Leftrightarrow x\models_\cM \phi_2$, thus $x^\prime\models_\cM \phi_1\Rightarrow x^\prime\models_\cM \phi_2$. In summary, $\rt{x}{\Lfree{\phi_1\rightarrow\phi_2}} \models_\cM \phi_1\rightarrow \phi_2.$
				\item[--] $\phi \equiv \phi_1\ast \phi_2.$ $x\models_\cM \phi_1\ast \phi_2$ implies that, $\Ldom{x}\supseteq\Lfree{\phi_1\ast \phi_2} = \Lfree{\phi_1}\cup\Lfree{\phi_2}$ and there exists $x^\prime, x_1, x_2$ such that $x\succeq x^\prime\in x_1\circ x_2$, $x_1\models_\cM\phi_1$ and $x_2\models_\cM\phi_2$. 
				By inductive hypothesis and monotonicity, $\rt{x_1}{\Lfree{\phi_1}}\models_\cM\phi_1$ and $\rt{x_2}{\Lfree{\phi_2}}\models_\cM\phi_2$. Note that $x^\prime\in x_1\circ x_2$ and $\rt{x_1}{\Lfree{\phi_1}}\preceq x_1$ and $\rt{x_2}{\Lfree{\phi_2}}\preceq x_2$, by downwards closed property of $\circ$, there exists $z\in \rt{x_1}{\Lfree{\phi_1}}\circ \rt{x_2}{\Lfree{\phi_2}}$ such that $z\preceq x^\prime \preceq x$, and obviously, $z\models_\cM \phi_1\ast \phi_2$. Moreover, by domain union of $\circ$, $\Ldom{z} = \Ldom{\rt{x_1}{\Lfree{\phi_1}}}\cup\Ldom{\rt{x_2}{\Lfree{\phi_2}}} = \Lfree{\phi_1}\cup\Lfree{\phi_2}$, and by the uniqueness of domain restriction, $z = \rt{x}{\Lfree{\phi_1}\cup\Lfree{\phi_2}} = \rt{x}{\Lfree{\phi_1\ast \phi_2}}$, therefore, $\rt{x}{\Lfree{\phi_1\ast \phi_2}}\models_\cM \phi_1\ast \phi_2$.
				\item[--] $\phi \equiv \phi_1\sepimp \phi_2.$ $x\models_\cM \phi_1\sepimp \phi_2$ implies $\Ldom{x}\supseteq\Lfree{\phi_1\sepimp \phi_2}$, and thus $\Ldom{\rt{x}{\Lfree{\phi_1\sepimp \phi_2}}}\supseteq\Lfree{\phi_1\sepimp \phi_2}$. By definition, $\rt{x}{\Lfree{\phi_1\sepimp \phi_2}}\models_\cM \phi_1\sepimp \phi_2$ if we realize $\rt{\left(\rt{x}{\Lfree{\phi_1\sepimp \phi_2}}\right)}{\Lfree{\phi_1\sepimp \phi_2}} = \rt{x}{\Lfree{\phi_1\sepimp \phi_2}}$.
			\end{enumerate}		
			
		\end{proof}
		
		\vspace{0.5cm}
		
		\begin{theorem}[Deduction Theorem for BID]
			\label{thm deduction for BI}
			For any BI formulas $\phi$ and $\psi$, $\phi\vdash\psi$ is provable iff $\phi\rightarrow\psi$ is provable.
		\end{theorem}
		
		\noindent\textbf{Proof of Theorem \ref{thm deduction for BI}}
		
		\vspace{0.2cm}
		
		Indeed, with Hilbert rule 1, 2, 6, 9, 10 shown in Figure \ref{fig HR for BI}, the deduction theorem holds. For example, we may take the following proofs:
		
		(only if part): $\phi\vdash\psi$, $\top\wedge\phi\vdash\psi$ (6), $\top\vdash\phi\rightarrow\psi$ (10).
		
		(if part): $\top\vdash\phi\rightarrow\psi$, $\phi\wedge\top\vdash\phi\rightarrow\psi$ (6), $\phi\vdash\top\rightarrow(\phi\rightarrow\psi)$ (10), $\phi\vdash\top$ (2), $\phi\vdash\phi\rightarrow\psi$ (9), $\phi\vdash\phi$ (1), $\phi\vdash\psi$ (10).

		\vspace{0.5cm}
		
		\noindent\textbf{Proof of Theorem \ref{thm sound BID}}
		
		\begin{theorem}[Soundness of BID, Theorem \ref{thm sound BID}]
			For any BI-formulas $\phi$ and $\psi$, if $\phi\vdash\psi$ is provable in the BID proof system, then $\phi\models\psi$ for BID models.
		\end{theorem}
		
		\begin{figure}[h]
			\small
			\begin{align*}
			&1. \quad \frac{}{\phi\vdash\phi}  \qquad
			2. \quad \frac{}{\phi\vdash\top} \qquad 
			3. \quad \frac{}{\bot\vdash\phi} \qquad
			4. \quad \frac{\mu\vdash\phi\quad\mu\vdash\psi}{\mu\vdash\phi\wedge\psi} \qquad
			5. \quad \frac{\phi\vdash\psi_1\wedge\psi_2}{\phi\vdash\psi_i} \\	
			&6. \quad \frac{\phi\vdash\psi}{\mu\wedge\phi\vdash\psi} \qquad
			7. \quad \frac{\mu\vdash\psi\quad\phi\vdash\psi}{\mu\vee\phi\vdash\psi} \qquad
			8. \quad \frac{\phi\vdash\psi_i}{\phi\vdash\psi_1\vee\psi_2} \qquad
			9. \quad \frac{\mu\vdash\phi\rightarrow\psi\quad\mu\vdash\phi}{\mu\vdash\psi} \qquad
			10. \quad \frac{\mu\wedge\phi\vdash\psi}{\mu\vdash\phi\rightarrow\psi} \\ 
			&11^\prime. \quad \frac{\xi\vdash\phi\quad\mu\vdash\psi\quad \Lfree{\phi}\subseteq\Lfree{\xi}\quad\Lfree{\psi}\subseteq\Lfree{\mu}}{\xi\ast\mu\vdash\phi\ast\psi} \qquad
			12^\prime. \quad \frac{\mu\ast\phi\vdash\psi\quad \Lfree{\mu}\subseteq\Lfree{\psi}\backslash\Lfree{\phi}}{\mu\vdash\phi\sepimp\psi} \\
			&13^\prime. \quad \frac{\xi\vdash\phi\sepimp\psi\quad\mu\vdash\phi \quad \Lfree{\psi}\backslash\Lfree{\phi}\subseteq\Lfree{\xi}\quad \Lfree{\phi}\subseteq\Lfree{\mu} }{\xi\ast\mu\vdash\psi} \qquad
			14. \quad \frac{}{\phi\ast\psi\vdash\psi\ast\phi} \\
			&15. \quad \frac{}{(\phi\ast\psi)\ast\xi\vdash\phi\ast(\psi\ast\xi)} \qquad
			16. \quad \frac{}{\phi\ast\top\dashv\vdash\phi} \qquad
			17. \quad  \frac{\Lfree{\mu}\subseteq\Lfree{\phi}}{(\phi\ast\psi)\wedge\mu\vdash(\phi\wedge\mu)\ast\psi}
			\end{align*}
			\caption{Hilbert-style rules for BID. $i=1$ or $2$ for rules 5 and 8. 
			}
		\end{figure}
		
		\begin{proof}
			
			Due to the domain assumption, we write $\phi\models_\cM\psi$ iff for all $x$ such that $\Ldom{x}\supseteq\Lfree{\phi}\cup\Lfree{\psi}$, $x\models_\cM\phi$ implies $x\models_\cM\psi$. Indeed, in \cite{IO01}, the authors explained $\phi\models_\cM\psi$ in the same way, see Section 3.2.
			
			It is enough to show the soundness of each rule shown in Figure \ref{fig HR for BID}. Suppose $\cM$ is an arbitrary 2-BID model and let $x$ be an arbitrary state in $\cM$. As shown in Definition \ref{def satisfaction BID}, we will use the fact:  $\phi\models_\cM\psi\triangleq\ \models_\cM \phi\rightarrow\psi$ iff for all $x\in X$ such that $\Ldom{x}\supseteq \Lfree{\phi\rightarrow\psi} = \Lfree{\phi}\cup\Lfree{\psi}$, $x\models_\cM\phi$ implies $x\models_\cM\psi$.
			\begin{itemize}
				\item[--] rules 1-3: obvious.
				\item[--] rule 4: by assumptions, 1. $\forall x: \Ldom{x}\supseteq\Lfree{\mu}\cup\Lfree{\phi}, \ x\models_\cM\mu\Rightarrow x\models_\cM\phi$ and 2. $\forall x: \Ldom{x}\supseteq\Lfree{\mu}\cup\Lfree{\psi}, \ x\models_\cM\mu\Rightarrow x\models_\cM\psi$. So for any $x$ such that $\Ldom{x}\supseteq\Lfree{\mu}\cup\Lfree{\phi}\cup\Lfree{\psi}$, if $x\models_\cM\mu$, then $x\models_\cM\phi$ and $x\models_\cM\psi$, or equivalently, $x\models_\cM\phi\wedge\psi$.
				\item[--] rule 5: by assumptions, $\forall x: \Ldom{x}\supseteq\Lfree{\phi}\cup\Lfree{\psi_1}\cup\Lfree{\psi_2}, \ x\models_\cM\phi\Rightarrow x\models_\cM\psi_1$ and $x\models_\cM\psi_2$. For any $x$ such that $\Ldom{x}\supseteq\Lfree{\phi}\cup\Lfree{\psi_1}$, by existence of domain extension, there exists $y\succeq x$ such that $\Ldom{y}\supseteq\Lfree{\phi}\cup\Lfree{\psi_1}\cup\Lfree{\psi_2}$. If $x\models_\cM\phi$, by Proposition \ref{prop BI mon and res}, $y\models_\cM\phi$, by assumption, $y\models_\cM\psi_1$, so $x\models_\cM\psi_1$. Similar for $\phi\vdash\psi_2$.
				\item[--] rule 6: by assumptions, $\forall x: \Ldom{x}\supseteq\Lfree{\phi}\cup\Lfree{\psi}, \ x\models_\cM\phi\Rightarrow x\models_\cM\psi$. For any $x$ such that $\Ldom{x}\supseteq\Lfree{\mu}\cup\Lfree{\phi}\cup\Lfree{\psi}$, if $x\models_\cM\mu\wedge\phi$, then $x\models_\cM\phi$, by assumption, $x\models_\cM\psi$.
				\item[--] rule 7: by assumptions, 1. $\forall x: \Ldom{x}\supseteq\Lfree{\mu}\cup\Lfree{\psi}, \ x\models_\cM\mu\Rightarrow x\models_\cM\psi$ and 2. $\forall x: \Ldom{x}\supseteq\Lfree{\phi}\cup\Lfree{\psi}, \ x\models_\cM\phi\Rightarrow x\models_\cM\psi$. So for any $x$ such that $\Ldom{x}\supseteq\Lfree{\mu}\cup\Lfree{\phi}\cup\Lfree{\psi}$, if $x\models_\cM\mu\vee\phi$, then $x\models_\cM\mu$ or $x\models_\cM\psi$, by assumption, either of them implies $x\models_\cM\psi$.
				\item[--] rule 8: suppose $\phi\models_\cM\psi_1$, then $\forall x: \Ldom{x}\supseteq\Lfree{\phi}\cup\Lfree{\psi_1}, \ x\models_\cM\phi\Rightarrow x\models_\cM\psi_1$. So for any $x$ such that $\Ldom{x}\supseteq\Lfree{\phi}\cup\Lfree{\psi_1}\cup\Lfree{\psi_2}$, if $x\models_\cM\phi$, then by assumption, $x\models_\cM\psi_1$, so $x\models_\cM\psi_1\vee\psi_2$.
				\item[--] rule 9: by assumptions, 1. $\forall x: \Ldom{x}\supseteq\Lfree{\mu}\cup\Lfree{\phi}\cup\Lfree{\psi}, \ x\models_\cM\mu\Rightarrow x\models_\cM\phi\rightarrow\psi$ and 2. $\forall x: \Ldom{x}\supseteq\Lfree{\mu}\cup\Lfree{\phi}, \ x\models_\cM\mu\Rightarrow x\models_\cM\phi$. For any $x$ such that $\Ldom{x}\supseteq \Lfree{\mu}\cup\Lfree{\psi}$, by existence of domain extension, there exists $y\succeq x$ such that $\Ldom{y}\supseteq\Lfree{\mu}\cup\Lfree{\phi}\cup\Lfree{\psi}$. If $x\models_\cM\mu$, by Proposition \ref{prop BI mon and res}, $y\models_\cM\mu$, by assumptions, $y\models_\cM\phi$, $y\models_\cM\phi\rightarrow\psi$, and note $y\preceq y$, so $y\models_\cM\psi$ by definition, thus, $x\models_\cM\psi$ as desired.
				\item[--] rule 10: by assumption,  $\forall x: \Ldom{x}\supseteq\Lfree{\mu}\cup\Lfree{\phi}\cup\Lfree{\psi}, \ x\models_\cM\mu\wedge\phi\Rightarrow x\models_\cM\psi$. For any $x$ such that $\Ldom{x}\supseteq\Lfree{\mu}\cup\Lfree{\phi}\cup\Lfree{\psi}$, if $x\models_\cM \mu$, then for all $x^\prime \succeq x$, if $x^\prime\models_\cM\phi$, then by monotonicity, $x^\prime\models_\cM \mu$, so $x^\prime\models_\cM\mu\wedge\phi$, by assumption, $x^\prime\models_\cM\psi$, therefore, $x\models_\cM \phi\rightarrow\psi$.
				\item[--] rule 11$^\prime$: by assumptions, 1. $\forall x: \Ldom{x}\supseteq\Lfree{\xi}\cup\Lfree{\phi} = \Lfree{\xi}, \ x\models_\cM\xi\Rightarrow x\models_\cM\phi$ and 2. $\forall x: \Ldom{x}\supseteq\Lfree{\mu}\cup\Lfree{\psi} = \Lfree{\mu}, \ x\models_\cM\mu\Rightarrow x\models_\cM\psi$. For any $x$ such that $\Ldom{x}\supseteq \Lfree{\xi}\cup\Lfree{\mu}\cup\Lfree{\phi}\cup\Lfree{\psi} = \Lfree{\xi}\cup\Lfree{\mu}$, if $x\models_\cM\xi\ast\mu$, then there exists $x^\prime,x_1,x_2$ such that $x\succeq x^\prime \in x_1\circ x_2$, $x_1\models_\cM \xi$ and $x_2\models_\cM \mu$. Note that $\Ldom{x_1}\supseteq\Lfree{\xi}$,  so $x_1\models_\cM \phi$, and similarly, $x_2\models_\cM \psi$, therefore, $x\models_\cM\phi\ast\psi$.
				\item[--] rule 12$^\prime$: by assumption, 1. $\Lfree{\mu}\subseteq\Lfree{\psi}\backslash\Lfree{\phi} = \Lfree{\phi\sepimp\psi}$, and so 2. $\Lfree{\mu}\cup\Lfree{\phi}\cup\Lfree{\psi} = \Lfree{\phi}\cup\Lfree{\psi}$, and 3. $\forall x: \Ldom{x}\supseteq\Lfree{\phi}\cup\Lfree{\psi}, x\models_\cM \mu\ast\phi\Rightarrow x\models_\cM\psi$. For any $x$ such that $\Ldom{x}\supseteq\Lfree{\mu}\cup\Lfree{\phi\sepimp\psi} = \Lfree{\phi\sepimp\psi}$, if $x\models_\cM\mu$, then for any $x^\prime\succeq \rt{x}{\Lfree{\phi\sepimp\psi}}$, note that $\rt{x}{\Lfree{\phi\sepimp\psi}}\succeq \rt{x}{\Lfree{\mu}}$ by Proposition \ref{pro-unique}, so according to Proposition \ref{prop BI mon and res}, we have $x\models_\cM\mu\Leftrightarrow \rt{x}{\Lfree{\mu}}\models_\cM\mu\Leftrightarrow \rt{x}{\Lfree{\phi\sepimp\psi}}\models_\cM\mu\Leftrightarrow x^\prime\models_\cM\mu$. For any $y\models_\cM\phi$ and $z\in x^\prime\circ y$, by definition, $z\models_\cM\mu\ast\phi$, and note that $\Ldom{z} = \Ldom{x^\prime}\cup\Ldom{y}\supseteq \Lfree{\phi\sepimp\psi}\cup\Lfree{\phi} = \Lfree{\psi}\cup\Lfree{\phi}$, so $z\models_\cM\psi$. Thus, we have $\mu\models_\cM \phi\sepimp\psi$.
				\item[--] rule 13$^\prime$: for any $x$ such that $\Ldom{x}\supseteq \Lfree{\xi\ast\mu}\cup\Lfree{\phi}$, suppose $x\models_\cM\xi\ast\mu$, then there exist $x^\prime,y,z$ such that $x\succeq x^\prime\in y\circ z$, $y\models_\cM\xi$, $z\models_\cM\mu$. On the one hand, by assumptions $\mu\models_\cM\phi$ and $\Lfree{\phi}\subseteq\Lfree{\mu}$, so $\Ldom{z}\supseteq\Lfree{\mu} = \Lfree{\mu}\cup\Lfree{\phi}$, and $z\models_\cM\phi$. On the other hand, by another two assumptions, we realize that $\Ldom{y}\supseteq \Lfree{\xi} = \Lfree{\xi}\cup(\Lfree{\psi}\backslash\Lfree{\phi}) = \Lfree{\xi}\cup\Lfree{\phi\sepimp\psi}$ and thus $y\models_\cM\phi\sepimp\psi$. Recall that $y\succeq \rt{y}{\Lfree{\phi\sepimp\psi}}$, so $z\models_\cM\mu$ and $x^\prime\in y\circ z$ imply $x^\prime\models_\cM\psi$. Finally by monotonicity, $x\models_\cM \psi$.
				\item[--] rule 14: for any $x$ such that $\Ldom{x}\supseteq\Lfree{\phi}\cup\Lfree{\psi}$, if $x\models_\cM\phi\ast\psi$, then there exists $x^\prime,x_1,x_2$ such that $x\succeq x^\prime \in x_1\circ x_2$, $x_1\models_\cM \phi$ and $x_2\models_\cM \psi$. The commutativity of $\circ$ ensures that $x^\prime \in x_2\circ x_1$, therefore, $x\models_\cM\psi\ast\phi$.
				\item[--] rule 15: For any $x$ such that $\Ldom{x}\supseteq\Lfree{\phi}\cup\Lfree{\psi}\cup\Lfree{\xi}$, $x\models_\cM(\phi\ast\psi)\ast\xi$, then there exists $x^\prime,y,z$ s.t. $x\succeq x^\prime\in y\circ z$, $y\models_\cM\phi\ast\psi$, $z\models_\cM\xi$, then there exists $y^\prime,w,t$ s.t. $y\succeq y^\prime\in w\circ t$, $w\models_\cM\phi$, $t\models_\cM\psi$. Note that $x^\prime\in y\circ z$ and $y\succeq y^\prime\in w\circ t$, by associativity, there exists $s,s^\prime,w^\prime$, such that $s^\prime\succeq s\in t\circ z$ and $x^\prime\succeq w^\prime\in w\circ s^\prime$. So, $s\models_\cM\psi\ast\xi$ and by monotonicity, $s^\prime\models_\cM\psi\ast\xi$, and thus, $x^\prime\models_\cM\phi\ast(\psi\ast\xi)$, again by monotonicity, we conclude $x\models_\cM\phi\ast(\psi\ast\xi)$ as desired.
				\item[--] rule 16: ($\vdash$) For any $x$ such that $\Ldom{x}\supseteq\Lfree{\phi}$, if $x\models_\cM\phi\ast\top$, so there exists $x^\prime,x_1,x_2$ such that $x\succeq x^\prime \in x_1\circ x_2$, $x_1\models_\cM \phi$ and $x_2\models_\cM \top$. Coherence property of $\circ$ ensures that $x\succeq x^\prime\succeq x_1$, and by Proposition \ref{prop BI mon and res}, $x\models_\cM\phi$.
				
				($\dashv$) For any $x$ such that $\Ldom{x}\supseteq\Lfree{\phi}$, if $x\models_\cM\phi$, the existence of unit of $\circ$ ensures that, $x\succeq x\in x\circ e$, and note that $e\models_\cM\top$, so $x\models_\cM \phi\ast\top$.
				\item[--] rule 17: 	For any $x$ such that $x\models_\cM(\phi\ast\psi)\wedge\mu$, there exists $x^\prime,x_1,x_2$ such that $x\succeq x^\prime\in x_1\circ x_2$ such that $x_1\models_\cM\phi$ and $x_2\models_\cM\psi$. Note that, $\Ldom{x_1}\supseteq\Lfree{\phi}\supseteq\Lfree{\mu}$, and $x\succeq x^\prime\succeq x_1$ by coherence of $\circ$, thus $x_1\models_\cM\mu$ by Proposition \ref{prop BI mon and res}, so $x_1\models_\cM\phi\wedge\mu$ which leads to $x\models_\cM(\phi\wedge\mu)\ast\psi$.
			\end{itemize}
			
		\end{proof}
		
		\vspace{0.5cm}
		
		\noindent\textbf{Proof of Proposition \ref{prop useful HR}}
		
		\begin{proposition}[Proposition \ref{prop useful HR}]
			\begin{enumerate}
				\item If $\models\phi\rightarrow\psi$ and $\models\psi\rightarrow\mu$, then $\models\phi\rightarrow\mu$.
				\item $\models (\phi\wedge\psi)\ast\mu\rightarrow(\phi\ast\mu)\wedge(\psi\ast\mu)$, \quad
				$\models (\phi\wedge\psi)\sd\mu\rightarrow(\phi\sd\mu)\wedge(\psi\sd\mu);$
				\item $\models \phi\ast\psi\rightarrow\phi, \quad \models \phi\sd\psi\rightarrow\phi$;
				\item $\models \phi\ast\psi\rightarrow\phi\sd\psi, \quad \models \phi\sd\psi\rightarrow\phi\wedge\psi;$
				\item If $\models\phi\leftrightarrow\psi$ and $\Lfree{\phi} = \Lfree{\psi}$, then for any $\mu$, $\models\mu\leftrightarrow\mu[\phi/\psi]$ where $\mu[\phi/\psi]$ is obtained by replacing all $\psi$ in $\mu$ by $\phi$.
			\end{enumerate}
		\end{proposition}
		\begin{proof}	
			\begin{enumerate}
				\item Trivially using rules 6, 9 and 10.
				\item By rule 11$^\prime$, $\models(\phi\wedge\psi)\ast\mu\rightarrow\phi\ast\mu$ and $\models(\phi\wedge\psi)\ast\mu\rightarrow\psi\ast\mu$, then by rule 4, $\models(\phi\wedge\psi)\ast\mu\rightarrow(\phi\ast\mu)\wedge(\psi\ast\mu)$. Same for $\sd$.
				\item By rule 11$^\prime$,  $\phi\sd\psi\models\phi\sd\top$ and by rule 16, $\phi\sd\psi\models\phi$.
				\item $\models \phi\ast\psi\rightarrow\phi\sd\psi$ followed by rule Conjunction Weakening.  By (3) and rule 4, it is straightforward that $\phi\sd\psi\models\phi\wedge\psi$.
				\item Since $\models\phi\leftrightarrow\psi$ and $\Lfree{\phi} = \Lfree{\psi}$, we can realize that $\sem{\phi} = \sem{\psi}$, and thus $\models\mu\leftrightarrow\mu[\phi/\psi]$ is straightforward. Mathematically, it can be proved by induction on the structure of $\mu$ and we omit it here.
			\end{enumerate}
		\end{proof}
		
		Since we fixed the quantum interpretation of 2-BID, we have the following proposition which is convenient for uses.
		\begin{proposition}[Equivalent forms]
			\label{prop alt form} 
			We introduce the following proposition for some satisfaction relations based on the definitions of $\bullet$ and $\circ$ as they are more convenient in use.
			\begin{enumerate}
				\item[${\rm 1.}$] ${\rm (a)}$ $\rho\models\phi\sd\psi$ iff 
				
				${\rm (b)}$ $\Ldom{\rho}\supseteq \Lfree{\phi}\cup\Lfree{\psi}$ and exists disjoint $S_1,S_2\subseteq\Ldom{\rho}$ such that $\rt{\rho}{S_1}\models\phi$ and $\rt{\rho}{S_2}\models\psi$ iff 
				
				${\rm (c)}$ $\Ldom{\rho}\supseteq \Lfree{\phi}\cup\Lfree{\psi}$ and $\Lfree{\phi}\cap\Lfree{\psi} = \emptyset$, $\rho\models\phi$, $\rho\models\psi$.
				
				\item[${\rm 2.}$] ${\rm (a)}$ $\rho\models\phi\ast\psi$ iff
				${\rm (b)}$ $\Ldom{\rho}\supseteq \Lfree{\phi}\cup\Lfree{\psi}$ and $\Lfree{\phi}\cap\Lfree{\psi} = \emptyset$, $\rho\models\phi$, $\rho\models\psi$, $\rho\succeq\rt{\rho}{\Lfree{\phi}}\otimes\rt{\rho}{\Lfree{\psi}}$.
				
				\item[${\rm 3.}$] ${\rm (a)}$ $\rho\models\phi\sdimp\psi$ iff
				${\rm (b)}$ $\forall \rho_1,\rho_2$ such that $\Ldom{\rho_1} = \Lfree{\phi}$ and $\rho_2\in\rt{\rho}{\Lfree{\phi\sdimp\psi}}\bullet\rho_1$, $\rho_1\models\phi$ implies $\rho_2\models\psi$.
				
				\item[${\rm 4.}$] ${\rm (a)}$ $\rho\models\phi\sepimp\psi$ iff
				${\rm (b)}$ $\forall \rho_1$ such that $\Ldom{\rho_1} = \Lfree{\phi}$, $\rho_1\models\phi$ implies $\rt{\rho}{\Lfree{\phi\sdimp\psi}}\otimes\rho_1\models\psi$.
				
				\item[${\rm 5.}$] ${\rm (a)}$ $\rho\models\phi\rightarrow\psi$ iff
				${\rm (b)}$ $\Ldom{\rho} \supseteq \Lfree{\phi\rightarrow\psi}$, and $\rho\models\phi$ implies $\rho\models\psi$.
			\end{enumerate}
		\end{proposition}
		\begin{remark}
			One might question why we set $\sd$ as a primitive connective in assertion logic since by 1. (c), the domain conditions $\Lfree{\phi}\cap\Lfree{\psi} = \emptyset$ can be checked syntactically and the rest conditions $\rho\models\phi$, $\rho\models\psi$ can be explained by $\rho\models\phi\wedge\psi$. An important consideration for $\sd$ is the possible extension of our logic when subscripting/aliasing is allowed in quantum programming languages, since subscripts and aliases are widely used for large quantum programs in most of the current quantum programming platforms. Similar to the motivation of pointer separation logic, the domain side condition $\Lfree{\phi}\cap\Lfree{\psi} = \emptyset$ can no longer be syntactically checked when subscripting/aliasing is allowed and then the separation conjunction $\sd$ is helpful and necessary: spatial separation can be encoded in $\sd$ and thus some involving domain checking may be handled easier elsewhere.
		\end{remark}
		\begin{proof}
			\begin{enumerate}
				\item[${\rm 1.}$] If (a), then $\Ldom{\rho}\supseteq \Lfree{\phi}\cup\Lfree{\psi}$ and exists 
				$\rho^\prime,\rho_1,\rho_2\ \text{s.t.}\ \rho\succeq \rho^\prime\in \rho_1\bullet \rho_2,\ \rho_1\models \phi_1\ \text{and}\ \rho_2\models\phi_2$. Set $S_1=\Ldom{\rho_1}$ and $S_2=\Ldom{\rho_2}$, so $S_1,S_2\subseteq\Ldom{\rho}$, $S_1\cap S_2 = \emptyset$, $\rt{\rho}{S_1} = \rho_1$, $\rt{\rho}{S_2} = \rho_1$, $\rt{\rho}{S_1}\models\phi$ and $\rt{\rho}{S_2}\models\psi$, which implies (b).
				
				If (b), then $\Ldom{\rho}\supseteq \Lfree{\phi}\cup\Lfree{\psi}$, $S_1\supseteq\Lfree{\phi}$, $S_2\supseteq\Lfree{\psi}$, $\Lfree{\phi}\cap\Lfree{\psi} = \emptyset$, and by Kripke monotonicity, $\rho\models\phi$, $\rho\models\psi$, which is just (c).
				
				If (c), then set $\rho_1 = \rt{\rho}{\Lfree{\phi}}$ and $\rho_2 = \rt{\rho}{\Lfree{\psi}}$, as their domain are disjoint, so $\rho^\prime \triangleq \rt{\rho}{\Lfree{\phi}\cup\Lfree{\psi}}\in\rho_1\bullet\rho_2$ and $\rho\succeq\rho^\prime$, by Proposition \ref{prop BI mon and res}, $\rho_1\models\phi$, $\rho_2\models\psi$, which is (a).
				
				Therefore, (a) iff (b) iff (c).
				
				\item[${\rm 2.}$] If (a), by definition, there exists $\rho^\prime,\rho_1,\rho_2$ such that $\rho\succeq\rho^\prime\in\rho_1\circ\rho_2$ and $\rho_1\models\phi$, $\rho_1\models\psi$. Note that if $\rho_1\circ\rho_2$ is not empty, then $\rho^\prime=\rho_1\otimes\rho_2$ is the only element. By monotonicity, $\rho\models\phi$, $\rho\models\psi$, and downwards closed property, $\rho\succeq\rho^\prime=\rho_1\otimes\rho_2\succeq\rt{\rho}{\Lfree{\phi}}\otimes\rt{\rho}{\Lfree{\psi}}$ if we realize that $\rho\succeq\rho_1\succeq\rt{\rho}{\Lfree{\phi}}$ and $\rho\succeq\rho_2\succeq\rt{\rho}{\Lfree{\psi}}$.
				
				If (b), by Proposition \ref{prop BI mon and res}, $\rt{\rho}{\Lfree{\phi}}\models\phi$ and $\rt{\rho}{\Lfree{\psi}}\models\psi$, and $\rho\succeq\rt{\rho}{\Lfree{\phi}}\otimes\rt{\rho}{\Lfree{\psi}} \in \rt{\rho}{\Lfree{\phi}}\circ\rt{\rho}{\Lfree{\psi}}$. These lead to (a).
				
				Thus, (a) iff (b).
				
				\item[${\rm 3.}$] If (a), then by definition, (b) trivially holds.
				
				If (b), for all $\rho^\prime,\rho_1^\prime, \rho_2^\prime$ such that $\rho^\prime\succeq\rt{\rho}{\Lfree{\phi\sdimp\psi}}$ and $\rho_2^\prime\in\rho^\prime\bullet\rho_1^\prime$, and suppose $\rho_1^\prime\models\phi$. Set $\rho_1 = \rt{\rho_1^\prime}{\Lfree{\phi}}$, and as $\rho_1^\prime\models\phi$, so $\Ldom{\rho_1^\prime}\supseteq\Lfree{\phi}$ and $\rho_1 = \Lfree{\phi}$, and $\rho_1\preceq\rho_1^\prime$, $\Ldom{\rho_1}\models\phi$ by Proposition \ref{prop BI mon and res}. Moreover, by downwards closed property, we know that there exists $\rho_2$ such that $\rho_2\in \rt{\rho}{\Lfree{\phi\sdimp\psi}}\bullet\rho_1$ and $\rho_2\preceq\rho_2^\prime$. By (b), we know that $\rho_2\models\psi$, which leads to $\rho_2^\prime\models\psi$. Therefore, $\rho\models\phi\sdimp\psi$.
				
				In summary, (a) iff (b).
				\item[${\rm 4.}$] Similar to arguments of ${\rm 2}$, and realize the set of $\sigma_1\circ\sigma_2$ is a empty set or singleton (only element $\sigma_1\otimes\sigma_2$).
				
				\item[${\rm 5.}$] Trivial by Proposition \ref{prop BI mon and res}. In fact, once monotonicity and restriction are assumed, the interpretation of $\rightarrow$ in intuitionistic logic are equivalent to its in classical logic.
			\end{enumerate}
		\end{proof}
		
		\vspace{0.5cm}
		
		\noindent\textbf{Remarks for Definition \ref{def sub 2-BID form}}
		
		\begin{remark}
			Indeed, we can use the conventional modification (3) for also $\sdimp \text{and}\ \sepimp$ when $\prog\equiv\qU$. However, (5) is strictly more powerful in the sense that, 1) when $\phi_1[\prog]$ is not defined but $\phi_2[\prog]\downarrow$, (5) gives a valid modification but (3) gives an undefined one and 2) when both $\phi_1[\prog]\downarrow$ and $\phi_2[\prog]\downarrow$, (3) is derivable from (5) by using Proposition \ref{prop equal substitution unitary}.
		\end{remark}

		\begin{remark}
			As implication is considered, to make the proof rule $({\rm Init})$ sound, it is necessary to find the weakest precondition of $\phi$. 
			That is why the modification for $\ast$ is somewhat different: if $\rt{\sem{\qI}(\rho)}{\free{\phi_1}\cup\free{\phi_2}}$ is a tensor product state between $\free{\phi_1}$ and $\free{\phi_2}$ and if $q\in\free{\phi_1}$, then the input state $\rho$ only need to be a tensor product state between $\free{\phi_1}\backslash q$ and $\free{\phi_2}$. In fact, the initialization of $q$ makes $q$ separable from all other variables, so the input $\rho$ is not necessary to be a tensor product state between $\free{\phi_1}$ and $\free{\phi_2}$.
		\end{remark}

		\vspace{0.5cm}
		
		\noindent\textbf{Proof of Proposition \ref{pro modification 2BID}}
		
		\begin{proposition}
			Let $\prog$ be unitary transformation $\qU$ or initialisation $\qI$, and $\phi$ be any 2-BID formula. If its modification $\phi[\prog]$ is defined according to Definition \ref{def sub atomic prop 2BID} and \ref{def sub 2-BID form}, then:
			\begin{enumerate}
				\item $\phi$ and $\phi[\prog]$ have the same domain: $\free{\phi} = \free{\phi[\prog]}$;
				\item for all $\rho\in\cD(\free{\phi}\cup\var(\prog))$, $\rho\models \phi[\prog]$ if and only if $\sem{\prog}(\rho)\models \phi$.
			\end{enumerate}
		\end{proposition}
		\begin{proof}
			
			\noindent (1). Induction on the structure of $\phi$.
			
			\vspace{0.2cm}
			
			\noindent (2). 	We will introduce following lemmas which can be realized easily, and set variable set $\vars = \cD(\free{\phi}\cup\var(\prog))$.	
			
			Now we start to prove (2) by following two statements:
			
			{\bf Statement 1:} For any $\rho\in\cD(\vars)$, $\rho\models\phi[\qI]$ if and only if $\sem{\qI}(\rho)\models\phi$.
			\begin{enumerate}
				\item $\phi\in\AP$. By Definition \ref{def sub atomic prop 2BID}.
				
				\item $\phi \equiv \top$ or $\bot$. Trivial.
				
				\item $\phi \equiv \phi_1\wedge\!(\vee)\ \phi_2$. For any $\rho\in\cD(\vars)$, first by induction hypothesis, $\rho\models\phi_i[\qI]\Leftrightarrow\sem{\qI}(\rho)\models\phi_i$ for $i = 1,2$. Thus, 
				\begin{align*}
				&\rho\models\phi[\qI] \equiv \phi_1[\qI]\wedge\!(\vee)\ \phi_2[\qI] \\
				\Longleftrightarrow\ &\rho\models\phi_1[\qI] \text{\ and(or)\ }\rho\models\phi_2[\qI] \\
				\Longleftrightarrow\ &\sem{\qI}(\rho)\models\phi_1 \text{\ and(or)\ }\sem{\qI}(\rho)\models\phi_2 \\
				\Longleftrightarrow\ &\sem{\qI}(\rho)\models\phi_1\wedge\!(\vee)\ \phi_2.
				\end{align*}
				
				\item $\phi \equiv \phi_1\rightarrow\phi_2$. For any $\rho\in\cD(\vars)$, first by induction hypothesis, $\rho\models\phi_i[\qI]\Leftrightarrow\sem{\qI}(\rho)\models\phi_i$ for $i = 1,2$. Thus, 
				\begin{align*}
				&\rho\models\phi[\qI] \equiv \phi_1[\qI]\rightarrow\phi_2[\qI] \\
				\Longleftrightarrow\ &\rho\models\phi_1[\qI] \text{\ implies\ }\rho\models\phi_2[\qI] \\
				\Longleftrightarrow\ &\sem{\qI}(\rho)\models\phi_1 \text{\ implies\ }\sem{\qI}(\rho)\models\phi_2 \\
				\Longleftrightarrow\ &\sem{\qI}(\rho)\models\phi_1\rightarrow\phi_2.
				\end{align*}
				
				\item $\phi \equiv \phi_1\sd\phi_2$. For any $\rho\in\cD(\vars)$, first by induction hypothesis, $\rho\models\phi_i[\qI]\Leftrightarrow\sem{\qI}(\rho)\models\phi_i$ for $i = 1,2$. Thus, by Proposition \ref{prop alt form},
				\begin{align*}
				&\rho\models\phi[\qI] \equiv \phi_1[\qI]\sd\phi_2[\qI] \\
				\Longleftrightarrow\ &\free{\phi_1[\qI]}\cap\free{\phi_2[\qI]} = \emptyset,\ \rho\models\phi_1[\qI]\text{\ and\ } \rho\models\phi_2[\qI] \\
				\Longleftrightarrow\ &\free{\phi_1}\cap\free{\phi_2} = \emptyset,\ \sem{\qI}(\rho)\models\phi_1\text{\ and\ } \sem{\qI}(\rho)\models\phi_2 \\
				\Longleftrightarrow\ &\sem{\qI}(\rho)\models\phi_1\sd\phi_2. 
				\end{align*}
				
				\item $\phi \equiv \phi_1\ast\phi_2$. For any $\rho\in\cD(\vars)$, first by induction hypothesis, $\rho\models\phi_i[\qI]\Leftrightarrow\sem{\qI}(\rho)\models\phi_i$ for $i = 1,2$. 
				\begin{itemize}
					\item[$\cdot$]Case 1: $q\in\free{\phi_1}\cup\free{\phi_2}$ and  $\phi[\qI]\equiv (\phi_1[\qI] \sd \phi_2[\qI])\wedge(\id_{\free{\phi_1}\backslash q}\ast\id_{\free{\phi_2}\backslash q})$. Following by Proposition \ref{prop alt form} and Lemma \ref{lem sound proof 6}, we have :				
					\begin{align*}
					&\rho\models(\phi_1[\qI] \sd \phi_2[\qI])\wedge(\id_{\free{\phi_1}\backslash q}\ast\id_{\free{\phi_2}\backslash q}) \\
					\Longleftrightarrow\ &\free{\phi_1[\qI]}\cap\free{\phi_2[\qI]} = \emptyset,\ \rho\models\phi_1[\qI],\ \rho\models\phi_2[\qI] \text{\ and\ }\\
					& \rho\succeq\rt{\rho}{\free{\phi_1}\backslash q}\otimes\rt{\rho}{\free{\phi_2}\backslash q} \\
					\Longleftrightarrow\ &\free{\phi_1}\cap\free{\phi_2} = \emptyset,\ \sem{\qI}(\rho)\models\phi_1, \sem{\qI}(\rho)\models\phi_2 \text{\ and\ }\\
					&\sem{\qI}(\rho)\succeq\rt{\sem{\qI}(\rho)}{\free{\phi_1}}\otimes\rt{\sem{\qI}(\rho)}{\free{\phi_2}} \\
					\Longleftrightarrow\ &\sem{\qI}(\rho)\models\phi_1\ast\phi_2. 
					\end{align*}
					\item[$\cdot$]Case 2: $q\notin\free{\phi_1}\cup\free{\phi_2}$, and $\phi[\qI]\equiv \phi_1[\qI] \ast \phi_2[\qI]$. So, $\rt{\rho}{\free{\phi_1}\cup\free{\phi_2}} = \rt{\sem{\qI}(\rho)}{\free{\phi_1}\cup\free{\phi_2}}$, then using induction hypothesis we have:
					\begin{align*}
					&\rho\models\phi_1[\qI] \ast \phi_2[\qI] \\
					\Longleftrightarrow\ &\free{\phi_1[\qI]}\cap\free{\phi_2[\qI]} = \emptyset,\ \rho\models\phi_1[\qI],\ \rho\models\phi_2[\qI] \text{\ and\ }\\
					& \rho\succeq\rt{\rho}{\free{\phi_1}}\otimes\rt{\rho}{\free{\phi_2}} \\
					\Longleftrightarrow\ &\free{\phi_1}\cap\free{\phi_2} = \emptyset,\ \sem{\qI}(\rho)\models\phi_1, \sem{\qI}(\rho)\models\phi_2 \text{\ and\ }\\
					&\sem{\qI}(\rho)\succeq\rt{\sem{\qI}(\rho)}{\free{\phi_1}}\otimes\rt{\sem{\qI}(\rho)}{\free{\phi_2}} \\
					\Longleftrightarrow\ &\sem{\qI}(\rho)\models\phi_1\ast\phi_2. 
					\end{align*}				
				\end{itemize}

				\item $\phi \equiv \phi_1\sdimp\phi_2$. For any $\rho\in\cD(\vars)$, we consider following two cases:
				\begin{itemize}
					\item[$\cdot$] Case 1: $q\notin\free{\phi_2}\backslash\free{\phi_1}$, $\phi[\qI]\equiv\phi$. Note that $\free{\phi} = \free{\phi_2}\backslash\free{\phi_1}$, so by Lemma \ref{lem sound proof 1}, $\rt{\rho}{\free{\phi}} = \rt{\sem{\qI}(\rho)}{\free{\phi}}$, and then $\rho\models\phi[\qI]\Leftrightarrow\rt{\rho}{\free{\phi}}\models\phi\Leftrightarrow\rt{\sem{\qI}(\rho)}{\free{\phi}}\models\phi\Leftrightarrow\sem{\qI}(\rho)\models\phi$.
					\item[$\cdot$] Case 2: $q\in\free{\phi_2}\backslash\free{\phi_1}$, $\phi[\qI]\equiv\phi_1\sdimp\phi_2[\qI]$. We proof the following two directions by Proposition \ref{prop alt form}.
					\begin{itemize} 
						\item[-] If $\rho\models\phi[\qI]$, then $\sem{\qI}(\rho)\models\phi$. 
						
						First we have $\forall\rho_1$ s.t. $\dom{\rho_1}=\free{\phi_1}$, $\forall\rho_2\in\rt{\rho}{\free{\phi}}\bullet\rho_1$, $\rho_1\models\phi_1\Rightarrow\rho_2\models\phi_2[\qI]$. 
						
						Then, $\forall\rho_1$ s.t. $\dom{\rho_1}=\free{\phi_1}$, $\forall\rho_3\in\rt{\sem{\qI}(\rho)}{\free{\phi}}\bullet\rho_1$, note that $\rt{\sem{\qI}(\rho)}{\free{\phi}} = \sem{\qI}(\rt{\rho}{\free{\phi}})$, by Lemma \ref{lem sound proof 5}, $\exists\rho_2\in\rt{\rho}{\free{\phi}}\bullet\rho_1$ s.t. $\sem{\qI}(\rho_2) = \rho_3$. If $\rho_1\models\phi_1$, then $\rho_2\models\phi_2[\qI]$, by induction hypothesis, $\sem{\qI}(\rho_2)\models\phi_2$, or equivalently, $\rho_3\models\phi_2$. So, $\sem{\qI}(\rho)\models\phi_1\sdimp\phi_2\equiv\phi$.
						
						\item[-] If $\sem{\qI}(\rho)\models\phi$, then $\rho\models\phi[\qI]$. 
						
						First we have $\forall\rho_1$ s.t. $\dom{\rho_1}=\free{\phi_1}$, $\forall\rho_3\in\rt{\sem{\qI}(\rho)}{\free{\phi}}\bullet\rho_1$, $\rho_1\models\phi_1\Rightarrow\rho_3\models\phi_2$. 
						
						Then, $\forall\rho_1$ s.t. $\dom{\rho_1}=\free{\phi_1}$, $\forall\rho_2\in\rt{\rho}{\free{\phi}}\bullet\rho_1$, by Lemma \ref{lem sound proof 5}, $\exists\rho_3\in\sem{\qI}(\rt{\rho}{\free{\phi}})\bullet\rho_1$ s.t. $\sem{\qI}(\rho_2) = \rho_3$. Note that $\rt{\sem{\qI}(\rho)}{\free{\phi}} = \sem{\qI}(\rt{\rho}{\free{\phi}})$, so  $\rho_3\in\rt{\sem{\qI}(\rho)}{\free{\phi}}\bullet\rho_1$. If $\rho_1\models\phi_1$, then $\rho_3\models\phi_2$, $\sem{\qI}(\rho_2)\models\phi_2$, by induction hypothesis, $\rho_2\models\phi_2[\qI]$. So, $\rho\models\phi_1\sdimp\phi_2[\qI]\equiv\phi[\qI]$.
					\end{itemize}
				\end{itemize}
				
				\item $\phi \equiv \phi_1\sepimp\phi_2$. For any $\rho\in\cD(\vars)$, we consider following two cases:
				\begin{itemize}
					\item[$\cdot$] Case 1: $q\notin\free{\phi_2}\backslash\free{\phi_1}$, $\phi[\qI]\equiv\phi$. Note that $\free{\phi} = \free{\phi_2}\backslash\free{\phi_1}$, so by Lemma \ref{lem sound proof 1}, $\rt{\rho}{\free{\phi}} = \rt{\sem{\qI}(\rho)}{\free{\phi}}$, and then $\rho\models\phi[\qI]\Leftrightarrow\rt{\rho}{\free{\phi}}\models\phi\Leftrightarrow\rt{\sem{\qI}(\rho)}{\free{\phi}}\models\phi\Leftrightarrow\sem{\qI}(\rho)\models\phi$.
					\item[$\cdot$] Case 2: $q\in\free{\phi_2}\backslash\free{\phi_1}$, $\phi[\qI]\equiv\phi_1\sepimp\phi_2[\qI]$. By Proposition \ref{prop alt form} and Lemma \ref{lem sound proof 1}, \ref{lem sound proof 5}, we observe:
					\begin{align*}
					&\rho\models\phi[\qI]\equiv\phi_1\sepimp\phi_2[\qI] \\
					\Longleftrightarrow\ &\forall\rho_1\text{\ s.t.\ }\dom{\rho_1}=\free{\phi_1}, \rho_1\models\phi_1\Rightarrow\rt{\rho}{\free{\phi}}\otimes\rho_1\models\phi_2[\qI] \\
					\Longleftrightarrow\ &\forall\rho_1\text{\ s.t.\ }\dom{\rho_1}=\free{\phi_1}, \rho_1\models\phi_1\Rightarrow\sem{\qI}(\rt{\rho}{\free{\phi}}\otimes\rho_1)\models\phi_2\\
					\Longleftrightarrow\ &\forall\rho_1\text{\ s.t.\ }\dom{\rho_1}=\free{\phi_1}, \rho_1\models\phi_1\Rightarrow\rt{\sem{\qI}(\rho)}{\free{\phi}}\otimes\rho_1\models\phi_2 \\
					\Longleftrightarrow\ &\sem{\qI}(\rho)\models\phi_1\sepimp\phi_2\equiv\phi
					\end{align*}
				\end{itemize}
				
			\end{enumerate}

			{\bf Statement 2:} For any $\rho\in\cD(\vars)$, $\rho\models\phi[\qU]$ if and only if $\sem{\qU}(\rho)\models\phi$.
			\begin{enumerate}
				\item $\phi\in\AP$. By Definition \ref{def sub atomic prop 2BID}.
				
				\item $\phi \equiv \top$ or $\bot$. Trivial.
				
				\item $\phi \equiv \phi_1\wedge\!(\vee)\ \phi_2$. Similar to Statement 1 (3).
				
				\item $\phi \equiv \phi_1\rightarrow\phi_2$.  Similar to Statement 1 (4).
				
				\item $\phi \equiv \phi_1\sd\phi_2$. Similar to Statement 1 (5).
				
				\item $\phi \equiv \phi_1\ast\phi_2$. By assumption $\qbar\in P_v(\phi)$, either $\qbar\subseteq\free{\phi_1}$ or $\qbar\subseteq\free{\phi_2}$ or $\qbar\cap(\free{\phi_1}\cup\free{\phi_2}) = \emptyset$. So according to Lemma \ref{lem sound proof 7} and induction hypothesis we have:
				\begin{align*}
				&\rho\models\phi_1[\qU] \ast \phi_2[\qU] \\
				\Longleftrightarrow\ &\free{\phi_1[\qU]}\cap\free{\phi_2[\qU]} = \emptyset,\ \rho\models\phi_1[\qU],\ \rho\models\phi_2[\qU] \text{\ and\ }\\
				& \rho\succeq\rt{\rho}{\free{\phi_1[\qU]}}\otimes\rt{\rho}{\free{\phi_2[\qU]}} \\
				\Longleftrightarrow\ &\free{\phi_1}\cap\free{\phi_2} = \emptyset,\ \sem{\qU}(\rho)\models\phi_1, \sem{\qU}(\rho)\models\phi_2 \text{\ and\ }\\
				&\sem{\qU}(\rho)\succeq\rt{\sem{\qU}(\rho)}{\free{\phi_1}}\otimes\rt{\sem{\qU}(\rho)}{\free{\phi_2}} \\
				\Longleftrightarrow\ &\sem{\qU}(\rho)\models\phi_1\ast\phi_2. 
				\end{align*}
				
				\item $\phi \equiv \phi_1\sdimp\phi_2$.  Similar to Statement 1 (7).
				
				\item $\phi \equiv \phi_1\sepimp\phi_2$.  Similar to Statement 1 (8).
			\end{enumerate}

		\end{proof}

		\vspace{0.5cm}
		
		\noindent\textbf{Proof of Proposition \ref{prop equal substitution unitary}}
		
		\begin{proposition}
			\label{prop equal substitution unitary}
			For any $\phi_1,\phi_2$ and any command $\prog\equiv\qU$ and $\phi_1[\prog]\Mexist$ and $\phi_2[\prog]\Mexist$, then:
			\begin{enumerate}
				\item If $\qbar\cap\free{\phi_2} = \emptyset$, $\models (\phi_1\qmimp\phi_2) \leftrightarrow (\phi_1[\prog]\qmimp\phi_2)$; 
				\item If $\qbar\cap\free{\phi_2}\backslash\free{\phi_1} = \emptyset$, $\models (\phi_1\qmimp\phi_2) \leftrightarrow (\phi_1[\prog]\qmimp\phi_2[\prog])$; 
			\end{enumerate}
			where $\qmimp$ stands for $\sdimp$ or $\sepimp$.
		\end{proposition}
		
		\vspace{0.2cm}
		
		\begin{proof}
			According to Proposition \ref{pro modification 2BID}, we have the following statement: if $\phi[\qU]\downarrow$, then 
			$$\text{{\bf Statement:} For any $\rho\in\cD(\vars)$, $\rho\models\phi[\qU]$ if and only if $\sem{\qU}(\rho)\models\phi$.} $$
			By restriction lemma and the existence of domain extension, we directly have: 
			\begin{itemize}
				\item[--] if $\qbar\subseteq\free{\phi}$, then $\forall\rho\text{\ s.t.\ }\dom{\rho} = \free{\phi},\ \rho\models\phi[\qU]\Leftrightarrow (U^{\qbar}\otimes\id_{\free{\phi}\backslash\qbar})\rho(U^{\qbar\dag}\otimes\id_{\free{\phi}\backslash\qbar})\models\phi.$
				\item[--] if $\qbar\cap\free{\phi}=\emptyset$, then $\forall\rho\text{\ s.t.\ }\dom{\rho} = \free{\phi},\ \rho\models\phi[\qU]\Leftrightarrow \rho\models\phi.$
			\end{itemize}
			
			Now let us start to prove two statements for $\sdimp$, and it is similar for $\sepimp$. For statements $\mathit 1$, using Proposition \ref{prop alt form} we observe:
			\begin{align*}
			&\rho\models\phi_1\sdimp\phi_2 \\
			\Longleftrightarrow\ &\forall\rho_1\text{\ s.t.\ }\dom{\rho_1}=\free{\phi_1}, \forall\rho_2\in\rt{\rho}{\free{\phi_2}\backslash\free{\phi_1}}\bullet\rho_1, \rho_1\models\phi_1\Rightarrow\rho_2\models\phi_2 \\
			\Longleftrightarrow\ &\forall\rho_1\text{\ s.t.\ }\dom{\rho_1}=\free{\phi_1}, \forall\rho_2\in\rt{\rho}{\free{\phi_2}\backslash\free{\phi_1}}\bullet\rho_1, \rho_1\models\phi_1[\qU]\Rightarrow\rho_2\models\phi_2 \\
			\Longleftrightarrow\ &\rho\models\phi_1[\qU]\sdimp\phi_2. 
			\end{align*}
			
			For statement $\mathit 2$, there are two cases:
			\begin{itemize}
				\item[$\cdot$]Case 1: $\qbar\subseteq\free{\phi_1}$. We have:
				\begin{align*}
				&\rho\models\phi_1[\qU]\sdimp\phi_2[\qU] \\
				\Longleftrightarrow\ &\forall\rho_1\text{\ s.t.\ }\dom{\rho_1}=\free{\phi_1}, \forall\rho_2\in\rt{\rho}{\free{\phi_2}\backslash\free{\phi_1}}\bullet\rho_1, \rho_1\models\phi_1[\qU]\Rightarrow\rho_2\models\phi_2[\qU] \\
				\Longleftrightarrow\ &\forall\rho_1^\prime\text{\ s.t.\ }\dom{\rho_1}=\free{\phi_1}, \forall\rho_2^\prime\in\rt{\rho}{\free{\phi_2}\backslash\free{\phi_1}}\bullet\rho_1^\prime, \rho_1^\prime\models\phi_1\Rightarrow\rho_2^\prime\models\phi_2 \\
				\Longleftrightarrow\ &\rho\models\phi_1\sdimp\phi_2. 
				\end{align*}
				by realizing that there is one-to-one correspondence between $\rho_1$ and $\rho_1^\prime\triangleq(U^{\qbar}\otimes\id_{\dom{\rho_1}\backslash\qbar})\rho(U^{\qbar\dag}\otimes\id_{\dom{\rho_1}\backslash\qbar})$, and between $\rho_2$ and $\rho_2^\prime\triangleq(U^{\qbar}\otimes\id_{\dom{\rho_2}\backslash\qbar})\rho(U^{\qbar\dag}\otimes\id_{\dom{\rho_2}\backslash\qbar})$; moreover, $\rho_2\in\rt{\rho}{\free{\phi_2}\backslash\free{\phi_1}}\bullet\rho_1$ if and only if $\rho_2^\prime\in\rt{\rho}{\free{\phi_2}\backslash\free{\phi_1}}\bullet\rho_1^\prime$. These facts come from the reversibility of unitary transformations.
				
				\item[$\cdot$]Case 2: $\qbar\cap\free{\phi_1} = \emptyset$. So $\qbar\cap(\free{\phi_1}\cup\free{\phi_2}) = \emptyset$. Then obviously,
				\begin{align*}
				&\rho\models\phi_1\sdimp\phi_2 \\
				\Longleftrightarrow\ &\forall\rho_1\text{\ s.t.\ }\dom{\rho_1}=\free{\phi_1}, \forall\rho_2\in\rt{\rho}{\free{\phi_2}\backslash\free{\phi_1}}\bullet\rho_1, \rho_1\models\phi_1\Rightarrow\rho_2\models\phi_2 \\
				\Longleftrightarrow\ &\forall\rho_1\text{\ s.t.\ }\dom{\rho_1}=\free{\phi_1}, \forall\rho_2\in\rt{\rho}{\free{\phi_2}\backslash\free{\phi_1}}\bullet\rho_1, \rho_1\models\phi_1[\qU]\Rightarrow\rho_2\models\phi_2[\qU] \\
				\Longleftrightarrow\ &\rho\models\phi_1[\qU]\sdimp\phi_2[\qU]. 
				\end{align*}
			\end{itemize}
		\end{proof}

		\vspace{0.5cm}
		
		\noindent\textbf{Proof of Theorem \ref{thm eq glb var set 2BID}}
		
		\begin{theorem}[Theorem \ref{thm eq glb var set 2BID}]
			For any two sets $\vars$ and $\vars^\prime$ of variables,
			$\vars\models\{\phi\}\prog\{\psi\} \text{\ if\ and\ only\ if\ }\vars^\prime\models\{\phi\}\prog\{\psi\}.$
		\end{theorem}
		
		\begin{proof}
			Similar to the proof of Theorem \ref{thm eq glb var set}  by employing Proposition \ref{prop BI mon and res}.
		\end{proof}

		\vspace{0.5cm}
		
		\noindent\textbf{Proof of Proposition \ref{prop CM 2BID}}
		
		\begin{proposition}[Proposition \ref{prop CM 2BID}, Extended Version] The formulas generated by following grammar are ${\rm CM}$. 
			$$
			\phi,\psi ::= p\in\AP\ |\ \top\ |\ \bot\ |\ \phi\wedge\psi\ |\ \phi\sd\psi\ |\ \mu \sepimp \psi\ |\ \phi\in{\rm SP}\ |\ \mu_1\ast\phi
			$$
			where $\mu$ is an arbitrary 2-BID formula, and $\mu_1\in {\rm SP}$.
		\end{proposition}
		
		\begin{proof}
			\begin{enumerate}
				\item $p\equiv P\in\AP$. Similar to the proof of Proposition \ref{prop CM}.
				\item $\top$ or $\bot$. Trivial.
				\item $\phi\wedge\psi$. Suppose $\rho,\rho^\prime\in\cD$ with same domain and $\rho\models \phi\wedge\psi$ and $\rho^\prime\models \phi\wedge\psi$, then by induction hypothesis, for any $\lambda\in[0,1]$, 
				$$\lambda\rho + (1-\lambda)\rho^\prime \models\phi,\quad \lambda\rho + (1-\lambda)\rho^\prime \models\psi$$
				and thus, $\lambda\rho + (1-\lambda)\rho^\prime \models\phi\wedge\psi$.
				\item $\phi\sd\psi$. Suppose $\rho,\rho^\prime\in\cD$ with same domain and $\rho\models \phi\sd\psi$ and $\rho^\prime\models \phi\sd\psi$, then by induction hypothesis and Proposition \ref{prop useful HR} (3) and \ref{prop alt form} (1.c), $\free{\phi}\cap\free{\psi} = \emptyset$, for any $\lambda\in[0,1]$, 
				$$\lambda\rho + (1-\lambda)\rho^\prime \models\phi,\quad \lambda\rho + (1-\lambda)\rho^\prime \models\psi$$
				and thus, $\lambda\rho + (1-\lambda)\rho^\prime \models\phi\sd\psi$.
				\item $\mu\sepimp\psi$. Suppose $\rho,\rho^\prime\in\cD$ with same domain and $\rho\models \mu\sepimp\psi$ and $\rho^\prime\models \mu\sepimp\psi$, then by induction hypothesis and Proposition \ref{prop alt form} (4), we have for any $\sigma\in\cD(\free{\mu})$ such that $\sigma\models\mu$, and for any $\lambda\in[0,1]$,
				\begin{align*}
				&\sigma\otimes\rt{\rho}{\free{\psi}}\models\psi,\quad \sigma\otimes\rt{\rho^\prime}{\free{\psi}}\models\psi \\
				\Rightarrow\ & \lambda\sigma\otimes\rt{\rho}{\free{\psi}} + (1-\lambda)\sigma\otimes\rt{\rho^\prime}{\free{\psi}} \models\psi \\
				\Rightarrow\ & \sigma\otimes\rt{\left(\lambda\rho + (1-\lambda)\rho^\prime\right)}{\free{\psi}} \models\psi
				\end{align*}
				which implies $\lambda\rho + (1-\lambda)\rho^\prime \models\mu\sepimp\psi$.
				\item $\phi\in{\rm SP}$. If $\sem{\phi} = \emptyset$, then trivially $\phi\in{\rm CM}$. Otherwise, suppose $\sigma$ is the least element of $\sem{\phi}$, and $\rho,\rho^\prime\in\cD$ with same domain and $\rho\models \phi$ and $\rho^\prime\models \phi$, we must have: for any $\lambda\in[0,1]$,
				\begin{align*}
				\rt{\rho}{\free{\phi}} = \rt{\rho^\prime}{\free{\phi}} = \sigma\quad\Rightarrow\quad \rt{\left(\lambda\rho + (1-\lambda)\rho^\prime\right)}{\free{\phi}} = \sigma
				\end{align*}
				and so $\lambda\rho + (1-\lambda)\rho^\prime\models\phi$.
				\item $\mu_1\ast\phi$. Suppose $\sigma$ is the least element of $\sem{\phi}$, $\rho,\rho^\prime\in\cD$ with same domain and $\rho\models \mu_1\ast\phi$ and $\rho^\prime\models \mu_1\ast\phi$, then by induction hypothesis and \ref{prop alt form} (2), $\free{\mu_1}\cap\free{\psi} = \emptyset$, for any $\lambda\in[0,1]$, 
				\begin{align*}
				&\rt{\rho}{\free{\mu_1\ast\phi}} = \sigma \otimes \rt{\rho}{\free{\phi}},\quad \rt{\rho^\prime}{\free{\mu_1\ast\phi}} = \sigma \otimes \rt{\rho^\prime}{\free{\phi}},\quad \rt{\rho}{\free{\phi}},\rt{\rho^\prime}{\free{\phi}}\models\phi \\
				\Rightarrow\ &\rt{\left(\lambda\rho + (1-\lambda)\rho^\prime\right)}{\free{\mu_1\ast\phi}} = \sigma \otimes \left(\lambda\rt{\rho}{\free{\phi}} + (1-\lambda)\rt{\rho^\prime}{\free{\phi}} \right),\quad \lambda\rt{\rho}{\free{\phi}} + (1-\lambda)\rt{\rho^\prime}{\free{\phi}} \models\phi
				\end{align*}
				and thus, $\lambda\rho + (1-\lambda)\rho^\prime\models\mu_1\ast\phi$.
			\end{enumerate}
		\end{proof}
		
		\vspace{0.5cm}
		
		\noindent\textbf{Proof of Proposition \ref{prop SP 2BID}}
		
		\begin{proposition}
			The formulas generated by following grammar are ${\rm SP}$: 
			$$
			\phi,\psi ::= \unia[S]\ |\ p\in\cP {\rm\ of\ rank\ 1}\ |\ \top\ |\ \bot\ |\ \phi\ast\psi\ |\ \mu_1 \sepimp \phi\ |\ \mu_1 \sdimp \phi
			$$
			where $\cP \text{of rank 1}$ consists all rank 1 projections, and $\mu_1$ is formula with non-empty interpretation.
		\end{proposition}
		
		\begin{proof}
			\begin{enumerate}
				\item $\unia[S]$. Trivially, $\frac{I_S}{\dim(S)}$ is the least element of $\sem{\unia[S]}$.
				\item $P\in\cP \text{of rank 1}$. Trivially, $P$ itself (interpreted as a pure quantum state) is the least element of $\sem{P}$.
				\item $\top$. Scalar number $1$ is the least element of $\sem{\top}$.
				\item $\bot$. Trivial.
				\item $\phi\ast\psi$. Suppose $\sigma_\phi$ and $\sigma_\psi$ are the least elements of  $\sem{\phi}$ and  $\sem{\psi}$ respectively, then it is straightforward to show $\sigma_\phi\otimes \sigma_\psi$ is the least element of $\phi\ast\psi$.
				\item $\mu_1 \sepimp \phi$. If $\sem{\mu_1 \sepimp \phi}$ is nonempty, and since $\sem{\mu_1}$ is also nonempty, $\sem{\phi}$ must be nonempty, and suppose $\sigma$ is the least element of $\sem{\phi}$, then it is not difficult to realize that $\rt{\sigma}{\free{\phi}\backslash\free{\mu_1}}$ is the least element of $\sem{\mu_1 \sepimp \phi}$.
				\item $\mu_1 \sdimp \phi$. Similar to (6).
			\end{enumerate}
		\end{proof}
		
		\vspace{0.5cm}
		
		\noindent\textbf{Proof of Proposition \ref{prop modification qo 2BID}}
		
		\begin{proposition}
			\label{prop modification qo 2BID}
			\begin{enumerate}
				\item If $\phi[\cE[\qbar]]\Mexist$, $\free{\phi[\cE[\qbar]]} = \free {\phi}$;
				\item 
				If $\phi[\cE[\qbar]]\Mexist$, then for any state $\rho\in\cD(\free{\phi}\supseteq\qbar)$, $\cE(\rho)\models\phi$ if and only if $\rho\models \phi[\cE[\qbar]]$.
			\end{enumerate}
		\end{proposition}
		\begin{proof}
			\noindent	(1). Induction on the structure of $\phi$.
			
			\noindent	(2). We prove it by induction on the structure of $\phi$.
			\begin{enumerate}
				\item[(a)] $\phi\equiv\top$ or $\bot$. Trivial.
				\item[(b)] $\phi\equiv p\in\AP$, trivial by the Definition \ref{def qo modification 2BID} Clause 1.
				\item[(c)] $\phi\wedge\psi$. By induction hypothesis, for any state $\rho\in\cD(\free{\phi\wedge\psi}\cup\qbar)$, $\cE(\rho)\models\phi\wedge\psi$ iff $\cE(\rho)\models\phi$ and $\cE(\rho)\models\psi$ iff $\rho\models \phi[\cE[\qbar]]$ and $\rho\models \psi[\cE[\qbar]]$ iff $\rho\models \phi[\cE[\qbar]]\wedge\psi[\cE[\qbar]]$ iff $\rho\models (\phi\wedge\psi)[\cE[\qbar]]$.
				\item[(d)] $\phi\vee\psi$. Similar to (c).
				\item[(e)] $\phi\rightarrow\psi$. By Proposition \ref{prop alt form} (5) and induction hypothesis, for any state $\rho\in\cD(\free{\phi\wedge\psi}\cup\qbar)$, $\cE(\rho)\models\phi\rightarrow\psi$ iff $\cE(\rho)\models\phi$ implies $\cE(\rho)\models\psi$ iff $\rho\models \phi[\cE[\qbar]]$ implies $\rho\models \psi[\cE[\qbar]]$ iff $\rho\models \phi[\cE[\qbar]]\rightarrow\psi[\cE[\qbar]]$ iff $\rho\models (\phi\rightarrow\psi)[\cE[\qbar]]$.
				\item[(f)] $\phi\sd\psi$. Similar to (c) by using Proposition \ref{prop alt form} (1) and statement (1).
			\end{enumerate}
		\end{proof}

		\vspace{0.5cm}
		
		\noindent\textbf{Proof of Theorem \ref{thm sound QSL 2BID}}

		The global variable set is denoted by $\vars$, which contains all variables of programs and formulas. 
		
		It is sufficient to show that each of the rules shown in Figure \ref{fig proof system 4} is sound, the proof of other rules are the same as in Proof of Theorem \ref{thm sound QSL}.
		
		\noindent -- \textsc{Perm}. Also proved in Proof of Theorem \ref{thm sound QSL}.
		
		\vspace{0.4cm}
		
		\noindent -- \textsc{RLoop$^\prime$}. We here use the notations similar to \cite{Ying16}, Section 3.3. Set quantum operation (and its cylinder extension) $\cE_i(\cdot) = M_i(\cdot) M_i^\dag$ for $i = 0,1$. We first claim:
		$$\textbf{Statement:} \rho\models\phi\ast\id_{\qbar} \text{\ implies\ } \cE_0(\rho)\models\phi\ast M_0,\ \cE_1(\rho)\models\phi\ast M_1, \ \sem{\prog}\circ_{c}\cE_1(\rho)\models\phi\ast I_{\qbar}$$
		by the premises and $\circ_{c}$ denote the composition of quantum operations, i.e., $(\cE\circ_{c}\cF)(\rho) = \cE(\cF(\rho))$.
		Next, by induction and the statement, we have: for all $i\ge 0$:
		$$ \rho\models\phi\ast\id_{\qbar} \text{\ implies\ } \cE_0\circ_{c}(\sem{\prog}\circ_{c}\cE_1)^i(\rho)\models\phi\ast M_0.$$
		Finally, it has been proved that (see \cite{Ying11})
		$$\sem{{\bf while}}(\rho) = \sum_{i=0}^\infty \cE_0\circ_{c}(\sem{\prog}\circ_{c}\cE_1)^i(\rho) $$
		and thus if $\rho\models\phi\ast\id_{\qbar}$, then $\sem{{\bf while}}(\rho)\models\phi$ and $\sem{{\bf while}}(\rho)\models M_0$ since $\phi,M_0\in{\rm CM}$. And note that $\free{\phi}\cap\free{M_0} = \emptyset$, so  $\sem{{\bf while}}(\rho)\models \phi\sd M_0$.	
		
		\vspace{0.4cm}
		
		\noindent -- \textsc{Weak}. By premise $\models (\phi\rightarrow\phi^\prime)\wedge(\psi^{\prime}\rightarrow\psi)$, we know that for any input $\rho\in\cD(\vars)$ that satisfies $\phi$, it must also satisfy $\phi^\prime$. By another premise $\{\phi^{\prime}\}\prog\{\psi^{\prime}\}$, then $\sem{\prog}(\rho)\models\psi^{\prime}$, and thus $\sem{\prog}(\rho)\models\psi$.
		
		\vspace{0.4cm}
		
		\noindent -- \textsc{FrameE}. For any input $\rho\in\cD(\vars)$ such that $\rho\models\phi\sd\mu$, we must have $\rho\models\phi\wedge\mu$ and $\free{\phi}\cap\free{\mu} = \emptyset$. Similar to \textsc{Const}, we have $\sem{\prog}(\rho)\models\psi\wedge\mu$ by first two premises. Moreover, notice that $\free{\psi}\cap\free{\mu} \subseteq (\free{\phi}\cup\var(\prog))\cap\free{\mu} = \emptyset$, thus by Proposition \ref{prop alt form}, $\sem{\prog}(\rho)\models\psi\sd\mu$.

	\end{appendices}

\end{document}